\documentclass[10pt,journal,compsoc]{IEEEtran}

\ifCLASSOPTIONcompsoc
  \usepackage[nocompress]{cite}
\else
  \usepackage{cite}
\fi
\usepackage{pgfplots}
\usepackage{subfig}
\usepackage{graphicx}
\usepackage{epstopdf}

\ifCLASSINFOpdf
\else
\fi
\usepackage{pgfplots}
\usepackage{amsmath}
\usepackage[utf8]{inputenc}
\usepackage[english]{babel}
\usepackage{graphicx}
\usepackage{tikz}
\usepackage{array}
\usepackage{booktabs}
\usepackage{caption}
\usepackage{tabu}
\usetikzlibrary{shapes,arrows,shadows}
\usepackage{epstopdf}
\usepackage{amsthm}

\newtheorem{lemma}{Lemma}

\usepackage{algorithm}
\usepackage{algpseudocode}
\usepackage{pifont}
\usepackage{subfig}
\usepackage{amssymb}
\theoremstyle{definition}

\makeatletter
\def\BState{\State\hskip-\ALG@thistlm}
\makeatother
\algrenewcommand\algorithmicforall{\textbf{forall}}
\addto\captionsenglish{
\renewcommand{\figurename}{Fig.}
\renewcommand{\tablename}{TABLE}

}
\begin{document}
%
\title{Simplified Biased Contribution Index (SBCI): A Mechanism to Make P2P Network Fair and Efficient for Resource Sharing}
%
%
%
%

\author{Sateesh~Kumar~Awasthi,
        Yatindra Nath Singh,~\IEEEmembership{Senior~Member,~IEEE,}
\IEEEcompsocitemizethanks{\IEEEcompsocthanksitem Sateesh Kumar Awasthi  is with the Department of Electrical Engineering, Indian Institute of Technology, Kanpur, India, e-mail: (sateesh@iitk.ac.in)
\IEEEcompsocthanksitem Yatindra~Nath~Singh  is with the Department of Electrical Engineering, Indian Institute of Technology, Kanpur, India, e-mail: (ynsingh@iitk.ac.in)}}
\IEEEtitleabstractindextext{%
\begin{abstract}
To balance the load and to discourage the free-riding in peer-to-peer (P2P) networks, many incentive mechanisms and policies have been proposed in recent years. Global peer ranking is one such mechanism. In this mechanism, peers are ranked based on a metric called contribution index. Contribution index is defined in such a manner that peers are motivated to share the  resources in the network. Fairness in the terms of upload to download ratio in each peer can be achieved by this method. However, calculation of contribution index is  not trivial. It is computed distributively and iteratively  in the entire network and requires strict clock synchronization among the peers. A very small error in clock synchronization may lead to wrong results. Furthermore, iterative calculation requires a lot of message overhead and storage capacity, which makes its implementation more complex. In this paper, we are proposing a simple incentive mechanism based on the contributions of peers, which can balance the upload and download amount of resources in each peer. It does not require iterative calculation, therefore, can be implemented with lesser message overhead and storage capacity without requiring strict clock synchronization. This approach is efficient as there are very less rejections among the cooperative peers. It can be implemented in a truly distributed fashion with $O(N)$ time complexity per peer.
                                                                                                                                                                                                                                                                                                                                                                                                                                                                                                                                                                                                                                                                                                                                                                                                                                                                                                                                                                                                                                                                                                                                                                                                                                                                                                                                                                                                                                                                                                                                                                                                                                                                                                                                                                                                                                                                                                                                                                                                                                                                                                                                                                                                                                                                                                                                                                                                                                                                                                                                                                                                                                                                                                                                                                                                                                                                                                                                                                                                                                                                                                                  
\end{abstract}

}

\maketitle

\IEEEdisplaynontitleabstractindextext

%
\IEEEpeerreviewmaketitle

\ifCLASSOPTIONcompsoc
\IEEEraisesectionheading{\section{Introduction}\label{sec:introduction}}
\else
\section{Introduction}
\label{sec:introduction}
\fi

%
%
%
%
\IEEEPARstart{P}{eer-to-peer (P2P)} networks gained a significant popularity in the last decade and now responsible for a large fraction of internet traffic \cite{survey2_1}, \cite{survey2_2}. The popularity of these networks is due to their inherent advantages over traditional client-server model, e.g., the diversity of available data, scalability, robustness  and cost effectiveness. The initial setup cost for these networks is very small because costly central servers are not needed. However, lack of central control leads to the problem of unfairness in these networks, i.e., large difference between upload and download amount at any peer. In such a situation, many peers free-ride and contribute very less or nothing which results in slow downloads for other peers \cite{freeride1}. Therefore, designing and implementing an efficient incentive policy to motivate the peers to share the resources  becomes important.\par In recent years, many incentive policies have  been proposed to maintain the fairness in P2P networks \cite{global}, \cite{bci}, \cite{robust}, \cite{tit}, \cite{mtit}, \cite{give}. In these policies, peers' cooperative behavior in the network is evaluated and resources are given to them in proportion to their cooperation.\par In \cite{tit}, \cite{mtit}, \cite{give}, peers' cooperation is evaluated locally, i.e., peer cooperate with only those peers who had cooperated with them in the past. To start the process of sharing, a small amount of data is given to every peer. In such scenario, free-riders can always find a new peer to download their desired data. Also, the cooperative peers are not allowed  to download more than this small amount of data from a new peer even though they have uploaded the large amount of data to some other peers\cite{global}.\par  In \cite{global}, \cite{bci}, \cite{robust}, peers' cooperative behavior in the entire network is taken into consideration. For this purpose, in \cite{robust}, every peer keeps the record of each transaction which has happened in the entire network. It makes the implementation of algorithm very complex. In comparison to this,  \cite{global}, \cite{bci} are simpler approaches. In these approaches, peers are ranked in the entire network. The rank of the peer is determined by the contribution index. It is estimated using two factors, resources contributed by the peer in the network and contribution index of peer with whom it is transacting. Estimation of contribution index is performed by iterative methods and can be implemented in a distributed fashion. These approaches are able to balance the amount of upload and download of resources in the network. However, there are some fundamental problems in its implementation.\par First, in each iteration, index managers, i.e., peers who are managing the contribution index of other peers, need the current contribution index of peers from other peers. If clocks of the peers are not synchronized, then the peers who are reporting the contribution index of peers may report the contribution index of the previous iteration,  which may lead to the wrong estimates  \cite{sync}. \par Second, updating the contribution index in each iteration requires a lot of message overhead. This is more important when the number of iterations required to converge the algorithm is large. If new transactions happen in the network, then contribution index need to be updated. Even  one transaction, between any two peers, can affect the contribution index of all the peers in the network.\par And lastly, index managers need to keep the record of all past transactions of a peer for whom they are estimating the  contribution index. This needs a large amount of storage capacity. Keeping all these points in view, a simple incentive policy is required, which can ensure the following:
\begin{itemize}
\item It should balance the upload and download amount of resources at each peer.
\item There must be minimum rejections among the cooperative peers. 
\item Cooperation of peers must be considered in the entire network.
\item Lower message overhead and storage capacity is desirable.
\item It should be robust to peer dynamics.
\item It should be implementable in truly distributed system.
\end{itemize}
In this paper, we are proposing an incentive policy, which considers  peers' cooperation in the entire network. We are assigning the contribution index to each peer. It is a simplified form of the Biased Contribution Index (BCI) \cite{bci}. We call it Simplified Biased Contribution Index (SBCI). It also depends on the cooperation of peers in sharing the resources and in  balancing  the load in the network. SBCI is updated at regular time intervals. At any time, SBCI is calculated  using previous SBCI and the cooperation made by the peers during this period, i.e., in between previous update to current update. In the estimation of SBCI, no iterative calculation is required, hence it automatically solves the first and second  problems. Once the peers'  cooperation is modeled in terms of SBCI, it need not store the history of peers' transactions, hence, it also solves the last problem. Our simulation results show that SBCI can balance the upload and download amount at each peer with minimum rejections among cooperative peers. Hence it meets all the above design considerations.\par Rest of the paper is organized as follows. Section \ref{rwork} covers the summary of related work. The proposed incentive model is introduced in Section \ref{proposed}. Section \ref{analysis} covers the analysis of algorithm. The transaction procedure for maximum efficiency is introduced in Section \ref{procedure}. Evaluation of algorithm, through simulation is discussed in Section \ref{experiment}. Finally, paper is concluded in  Section \ref{conclusion}.
\section{Related Work}\label{rwork}
Presence of free-riding peers and its impact on fairness in P2P network have been studied earlier also \cite{freeride1}, \cite{gnutela2002}. Several approaches have been proposed by the research community \cite{global}, \cite{bci}, \cite{robust}, \cite{tit}, \cite{mtit}, \cite{give}, \cite{whitewashing}, \cite{ccom2007}, \cite{tom9},\cite{online}, \cite{ton12}, \cite{jstsp2010}, \cite{evalgame2010}, \cite{tongame2006}.
\par BitTorrent \cite{torrent}, a most popular file sharing system, used tit-for-tat (TFT) approach to prevent the free-riding. In this approach, a  peer cooperates with other peers in the same proportion as they have cooperated with him in the previous round. In each round, every peer updates the contributions of peers in the previous round. To improve the performance, many variants of TFT have been proposed. Garbacki \textit{et al.}, \cite{tit}, proposed ATFT  in which bandwidth is used rather than content to decide the incentives. Dave \textit{et al.},\cite{auction}, proposed auction based model  to improve the TFT. In this model, peers reward one another with \textit{proportional shares}, \cite{prop_share}, of bandwidth. Sherman \textit{et al.}, \cite{ton12}, proposed FairTorrent. It is a deficit based distributed algorithm in which a peer uploads the next data block to the peer, whom it owes the most data as measured by a deficit counter. In Give-to-get \cite{give}, peer ranks all its neighbors,  based on the amount of data what have been received from them in the last round and then unchokes the top three forwarders. All these mechanisms consider the local and very short history of peers' cooperation. \par 
Global history of peers' cooperation is considered in \cite{global}, \cite{bci}, \cite{robust}, \cite{mtit}. In multilevel tit-for-tat (ML-TFT) \cite{mtit}, a peer ranks other peers based on the fraction of download, what he received from them. Its time complexity is much larger for n-step ranking of peers. Feldman \textit{et al.}, \cite{robust}, proposed a robust incentive technique, which considers the peers' cooperation in the entire network, but it is not trivial to implement in a large network. Its calculation have complexity of $O(N^3)$. In Global Contribution (GC) approach \cite{global}, a peers' GC point is defined such that all peers are motivated to download from low contributing peers and upload to high contributing peers. GC point is calculated using iterative methods such as the Jacobi and Gauss-Seidel. In another similar approach, Biased Contribution Index(BCI) \cite{bci}, second order iterative function is used to calculate the BCI of peers. BCI is defined as monotonically increasing function of biased upload to download ratio. Convergence of BCI \cite{bci} is faster than GC \cite{global}.
\par Many authors proposed  approaches based on game theory  \cite{online},  \cite{jstsp2010}, \cite{evalgame2010}, \cite{tongame2006}. Free-riding can be reduced by this approach. This approach is based on the assumption that the rules of the game are known to all the players. For practically large network,  this may not be true for all the peers.
\par Reputation management system, \cite{abs}, \cite{eigen}, \cite{sat}, \cite{sort}, is  another approach in which peers' behavior is modeled as  trust. Trust is estimated by each peer based on its interaction with the other peers and then it is aggregated in the whole network. Trust is a more generalized term and depends on the overall behavior of peer in the network. In the proposed SBCI, we are focusing on the particular issues of fairness and free-riding.

\section{Proposed Incentive model}\label{proposed}
 \subsection{Design Rules to Ensure the Fair and Efficient P2P Network}
 Let us make some design rules to ensure the design considerations mentioned in Section \ref{sec:introduction}.\\
 1). If any peer only downloads the resources from the network  then its SBCI must be zero.\\
 2). If it only uploads to the network ( at least once to other than free-rider) then its SBCI must be 1.\\
 3). Uploading to the free-riders should not increase the SBCI.\\
 4). Uploading to any other peer should always increase the SBCI.\\
 5). Download should always decrease the SBCI.\\
 6). Peers must be motivated to upload to high contributing peers.\\
 7). Peers must be motivated to download from low contributing peers. 
 \subsection{Simplified Biased Contribution Index}
 Let there be $N$ peers in a P2P network. Further, we considered time evolution in discrete instances. A time instance is represented by $t_n$, and if an event happened in the time interval, $(t_{n-1}, t_n]$, it is considered to happen at $t_n$. At any time, $t_n$, let the share matrix in the entire network be $\mathbf{S(t_n)} $. Where its $ij$ element is the amount of resource shared by peer $i$ to peer $j$ at time $t_n$, i.e., in $(t_{n-1}, t_n]$. The bias ratio, $R_i(t_n)$, for peer $i$ at time $t_n$ can be defined in the  similar way as in \cite{bci}.
 
 \begin{equation}\label{eq4.1}
 R_i(t_n)=\mathbf{\frac{e_i.S(t_n).x(t_n)}{e_i.S^{tr}(t_n).x(t_n)}}
 \end{equation}
 Here, $\mathbf{ x(t_n)}$ is the SBCI vector of peers at time $t_n$. $\mathbf{S^{tr}(t_n)}$ is transpose of matrix $\mathbf{S(t_n)}$ and $\mathbf{e_i}$ is a row vector with its $i^{th}$ entry as 1 and all others as zero. Now, let us define the SBCI, $x_i(t_n)$, of peer $i$ as a monotonically increasing function of the bias ratio at time  $t_{n-1}$.
 \begin{equation}\label{eq4.2}
 \begin{split}
 x_i(t_n) & =\frac{R_i(t_{n-1})}{1 + R_i(t_{n-1})}\\
          & =\mathbf{\frac{e_i.S(t_{n-1}).x(t_{n-1})}{e_i.S(t_{n-1}).x(t_{n-1})+e_i.S^{tr} (t_{n-1}).x(t_{n-1})}}
          \end{split}
 \end{equation}
If any peer $i$ does not upload anything in the network at time $t_{n-1}$, then $\mathbf{e_i.S(t_{n-1}).x(t_{n-1}})=0$. But if it  download something from the network at this time, then $\mathbf{e_i.S^{tr}(t_{n-1}).x(t_{n-1}}) \neq 0$ only if $\mathbf{x(t_{n-1})} \neq 0$.  Therefore, to make the denominator in (\ref{eq4.2}) nonzero for zero uploading and nonzero downloading, let us replace $\mathbf{ e_i.S^{tr} (t_{n-1}).x(t_{n-1})}$ by  $\alpha \mathbf{ e_i.S^{tr} (t_{n-1}).x(t_{n-1})} +(1-\alpha)\mathbf{ e_i.S^{tr} (t_{n-1}).e}$. Here, $\alpha \in (0, 1)$ is constant and $\mathbf{e}$ is a column  vector with each element as 1. Hence, (\ref{eq4.2}) will be:
\begin{equation}\label{eq4.4}
\begin{split}
x_i(t_n) &=[\mathbf{e_i.S(t_{n-1}).x(t_{n-1})}]/[\mathbf{e_i.S(t_{n-1}).x(t_{n-1})}+\\ & \alpha \mathbf{ e_i.S^{tr} (t_{n-1}).x(t_{n-1})} + (1-\alpha)\mathbf{ e_i.S^{tr} (t_{n-1}).e}].
\end{split}
\end{equation}
SBCI in the above equation is estimated using the transactions, which are happening only at time $t_{n-1}$. If  we consider all the past transactions, then SBCI can be modified as:
\begin{equation}\label{eq4.5}
\begin{split}
x_i(t_n) & = (1-\beta_i(t_{n-1}))x_i(t_{n-1}) + \\ & \beta_i(t_{n-1})[\mathbf{e_i.S(t_{n-1}).x(t_{n-1})}]/[\mathbf{e_i.S(t_{n-1}).x(t_{n-1})}+ \\ & \alpha \mathbf{ e_i.S^{tr} (t_{n-1}).x(t_{n-1})} + (1-\alpha)\mathbf{ e_i.S^{tr} (t_{n-1}).e}].   \end{split} 
\end{equation}
If peer $i$ does not  participate in any transaction at time $t_{n-1}$, then $x_i(t_n)$ should be $x_i(t_{n-1})$. Parameter $\beta_i(t_{n-1})$ can be decided by the fraction of transaction, which are happening at time, $t_{n-1}$, at node $i$, and can be  defined as:
\begin{equation}\label{eq4.6}
\beta_i(t_{n-1})= \begin{cases}
0, &\text{if $A_{u_i}=0$}.\\
\mathbf{\frac{e_i.[S(t_{n-1})+S^{tr}(t_{n-1})].e}{e_i.[S_{comp}(t_{n-1})+{S^{tr}}_{comp}(t_{n-1})].e}}, &{otherwise}
\end{cases}
\end{equation}
Here, $A_{u_i}=\mathbf{e_i.S(t_{n-1}).x(t_{n-1}) + e_i.S^{tr}(t_{n-1}).e}$. The $\mathbf{S_{comp}(t_{n-1})}$ is a complete share matrix with its $ij$ element as the amount of resources shared by peer $i$ to peer $j$, till time $t_{n-1}$. To start the process of sharing, the SBCI vector can be initialized as, $\mathbf{x(0)}=\alpha/(1+\alpha)\mathbf{e}$, later we will see that this choice of initialization will balance the upload and download amounts in the network.

\subsection{Justification For Design Rules}
If any peer $i$, does not upload anything and only download the resources from the network at time $t_{n-1}$, then $\mathbf{e_i.S(t_{n-1}).x(t_{n-1})}=0$ and $\mathbf{e_i.S^{tr}(t_{n-1}).e} \neq 0$, hence, from  (\ref{eq4.5}),   \[x_i(t_n)=(1-\beta_i(t_{n-1}))x_i(t_{n-1})\] Let it did not upload anything in the network till time $t_n$, and  started downloading the resource first time at time $t_m$, then from (\ref{eq4.6}), $\beta_i(t_m)=1$, hence \[x_i(t_n)=(1-\beta_i(t_{n-1}))(1-\beta_i(t_{n-2}))...(1-\beta_i(t_m))x_i(t_m)=0. \] Therefore, \textbf{if any peer $i$, only downloads from the network then its SBCI  will be zero.}\par 

At time  $t_{n-1}$, if any peer $i$ uploads only to the free-riders, i.e., peers who only download without uploading anything in the network, then $\mathbf{e_i.S(t_{n-1}).x(t_{n-1})}=0$, if it does not download anything at this time, $t_{n-1}$, then $\mathbf{e_i.S^{tr}(t_{n-1}).e}=0$. Therefore, $A_{u_i}=0$, and hence from  (\ref{eq4.6}), $\beta_i(t_{n-1})=0$, and from (\ref{eq4.5}) \[x_i(t_n)=x_i(t_{n-1})\]Therefore, \textbf{uploading to the free-riders will not increase the SBCI.}\par 

At time $t_{m-1}$, if any peer $i$, only uploads the resources in the network (at least one of the downloader should be other than free-rider) and does not download anything from it  then, $\mathbf{e_i.S(t_{m-1}).x(t_{m-1})} \neq 0$ and $\alpha \mathbf{e_i.S^{tr}(t_{m-1}).x(t_{m-1})} +(1-\alpha)\mathbf{e_i.S^{tr}(t_{m-1}).e}=0$. Hence from  (\ref{eq4.5}), \[x_i(t_m)=(1-\beta_i(t_{m-1})) x_i(t_{m-1}) + \beta_i(t_{m-1})\] Let it is first time when the peer $i$ makes any transaction  in the network, then from (\ref{eq4.6}), $\beta_i(t_{m-1})=1$. Hence, 
\[x_i(t_{m})= 0.x_i(t_m{-1}) + \beta_i(t_{m-1})=\beta_i(t_{m-1})=1\] Now, at time $t_m$, if it does not participate in any transaction then \[x_i(t_{m+1})=x_i(t_{m})=1\] if it uploads to only free-riders and does not download anything  then again \[x_i(t_{m+1})=x_i(t_{m})=1\] if it uploads to at least one of the peer other than free-rider without downloading anything then,
\[x_i(t_{m+1})=(1-\beta_i(t_{m})) x_i(t_{m}) + \beta_i(t_{m})\]
\[= (1-\beta_i(t_{m})) 1 + \beta_i(t_{m})=1.\]
Hence, from mathematical induction, we can say that this is true for any $n$ thus, $x_i(t_{n})=1$\\ Therefore, \textbf{if any peer $i$, only uploads to the network (at least once to other than free-rider) then its SBCI will be 1.}\par 
If any peer $i$, uploads the resources to  non-free-rider peer, at  time $t_{n-1}$, then $\mathbf{e_i.S(t_{n-1}).x(t_{n-1})} \neq 0$. If it does not download anything at this  time then, $\alpha\mathbf{e_i.S^{tr}(t_{n-1}).x(t_{n-1})} + (1-\alpha)\mathbf{e_i.S^{tr}(t_{n-1}).e}=0$. Hence, from  (\ref{eq4.5}),\[x_i(t_{n})=(1-\beta_i(t_{n-1})) x_i(t_{n-1}) + \beta_i(t_{n-1})\] It is a convex combination of $1$ and $x_i(t_{n-1})$ hence, 
\[x_i(t_{n-1}) < x_i(t_n) < 1 \hspace{6mm}\forall \beta_i(t_{n-1}) \in (0,1)\] Therefore, \textbf{uploading to the peer other than free-rider will always increase the SBCI.}\par 
If any peer $i$, downloads the resource from the network at time  $t_{n-1}$, then $\mathbf{e_i.S^{tr}(t_{n-1}).e} \neq 0$, hence $A_u \neq 0$, therefore, from  (\ref{eq4.6}), $\beta_i(t_{n-1}) > 0$. If it does not upload anything in the network at this time, then $\mathbf{e_i.S(t_{n-1}).x(t_{n-1})} = 0$, hence from  (\ref{eq4.5}) \[x_i(t_{n})=(1-\beta_i(t_{n-1})) x_i(t_{n-1}) + \beta_i(t_{n-1}).0 \]
\[=(1-\beta_i(t_{n-1})) x_i(t_{n-1})\] hence,
\[x_i(t_{n}) <  x_i(t_{n-1})\]Therefore, \textbf{ download will always decrease the SBCI.}\par 
It can be concluded from the above discussion that high contributions will lead to high SBCI. Now, observing directly  the (\ref{eq4.5}), if peers will upload the resources to  high SBCI peers then, they will earn more SBCI. Therefore, \textbf{peers will be motivated to upload the resources to high contributing peers.}\par 
It can also be observed from  (\ref{eq4.5}) that peers will lose less SBCI, if they will download from a low SBCI peer. Therefore, \textbf{peers will be motivated to download from low contributing peers.}
\begin{figure}
\begin{center}
\begin{tikzpicture}[->,>=stealth',shorten >=1pt,auto,node distance=3cm, thick,main node/.style={circle,draw,font=\sffamily\Large\bfseries}]

  \node[main node] (1) {1};
  \node[main node] (2) [right of=1] {2};
  \node[main node] (3) [below right of=2] {3};
  \node[main node] (4) [below left of=3] {4};
  \node[main node] (5) [left of=4] {5};

  \path[every node/.style={font=\sffamily\small}]
    (1) edge node  {100} (2)
        edge node[left] {200} (3)
        
    (2) edge [bend left] node [left] {100} (5)
        
    (3) edge [bend right] node[right] {100} (2)
        edge [bend left] node[left] {200} (4)
    (4) edge node [right] {100} (1)
    
    (5) edge node  {100} (4)
        edge [bend left] node [left] {200}  (1) ;
\end{tikzpicture}
\caption{Upload and download at each peer at time  $t=0$}\label{fig4.1}
\end{center}
\end{figure}
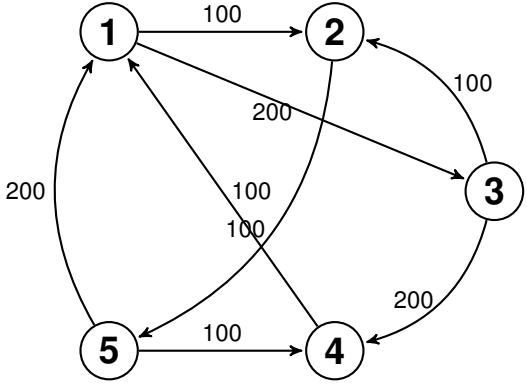
\par Let us understand the SBCI and its computation through an example. Let there be five peers A, B, C, D and E in a P2P network as shown in Fig. \ref{fig4.1}. If $\alpha= 0.9$, then initial SBCI of all the peers will be $\alpha/(1+\alpha)=0.4737$. At time $t=0$, let they share the resources as shown in figure, i.e., $S_{12}(0)=100, S_{13}(0)=200, S_{25}(0)=100, S_{32}(0)=100, S_{34}(0)=200, S_{41}(0)=100, S_{51}(0)=200, S_{54}(0)=100$ and all others are zero. Since, it is initial step, hence, for all $i$, $\beta_i(0)=1$. Using  (\ref{eq4.5}), SBCI vector at time $t=1$, can be calculated as, $\mathbf{x(1)}=[0.4737, 0.3103, 0.5745, 0.2308, 0.7297]^t$. \par Now, let peer $1$ needs the data amount of 100 units and all the four peers responded to his query, then peer 1 will select the peer with least SBCI as an uploader, in this case, peer 4 has least SBCI. After this transaction, let SBCI vector is updated at $t=2$. For $t=1$,  $S_{41}(1)=100$ and all others are zero. Hence for this time, $\beta_1(1)=1/7,\beta_2(1)=\beta_3(1)=\beta_5(1)=0$ and $\beta_4(1)=1/5$. Hence, updated SBCI  vector will be, $\mathbf{x(2)}=[0.4060, 0.3103, 0.5745, 0.3846, 0.7297]^t$. 
\subsection{Justification For Fairness}
\begin{lemma}\label{lemma1}
At any time  $t_{n-1}$, if upload and download at each peer is same and SBCI vector, $\mathbf{x(t_{n-1})}=\alpha/(1+\alpha)\mathbf{e}$,  then  SBCI vector, $\mathbf{x(t_n)}=\mathbf{x(t_{n-1})}$.
\end{lemma}
\begin{proof}
Let upload and download for any peer $i$ at time $t_{n-1}$ be $T_i(t_{n-1})$, then
\[\mathbf{e_i.S_i(t_{n-1}).e=e_i.S^{tr}_i(t_{n-1}).e}=T_i(t_{n-1}).\] Since, $\mathbf{ x(t_{n-1})}=\alpha/(1+\alpha)\mathbf{e}= a\mathbf{e}$, here, $a=\alpha/(1+\alpha)$, hence from (\ref{eq4.5}),
\begin{multline*}\
x_i(t_n) =(1-\beta_i(t_{n-1}))a + \\  \beta_i(t_{n-1})[a\mathbf{e_i.S(t_{n-1}).e}]/[a\mathbf{e_i.S(t_{n-1}).e}+\\ \alpha a\mathbf{ e_i.S^{tr} (t_{n-1}).e} + (1-\alpha)\mathbf{ e_i.S^{tr} (t_{n-1}).e}]
\end{multline*}
\begin{multline*}\
\hspace{5mm} =(1-\beta_i(t_{n-1}))a + \\  \beta_i(t_{n-1})[aT_i(t_{n-1})]/[T_i(t_{n-1})(a+\alpha a + (1-\alpha)]
\end{multline*}
\[\hspace{5mm} =(1-\beta_i(t_{n-1}))a + \beta_i(t_{n-1})a/(a(1+\alpha)  + (1-\alpha))\] Put $a=\alpha/(1+\alpha)$, hence
\[x_i(t_n) =(1-\beta_i(t_{n-1}))a + \beta_i(t_{n-1})a=a \hspace{3mm} \forall i\]
\end{proof}

\begin{lemma}\label{lemma2}
If SBCI vector at any two successive time instances, $t_{n-1}$ and $t_n$, is same and lie on vector $\mathbf{e}$, then upload and download at  time $t_{n-1}$ will be same in each peer.
\end{lemma} 
\begin{proof}
Let $\mathbf{x(t_n)=x(t_{n-1})}=a\mathbf{e}$, where $a$ is any constant, then from  (\ref{eq4.5})
\begin{multline*}\
a =(1-\beta_i(t_{n-1}))a + \\  \beta_i(t_{n-1})[a\mathbf{e_i.S(t_{n-1}).e}]/[a\mathbf{e_i.S(t_{n-1}).e}+\\ \alpha a\mathbf{ e_i.S^{tr} (t_{n-1}).e} + (1-\alpha)\mathbf{ e_i.S^{tr} (t_{n-1}).e}].
\end{multline*}
Manipulating  above, we get
\begin{multline*}\
a\beta_i(t_{n-1}) =  a\beta_i(t_{n-1})[\mathbf{e_i.S(t_{n-1}).e}]/[a\mathbf{e_i.S(t_{n-1}).e}+\\ \alpha a\mathbf{ e_i.S^{tr} (t_{n-1}).e} + (1-\alpha)\mathbf{ e_i.S^{tr} (t_{n-1}).e}].
\end{multline*}
For nonzero $a\beta_i(t_{n-1})$,
\[a\mathbf{e_i.S(t_{n-1}).e}+(a\alpha+1-\alpha)\mathbf{ e_i.S^{tr} (t_{n-1}).e}=\mathbf{e_i.S(t_{n-1}).e}.\] Solving,
\begin{equation}\label{eq4.7}
(a\alpha+1-\alpha)\mathbf{ e_i.S^{tr} (t_{n-1}).e}=(1-a)\mathbf{e_i.S(t_{n-1}).e} \hspace{3mm}\forall i
\end{equation}
 Since, $i=1,2,...,N$, hence this set of $N$ equations can be written in the form of matrix as follows,
\[(a\alpha+1-\alpha)\mathbf{S^{tr} (t_{n-1}).e}=(1-a)\mathbf{S(t_{n-1}).e}\] Pre-multiplying by $\mathbf{e^{tr}}$ on both sides,
\[(a\alpha+1-\alpha)\mathbf{e^{tr}.S^{tr} (t_{n-1}).e}=(1-a)\mathbf{e^{tr}.S(t_{n-1}).e}\] for any matrix, $\mathbf{S(t_{n-1})}$, $\mathbf{e^{tr}.S(t_{n-1}).e}$ will be the sum of all of its elements, hence $\mathbf{e^{tr}.S(t_{n-1}).e}=\mathbf{e^{tr}.S^{tr}(t_{n-1}).e}=T$ hence,
\[(a\alpha+1-\alpha)T=(1-a)T\] since $T \neq 0$ hence, 
\[a=\frac{\alpha}{1+\alpha}\] Substituting the value of $a$ in (\ref{eq4.7})
\[(\frac{\alpha^2}{1+\alpha}+1-\alpha)\mathbf{ e_i.S^{tr} (t_{n-1}).e}=(1-\frac{\alpha}{1+\alpha})\mathbf{e_i.S(t_{n-1}).e}\] or
\[(1+\alpha)\mathbf{ e_i.S^{tr} (t_{n-1}).e}=(1+\alpha)\mathbf{e_i.S(t_{n-1}).e}\] since $\alpha \in (0,1)$, hence
\[\mathbf{ e_i.S^{tr} (t_{n-1}).e}=\mathbf{e_i.S(t_{n-1}).e} \hspace{4mm} \forall i\] Hence, upload and download at time $t_{n-1}$ will be same in each peer $i$.

\end{proof}

\section{Analysis of Algorithm}\label{analysis}
\subsection{Implementation in Distributed System}
SBCI of each peer can be calculated distributively as shown in  Algorithm \ref{algo4.1}. Each peer's SBCI can be calculated  and  managed by some other peer in the network. We call it index manager and the peer whose SBCI is being calculated by this peer is called its daughter peer. The index manager peer can be located using distributed hash table (DHT) such as Chord \cite{chord}, CAN \cite{can}, Pastry \cite{pastry} and Tapestry \cite{tapestry}. Each peer $i$ will send the values of resources uploaded and downloaded to and from other peer $j$ to the index manager of peer $j$. An index manager peer  will collect the values of resources uploaded and downloaded by its daughter peer $k$, to other peers. Each index manager will locate the index manager of peer $j$ and will receive the current SBCI,  $x_j(t_{n-1})$, of peer $j$.\par Now each index manager possesses  all the things to calculate  the SBCI of its daughter peer using (\ref{eq4.5}). The $\beta_k(t_{n-1})$ can be calculated using (\ref{eq4.6}). If $A_u=0$ then it is zero, otherwise it is just a ratio of the current transaction amount to the total transaction amount made by peer $k$, till time $t_{n-1}$. The total amount of transactions can be updated by adding the current amount of transaction with the previous total amount of transactions.  
\subsection{Message Overhead, Storage Capacity and Time Complexity}
In this method, the SBCI is calculated directly while in other similar approaches \cite{global}, \cite{bci} iterative calculations are required. Therefore, the total number of messages required to calculate the SBCI, in this method will be $I_1$ and $I_2$ times lesser than \cite{global} and \cite{bci} respectively. Where $I_1$ and $I_2$  is the number of iterations  required to converge the algorithm in \cite{global} and \cite{bci} respectively.\par In this algorithm, index manager needs to store only two information about its daughter peer, i.e., current SBCI and total amount of transaction till $t_{n-1}$. While in \cite{global}, \cite{bci}, all transaction history of its daughter peer, i.e., amount of transaction, ID of peer with whom it transacted and whether it was upload or download, are required to be stored. Therefore, the required amount of storage is reduced very much. \par Time complexity of algorithm for one update can be calculated directly from (\ref{eq4.6}). It will be $O(N)$ per peer which  is same as in  \cite{global} and \cite{bci}.

\begin{algorithm}
\caption{For Updating the SBCI of Peers}\label{algo4.1}
\begin{algorithmic}[1]
\State \textbf{Input:} Amount of upload and download of peers
\State \textbf{Output:} SBCI with index managers
\Procedure{}{}
\For { each peer $i$ }
\ForAll {peer $j$, who is selected as source peer}
\State Download the resource 
\State Send the value of resource to the index manager of peer $j$;
\EndFor \textbf{all}
\ForAll {peer $j$, who selected peer $i$ as source peer}
\State Upload the resource 
\State Send the value of resource to the index manager of peer $j$;
\EndFor \textbf{all}
\If { Peer $i$ is index manager of peer $k$}
\ForAll{ peer $j$, who  transacted with peer $k$ }
\State Receive the value of resource uploaded $S_{kj}(t_{n-1})$; 
\State Receive the value of resource downloaded $S_{jk}(t_{n-1})$;
\If{$t=0$}
\State $\backslash\backslash$ Initialization of parameters
\State Set $\mathbf{x(0)}=(\alpha/(1+\alpha))\mathbf{e}$;
\Else
\State Locate $j^{th}$ peer's index manager; 
\State Receive the current SBCI, $x_j(t_{n-1})$ of peer $j$ from these index managers ;
\EndIf
\EndFor \textbf{all}
\State $\backslash\backslash$ Initialization of the amount of total transactions
\State Set $Ttr_k(0)=0$
\State Compute
\State $A_{u_k}=\mathbf{e_k.S(t_{n-1}).x(t_{n-1}) + e_k.S^{tr}(t_{n-1}).e}$
\If{$A_{u_k}=0$}
\State $\beta_k(t_{n-1}) =0$
\Else 
\State $\delta Ttr_k(t_{n-1})={\mathbf{ e_k.(S(t_{n-1})+S^{tr}(t_{n-1})).e}}$
\State $Ttr_k(t_{n-1})=Ttr_k(t_{n-2})+   \delta Ttr_k(t_{n-1})$
\State $\beta_k(t_{n-1}) =\frac{\delta Ttr_k(t_{n-1})}{Ttr_k(t_{n-1})}$
\EndIf
\State  Compute $x_k(t_n)$ using  (\ref{eq4.5})
\State  Save $x_k(t_n)$
\State  Save $Ttr_k(t_{n-1})$
\EndIf
\EndFor
\EndProcedure
\end{algorithmic}
\end{algorithm}
\section{Transaction Procedure for Maximum Efficiency}\label{procedure}
\subsection{Simple Procedure For Peer Selection}
All the peers are rational and aware of the fact that, if they will share their resources with peer having high SBCI, then their SBCI will be higher, and if they will download from a low SBCI peer then they will lose less SBCI. Therefore, the simple peer selection procedure for any peer $i$ is to download from low SBCI peer and to upload to high SBCI peer, as far as possible, as shown in Algorithm \ref{algo4.2}.

\begin{algorithm}
\caption{Simple Procedure for Peer Selection}\label{algo4.2}
\begin{algorithmic}[1]
\Procedure{}{}
\If { Peer $i$ needs a resource}
\State Send the request for resource;
\State Get the SBCI of responding peers from their respective index managers; 
\State Select the source peer having minimum SBCI;
\State Download the resource;
\State Send the value of resource to the index manager of source peer;
\EndIf
\If {Peer $i$ get a request for a resource}
\State Get the SBCI of requesting peer from their respective index managers;
\If{ SBCI of all requesting peer is less than the threshold}
\State Reject all the requesting peers;
\Else
\State Select the peer with maximum SBCI;
\State Upload the resource;
\State Send the value of resource to the index manager of downloading peer;
\EndIf
\EndIf
\EndProcedure
\end{algorithmic}
\end{algorithm}

\subsection{College Admission and The Stability of Marriage Based Approach For Peer Selection }
\subsubsection{Preliminaries}
College Admission and the stability of Marriage is a well-known problem, introduced by Gale and Shapley \cite{stable}. In its most popular variants, there are two disjoint sets of cardinality, $n$. One set is representing the men and the other one is representing the women. Each person has a different order of preference for his or her marriage partner. There are several ways by which one-to-one pairing can be done. But a pairing is said to be stable, if  there is no pair both of whom prefer each other to their actual partner.\par Gale and Shapley \cite{stable} provide the solution and the algorithm for stable pairing. They also proved that there always exists a stable match for such type of problem. In this algorithm, one of the group proposes his or her first preference, another group can reject the proposal or can keep it on hold until they get a better option. If any member from the proposing group get rejected, he or she tries on next preference. This process continues until proposing group is not rejected or rejected by all of his or her preferred  partners.\par If a proposal is given by men, then they get the better preferred partner as compared to any other stable pairing, hence it is called man optimal stable matching, the other way around women optimal stable matching.
\subsubsection{Application in Peer Selection}
We considered the situation where there are many uploaders and many downloaders for a resource. In order to earn the high SBCI, uploader would like to upload the resource to high SBCI peers, thus they have certain preferences for downloaders. On the other hand, for downloaders the resource and the SBCI both matter. Therefore, downloader may prefer the higher bandwidth uploader over low SBCI uploader. Thus, downloder have a different preference order for uploaders.\par In this situation, all uploaders and downloaders preference order can be collected at a certain node. We call it the resource manager node. This node can be found by hashing the resource identifier and finding corresponding root node in DHT network. On this node, the stable marriage algorithm can be used to pair the uploaders and downloaders. A message to each  pair will be sent after pairing, so that they can start the process of transaction. Detail of peer selection procedure in this situation is shown in Algorithm \ref{algo4.3}.  

\begin{algorithm}
\caption{College Admission and The Stability of Marriage Based Approach for Peer Selection}\label{algo4.3}
\begin{algorithmic}[1]
\Procedure{}{}
\If { Peer $i$ needs a resource}
\State Send the request for resource;
\State Get the SBCI of responding peers from their respective index managers; 
\State Learn about the bandwidth of responding peers;
\State Make the order of preference for uploader;
\State Send the order of preference to the resource manager;
\State Get the ID of uploader partner from resource manager;
\State Download the resource;
\State Send the value of resource to the index manager of uploading peer;
\EndIf
\If {Peer $i$ get a request for a resource}
\State Get the SBCI of requesting peers from their respective index managers;
\State Remove the peers, having SBCI  less than the threshold;
\State Make the order of preference according to their SBCI;
\State Send the order of preference to the resource manager;
\State Get the ID of downloader partner from resource manager;
\State Upload the resource;
\State Send the value of resource to the index manager of downloading peer;
\EndIf
\If {Peer $i$ is the resource manager}
\State Get the order of preference from uploaders;
\State Get the order of preference from downloaders;
\State Run the Stable marriage algorithm;
\State Send the ID of partner to each peer;
\EndIf
\EndProcedure
\end{algorithmic}
\end{algorithm}
\section{Experimental Evaluation}\label{experiment}
As in \cite{sfa} and \cite{eval}, we used NetLogo 5.2 \cite{netlogo}, to evaluate the performance of our algorithm. NetLogo is a multiagent programmable modeling environment where we can model different agents and can ask them to perform the task in parallel and independently. It is written mostly in Scala, with some parts in Java.
\subsection{Simulation Setup}
We simulated a typical P2P network with parameters and distributions taken from real world measurements as in \cite{gnutela}, \cite{ton13}. In this network, peers can send a query for the  resource. We assumed that ten percent of peers respond to this query. After selecting the source peer according to the procedure described in Section \ref{procedure}, resource is downloaded. We assumed the amount of resources requested by downloading peers varies randomly between 1 unit  to 255 units. After downloading the resource, SBCI of peer is updated by an index manager using  (\ref{eq4.5}). Any peer whose SBCI is less than the threshold value is rejected and cannot download the  resources from the network. We assumed the threshold value of SBCI to be $\alpha/(1+\alpha)$.   \par The number of nodes in the network is taken as 1000, which is reasonable size. However, the number of nodes can be increased up to any number, but this will not affect the results. Because, evaluation metrics are normalized with respect to the number of nodes. The initial value of SBCI of all the peers are taken as $\alpha/(1+\alpha)$. We conducted the experiment for $\alpha= 0.9,0.6$ and $0.3$. Percentage of free-riders were varied from 10\% to 80\%. 
\par The simulation is performed for three different peer distribution models, i.e., Simple, Adaptive and Extreme Model.\par In Simple Model, free-riders  vary from 10\% - 70\%. These free-riders do not share anything at any point of time in the simulation.\par In Adaptive Model, free-riders vary from 20\% - 60\%. Half of these free-riding peers do not share anything during the whole simulation. Remaining half behave as normal peers till midway of simulation, and thereafter convert
themselves to free-riders.\par 
In Extreme Model, at the beginning of simulation, 10\% peers are free-riders. After completion of every 12.5\% of total transactions, 10\% more peers convert themselves to free-riders. Thus, at the end of simulation, there will be 80\% free-riders. The simulation was run upto 100000 transactions. 
\subsection{Evaluation Metrics}
We plotted the graph between the total upload and download amounts of each peer for all the models. To get the deeper picture, we also calculated the average absolute deviation (AAD) of upload to download ratio from one, in any model as:
\[AAD = (1/N)\sum_{i=1}^{N}|1-\mathbf{e_i.S_{comp}(t_n).e/e_i.{S^{tr}}_{comp}(t_n).e}|.\] 
If upload amount for each peer is same as download amount, then the value of $AAD$ in the network will be zero. The larger value of $AAD$ implies, the larger difference between upload and download and thus, lesser fairness in the network.\par Network is said to be efficient if free-riders are not allowed to download anything, without affecting the transactions between non-free-rider peers. At any time, if SBCI of any cooperative peer is less than the threshold then it will also get rejected. This is not a desired state in the network. Therefore, we calculated the percentage of rejections among cooperative peers, i.e., cooperative peers  rejecting the request of cooperative peers. For efficient algorithm, percentage of rejections must be minimum. \par  For comparison,  we also simulated the GC for its best case \cite{global}, i.e., $\alpha=0.8$ and $\beta= 0.2$. The parameters $\alpha$ and $\beta$ are taken to be same as in \cite{global}. For fair comparison, we kept the threshold value for peer selection as $(2-\alpha(1+\beta))/(2+\alpha(1-\beta))$. We kept maximum value of  threshold in both GC as well as in SBCI. Rest of the settings for GC are same as in SBCI.

\subsection{Simulation Results of Simple Procedure For Peer Selection}


\begin{figure*}
\centering 
  \subfloat[Free Riders = 10 \%, $\alpha=0.9$]{ \includegraphics[width=6cm,height=4cm,scale=.18]{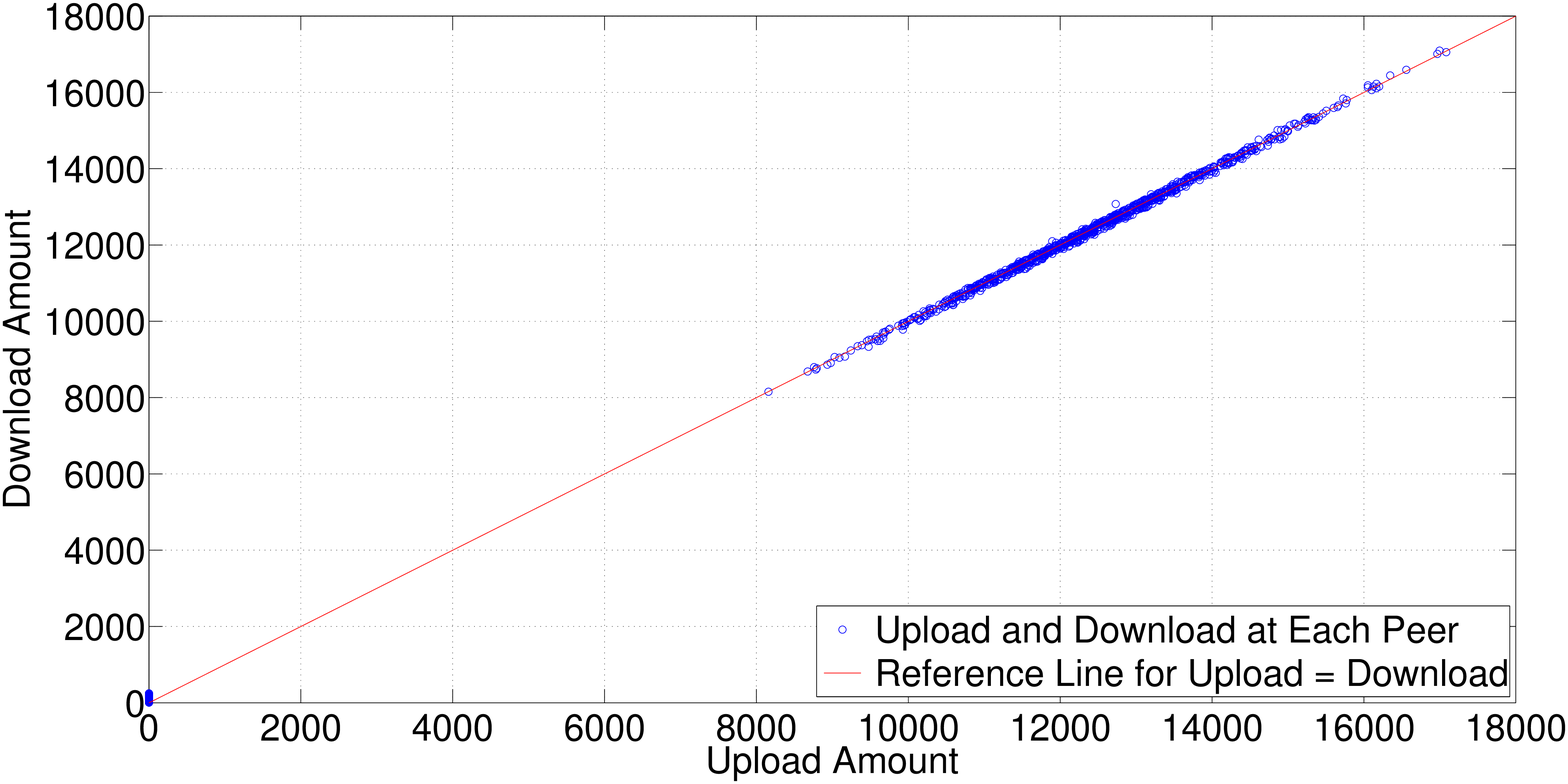}  }
\subfloat[Free Riders = 10 \%, $\alpha=0.6$]{
\includegraphics[width=6cm,height=4cm,scale=.18]
{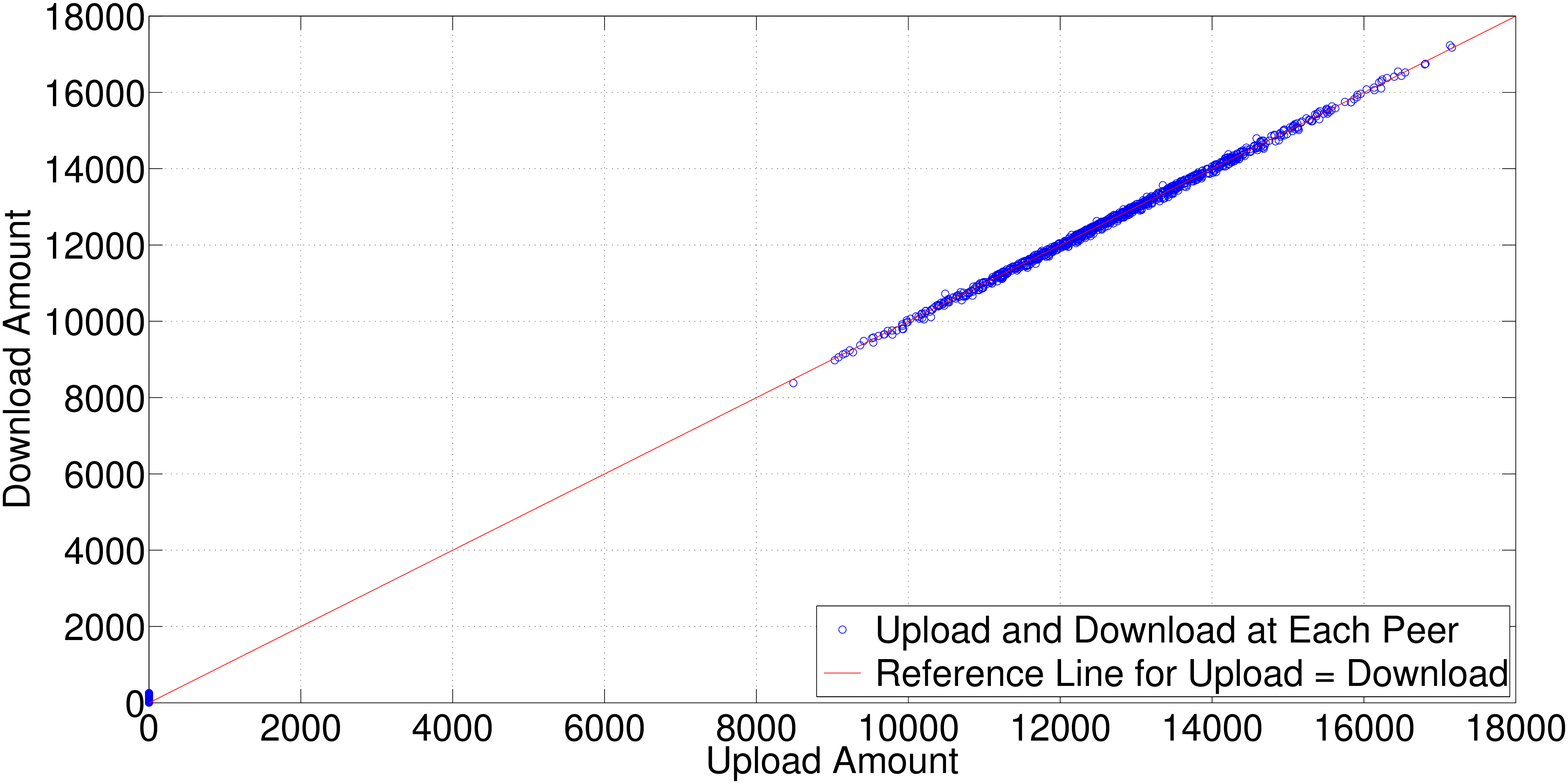} }
\subfloat[Free Riders = 10 \%, $\alpha=0.3$]{
\includegraphics[width=6cm,height=4cm,scale=.18]
{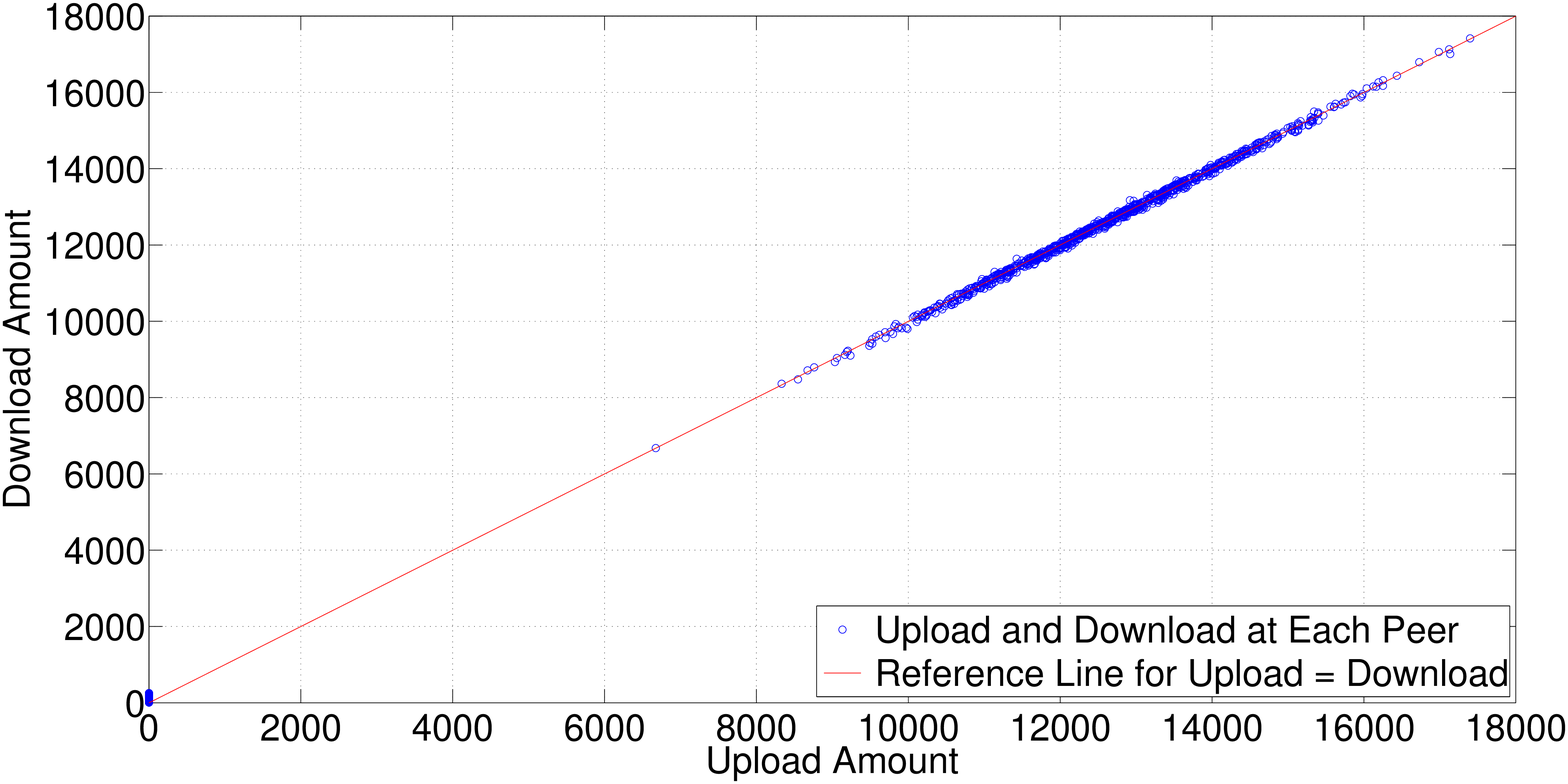}}\\

\subfloat[Free Riders = 30 \%, $\alpha=0.9$]{
\includegraphics[width=6cm,height=4cm,scale=.18]
{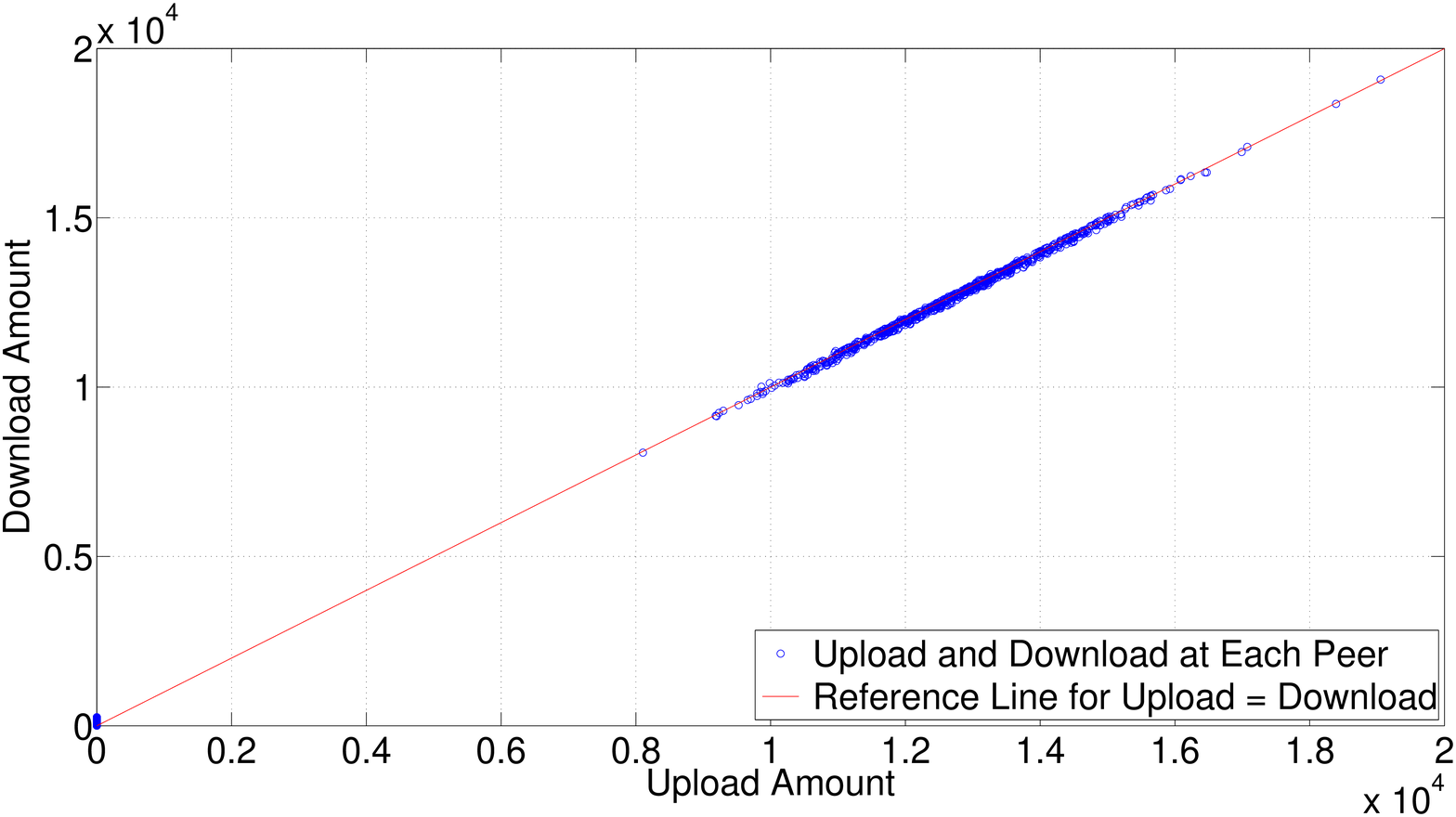} }
\subfloat[Free Riders = 30 \%, $\alpha=0.6$]{
\includegraphics[width=6cm,height=4cm,scale=.18]
{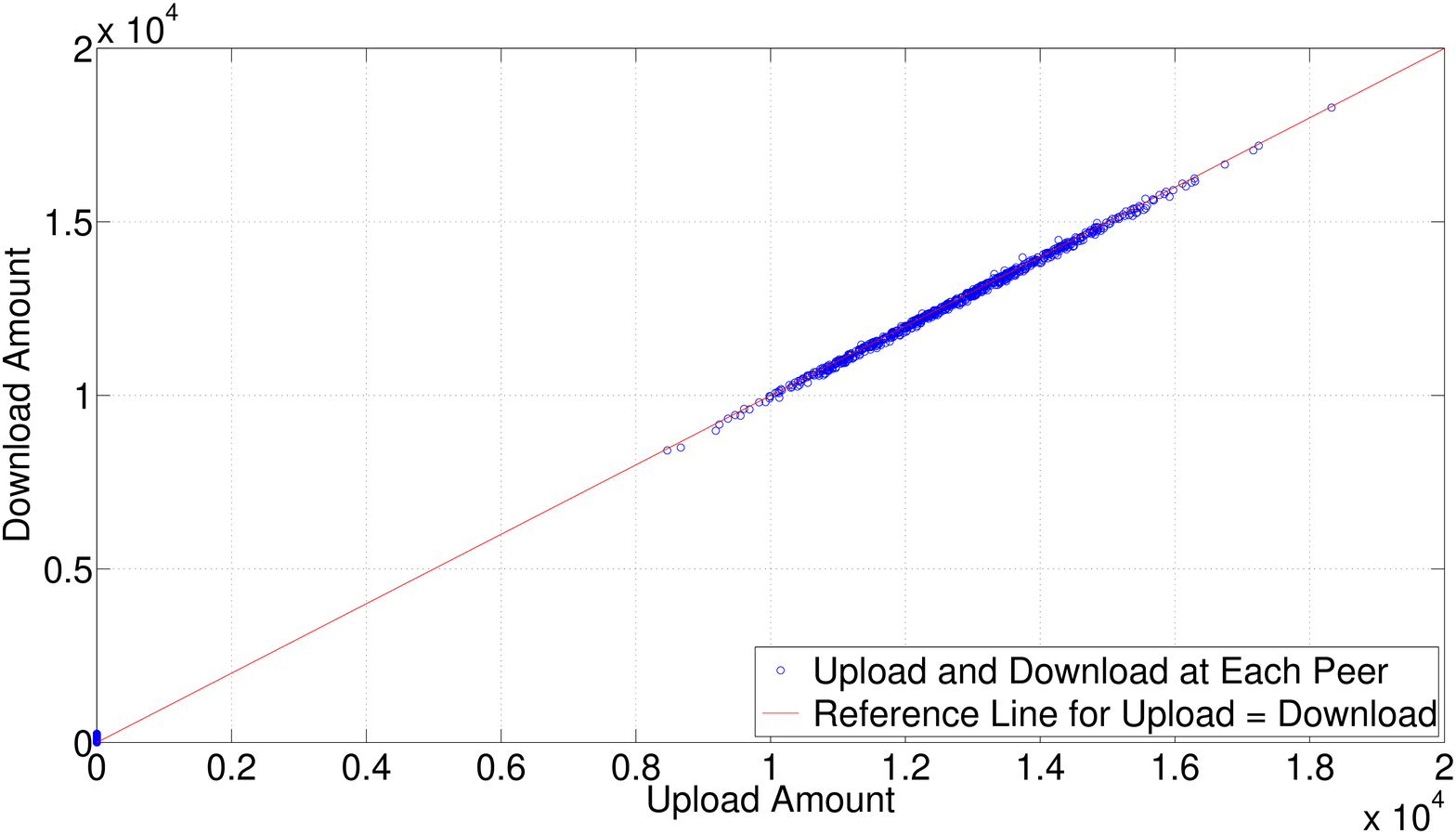}}
\subfloat[Free Riders = 30 \%, $\alpha=0.3$]{
\includegraphics[width=6cm,height=4cm,scale=.18]
{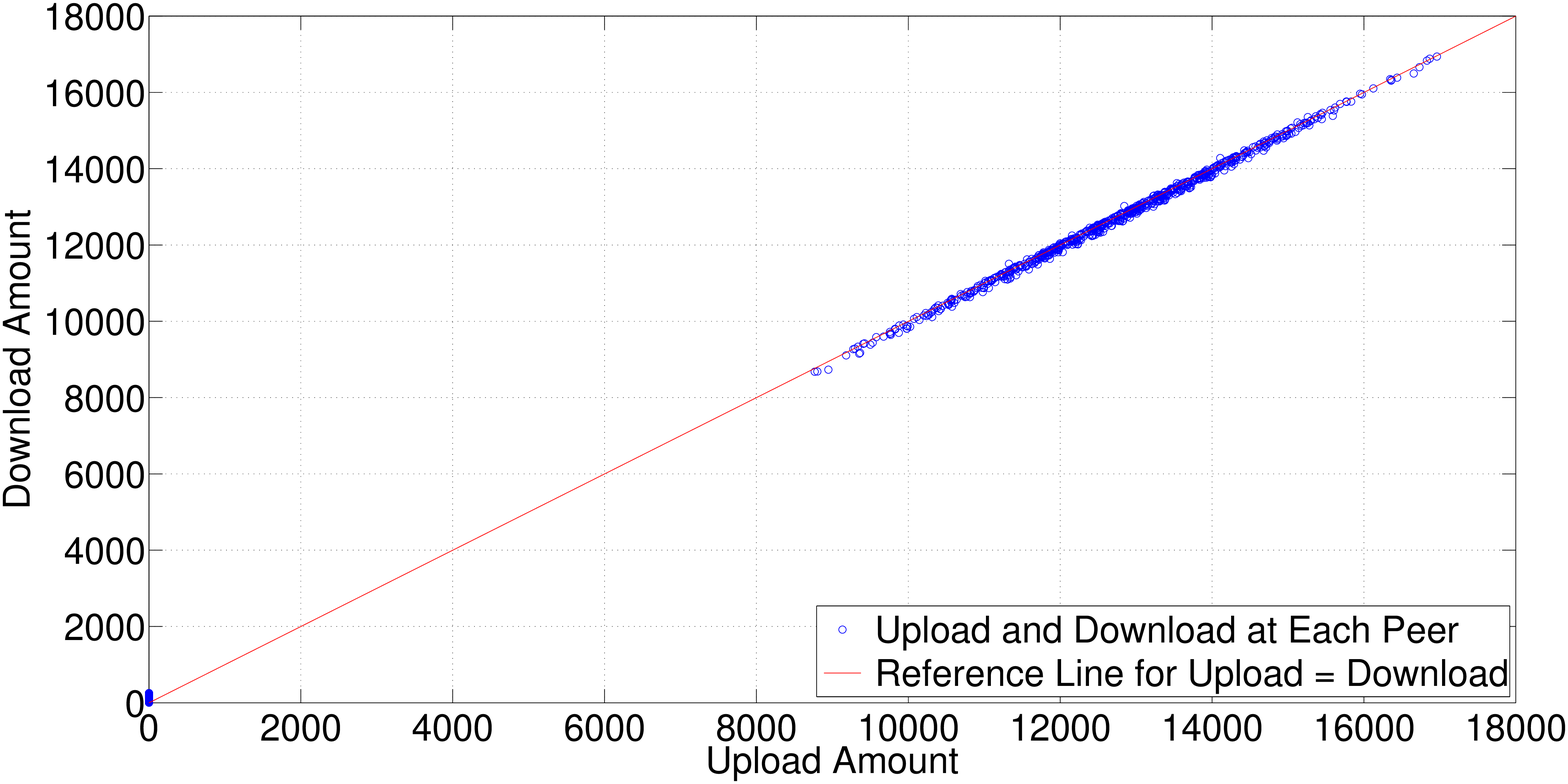} }\\
\subfloat[Free Riders = 50 \%, $\alpha=0.9$]{
\includegraphics[width=6cm,height=4cm,scale=.18]
{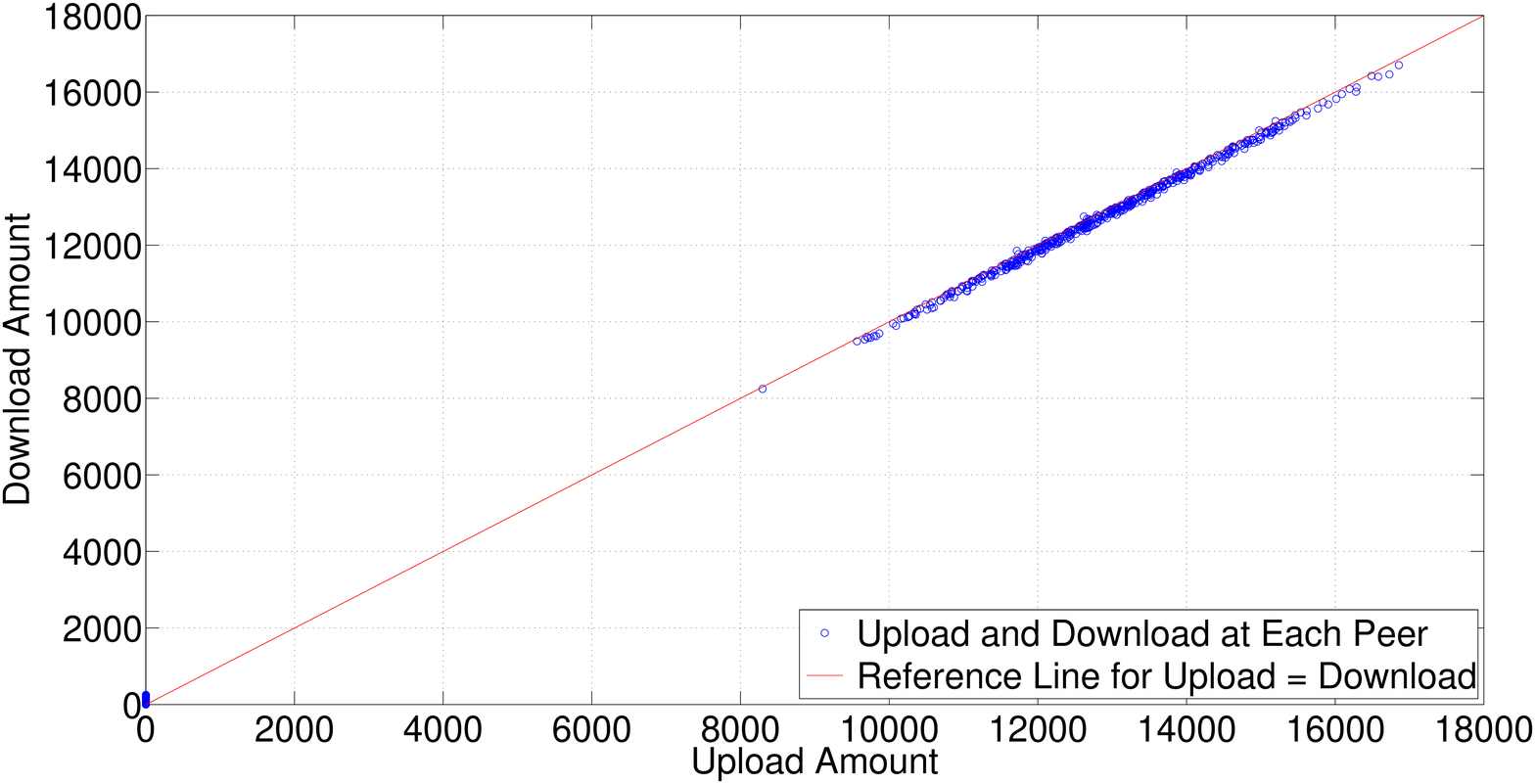}}
\subfloat[Free Riders = 50 \%, $\alpha=0.6$]{
\includegraphics[width=6cm,height=4cm,scale=.18]
{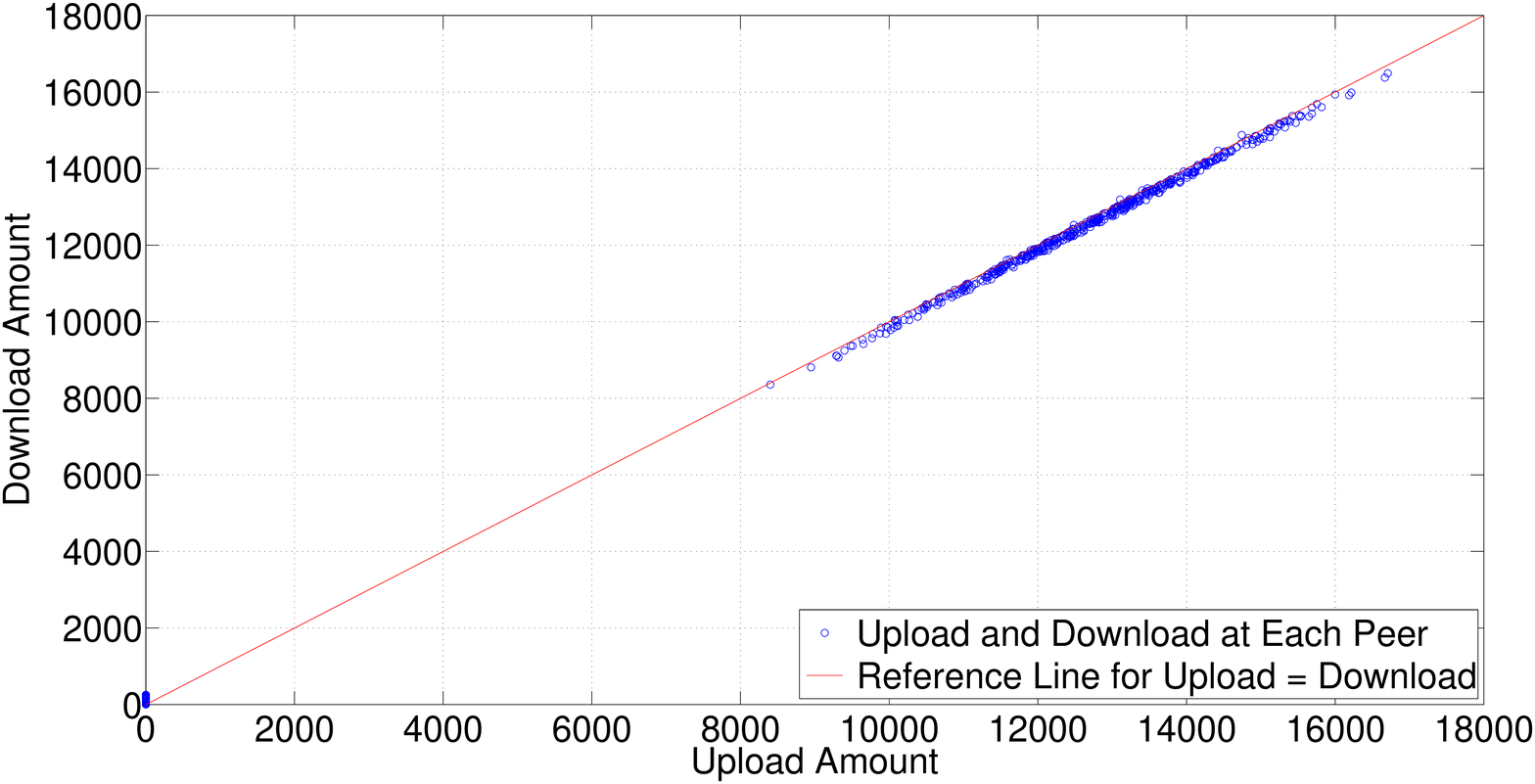}  }
\subfloat[Free Riders = 50 \%, $\alpha=0.3$]{
\includegraphics[width=6cm,height=4cm,scale=.18]
{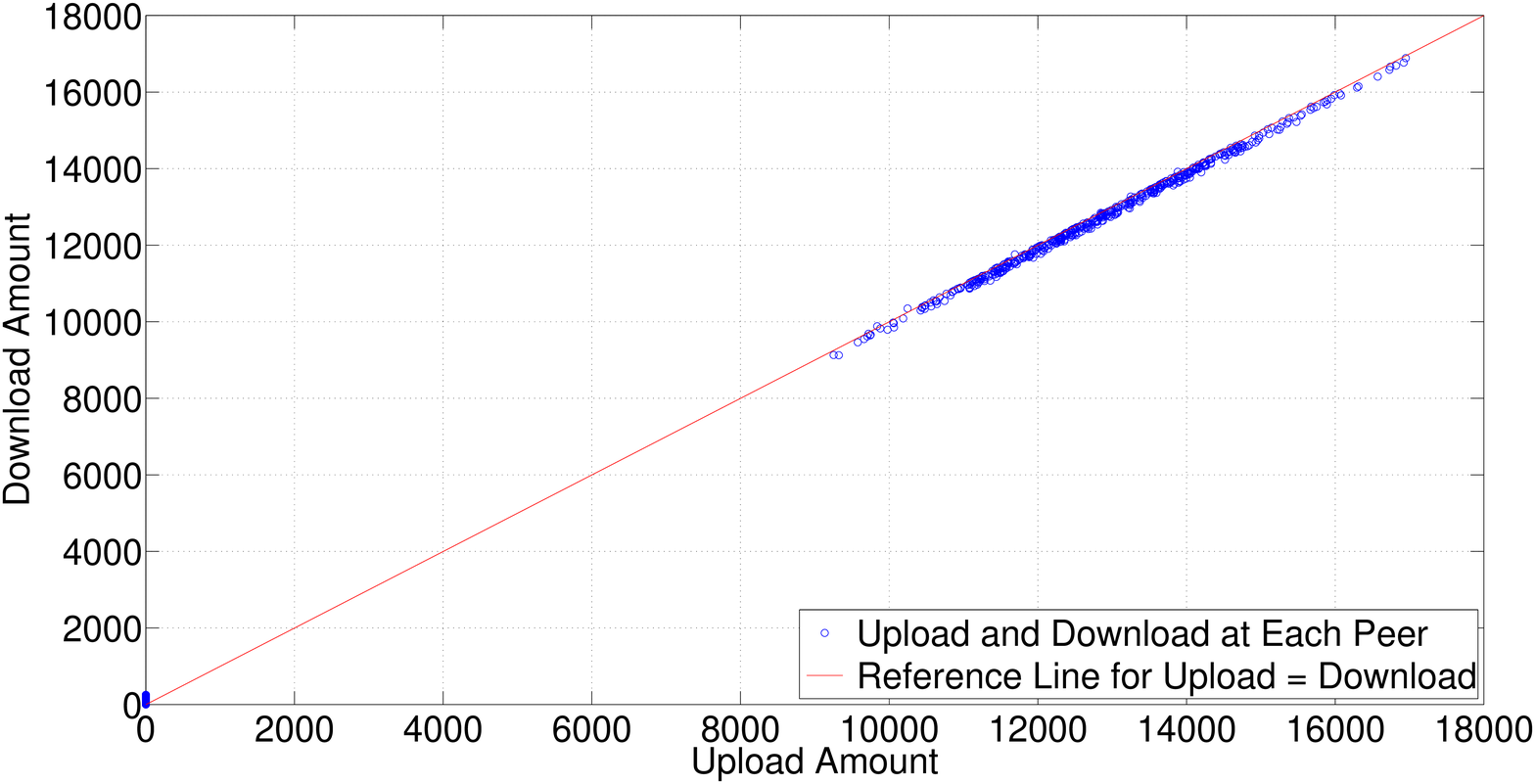}}\\
\subfloat[Free Riders = 70 \%, $\alpha=0.9$]{
\includegraphics[width=6cm,height=4cm,scale=.18]
{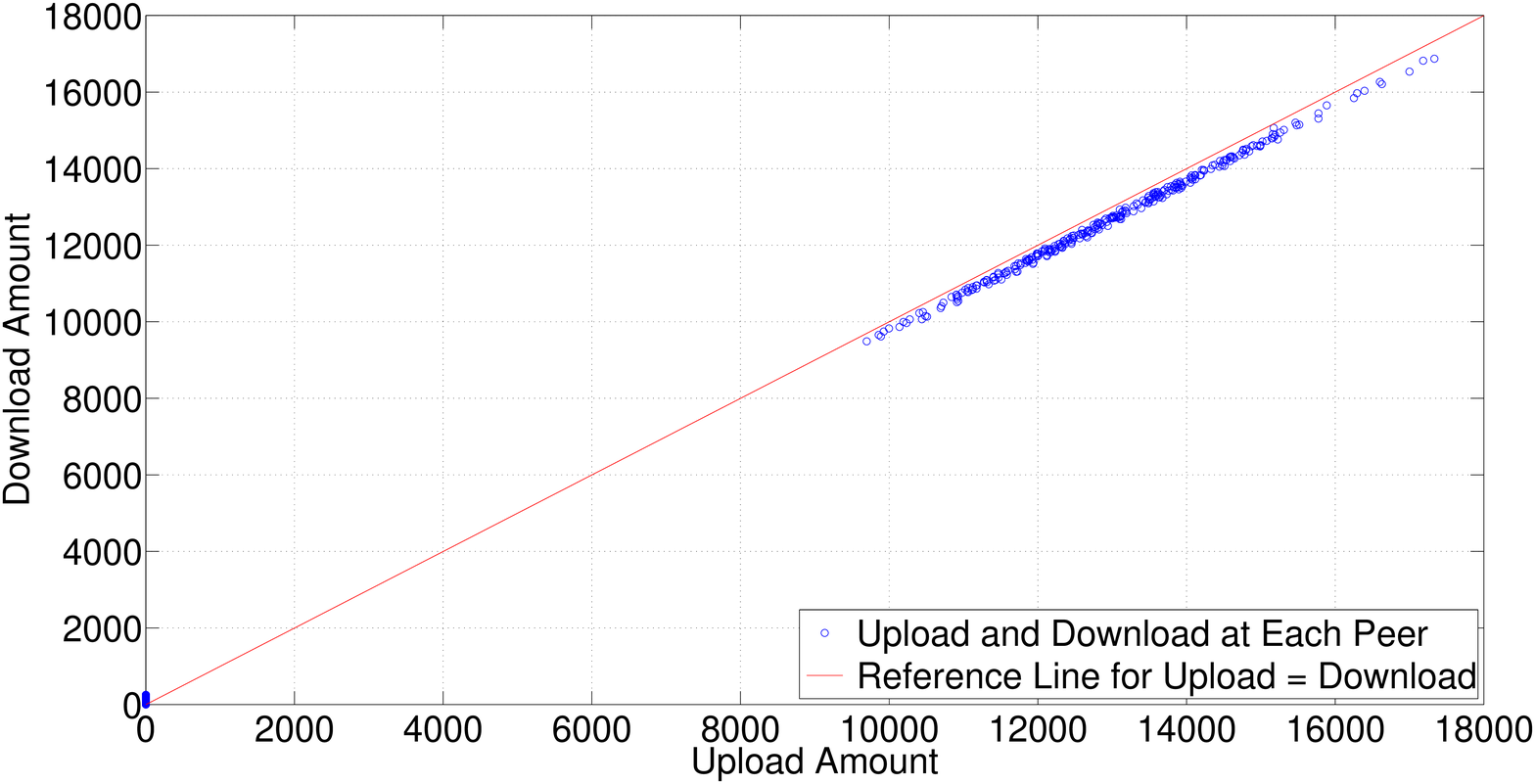} }
\subfloat[Free Riders = 70 \%, $\alpha=0.6$]{
\includegraphics[width=6cm,height=4cm,scale=.18]
{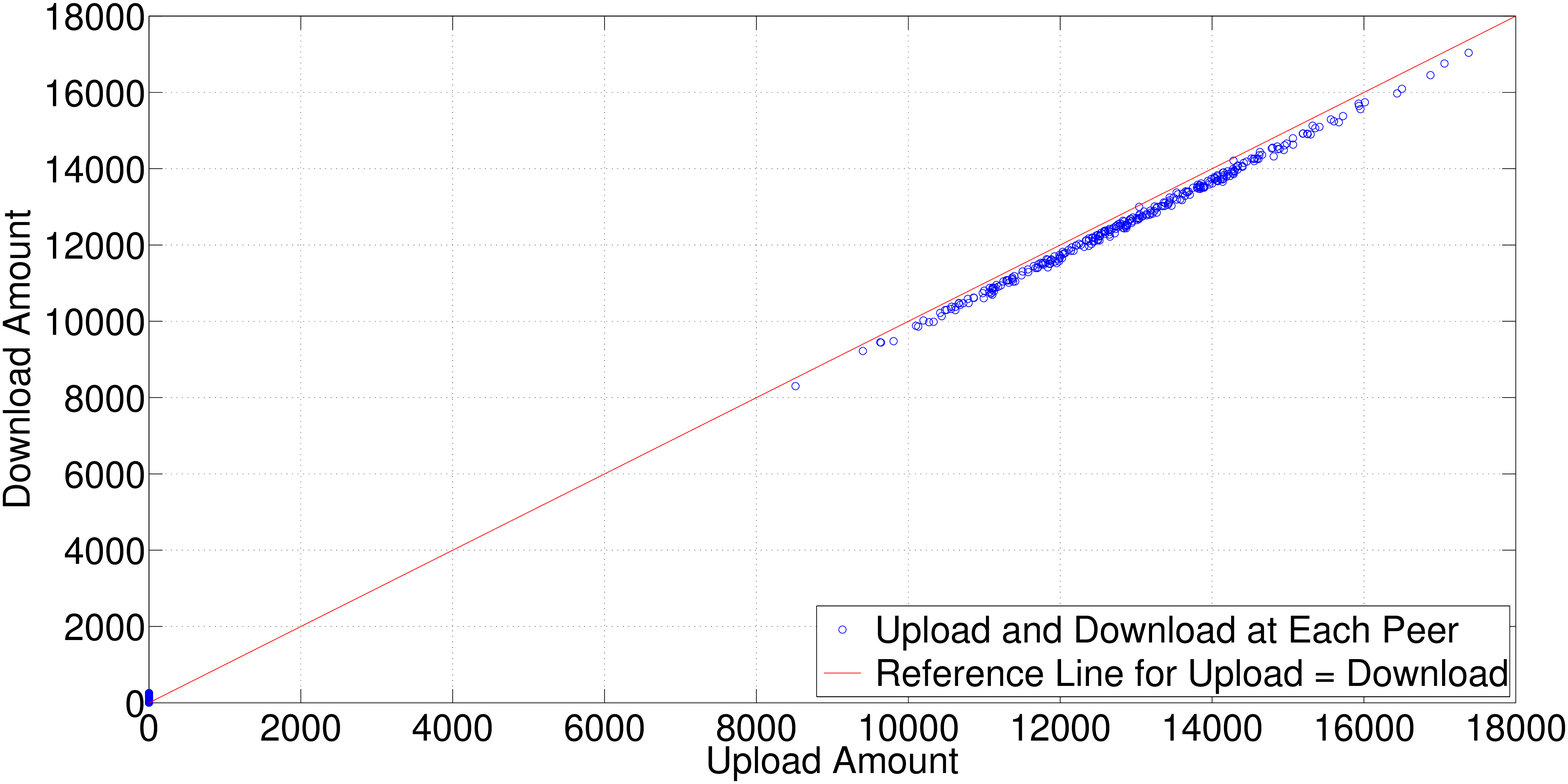}}
\subfloat[Free Riders = 70 \%, $\alpha=0.3$]{
\includegraphics[width=6cm,height=4cm,scale=.18]
{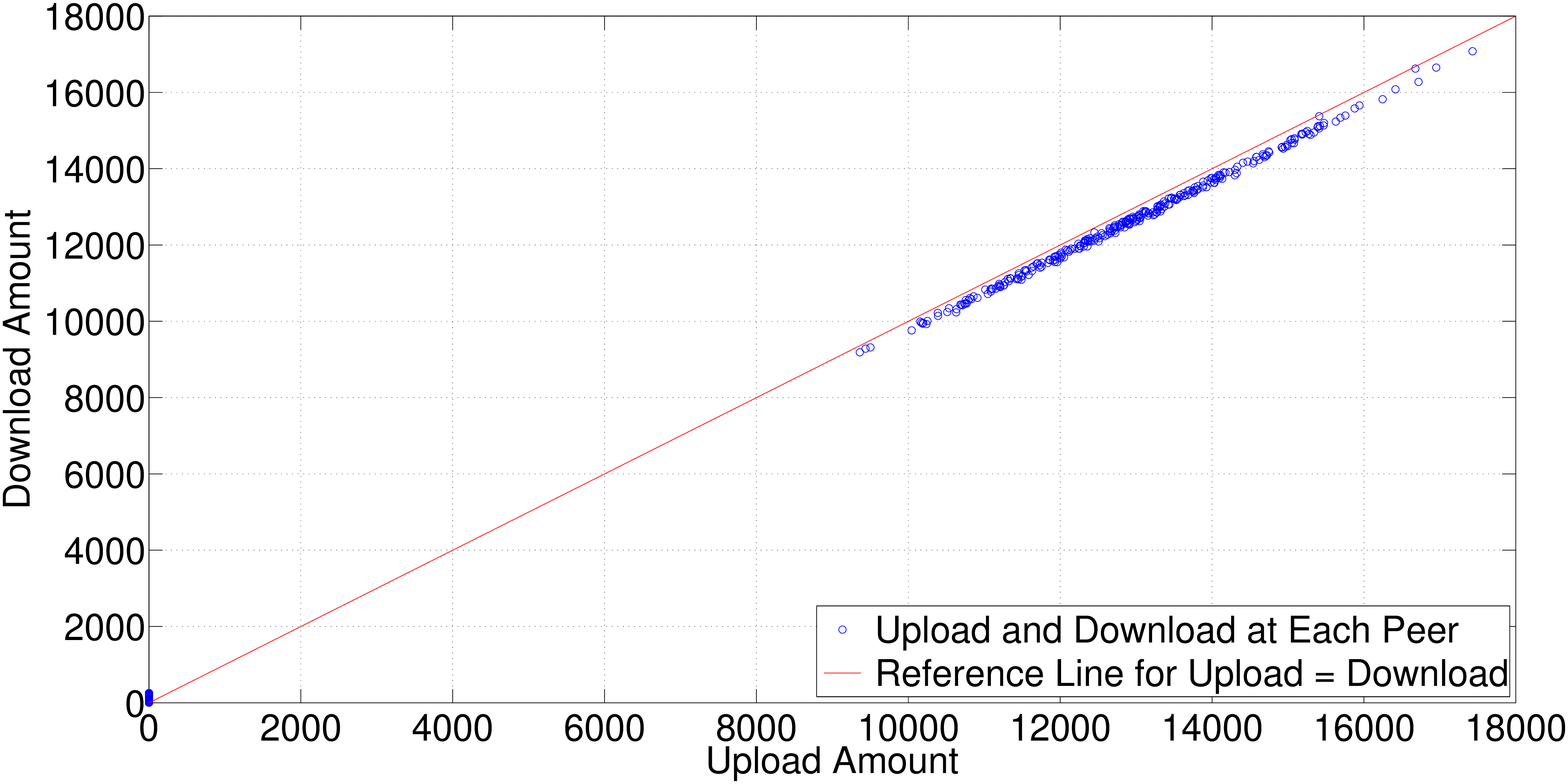} }
 \caption{Upload and Download Amount at Each Peer for SBCI in Simple Model for Simple Procedure of Peer Selection. Free-riders are varied from $10\%-70\%$ and the value of $\alpha$ is taken as $0.9, 0.6$ and 0.3}\label{seventy}
\end{figure*}

\begin{table}
\begin{center}
\caption{ AAD and $\%$ of Rejections for SBCI in Simple Model for Simple Procedure of Peer Selection }\label{table4.1}
\begin{tabular}{ | m{0.5cm} | m{1em}| m{5em}|m{5em}| m{5.5em}|} 
 \hline
  S.N.& $\alpha$ &Free-riders & AAD &$\%$ of Rejections \\[1ex] 
 \hline
 $1$& 0.9  & 10\% &  0.103772&1.817\\
    \hline
 $2$& 0.9  & 30\% &  0.303675&0.127\\
    \hline
 $3$& 0.9  & 50\% &  0.505119&0.023\\
    \hline
 $4$& 0.9  & 70\%&  0.707233&0.008\\
    \hline
 $5$& 0.6  & 10\% & 0.103667&0.064\\
    \hline
 $6$& 0.6  & 30\% &  0.303652&0.009\\
    \hline
 $7$& 0.6  & 50\% &  0.505459&0.006\\
    \hline
 $8$& 0.6  & 70\% &  0.707117&0.003\\
    \hline
 $9$& 0.3  & 10\% &  0.103844&0.009\\
  \hline
  $10$& 0.3  & 30\% &  0.303671&0.002\\
    \hline
 $11$& 0.3  & 50\% &  0.505088&0.001\\
    \hline
 $12$& 0.3  & 70\% &  0.707069&0\\
    \hline
\end{tabular}
\end{center}
\end{table}

We conducted the simulation experiment for simple procedure of peer selection, as explained in Section \ref{procedure}. Bandwidth of all the peers is assumed to be same. For simple model, simulation results for SBCI  are shown in Fig. \ref{seventy}. Corresponding $AAD$ and percentage of rejections among cooperative peers are shown in Table \ref{table4.1}. We can observe from this  figure that in initial transactions, free-riders got some resources after that their SBCI become zero, which disqualify them in taking any resources from the network. For all other peers, upload to download ratio is very close to the reference line, thus algorithm is able to maintain the fairness in the network. We can observe from Table \ref{table4.1} that the percentage of rejections among the cooperative peers are more for higher values of $\alpha$. Because for higher values of $\alpha$, threshold value of SBCI will be higher. But its impact on $AAD$ is not very significant in this model.\par  
\begin{figure*}
\centering 
  \subfloat[Free Riders = 20 \%, $\alpha=0.9$]{ \includegraphics[width=6cm,height=4cm,scale=.18]{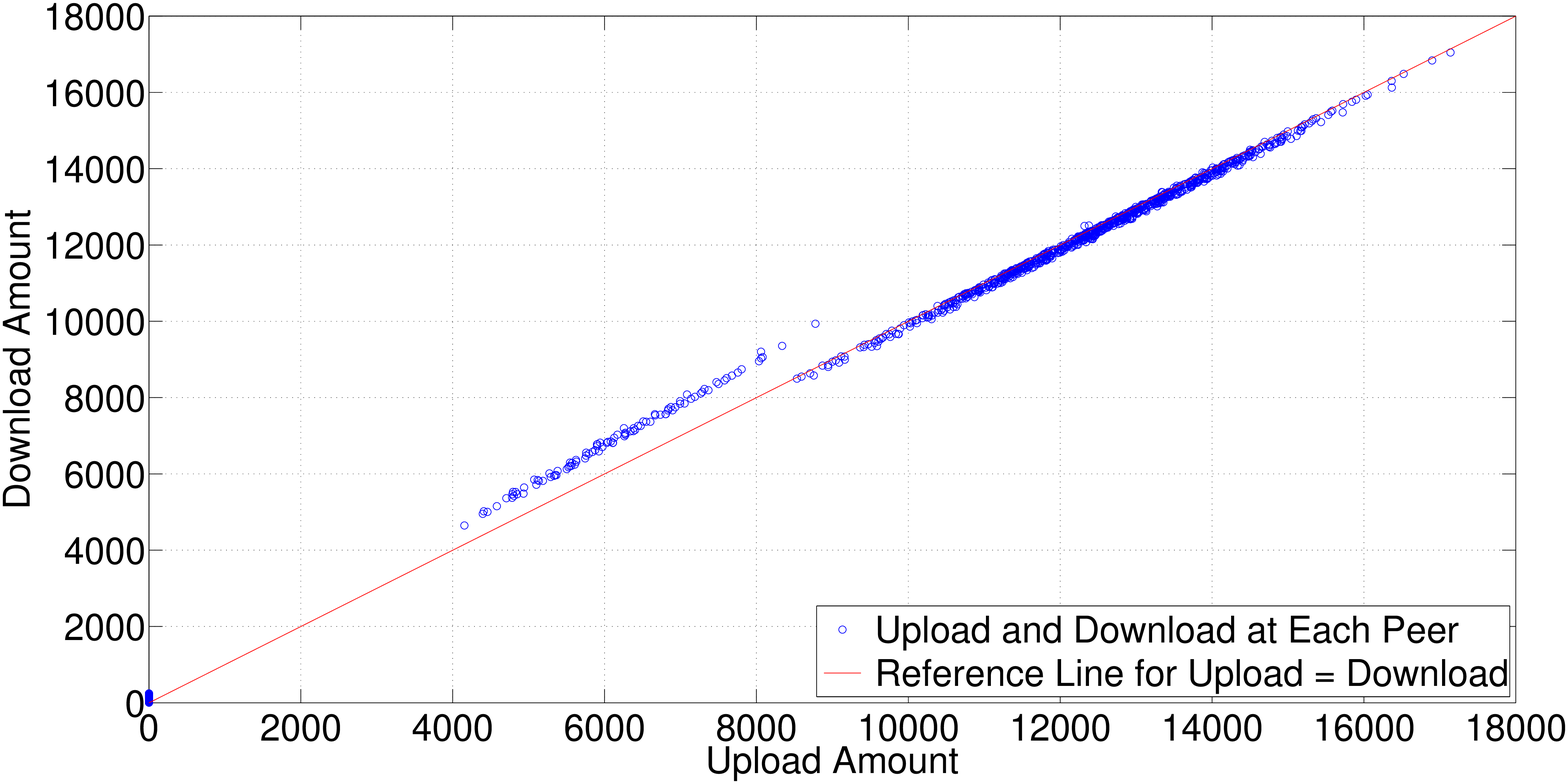}  }
\subfloat[Free Riders = 20 \%, $\alpha=0.6$]{ \includegraphics[width=6cm,height=4cm,scale=.18]{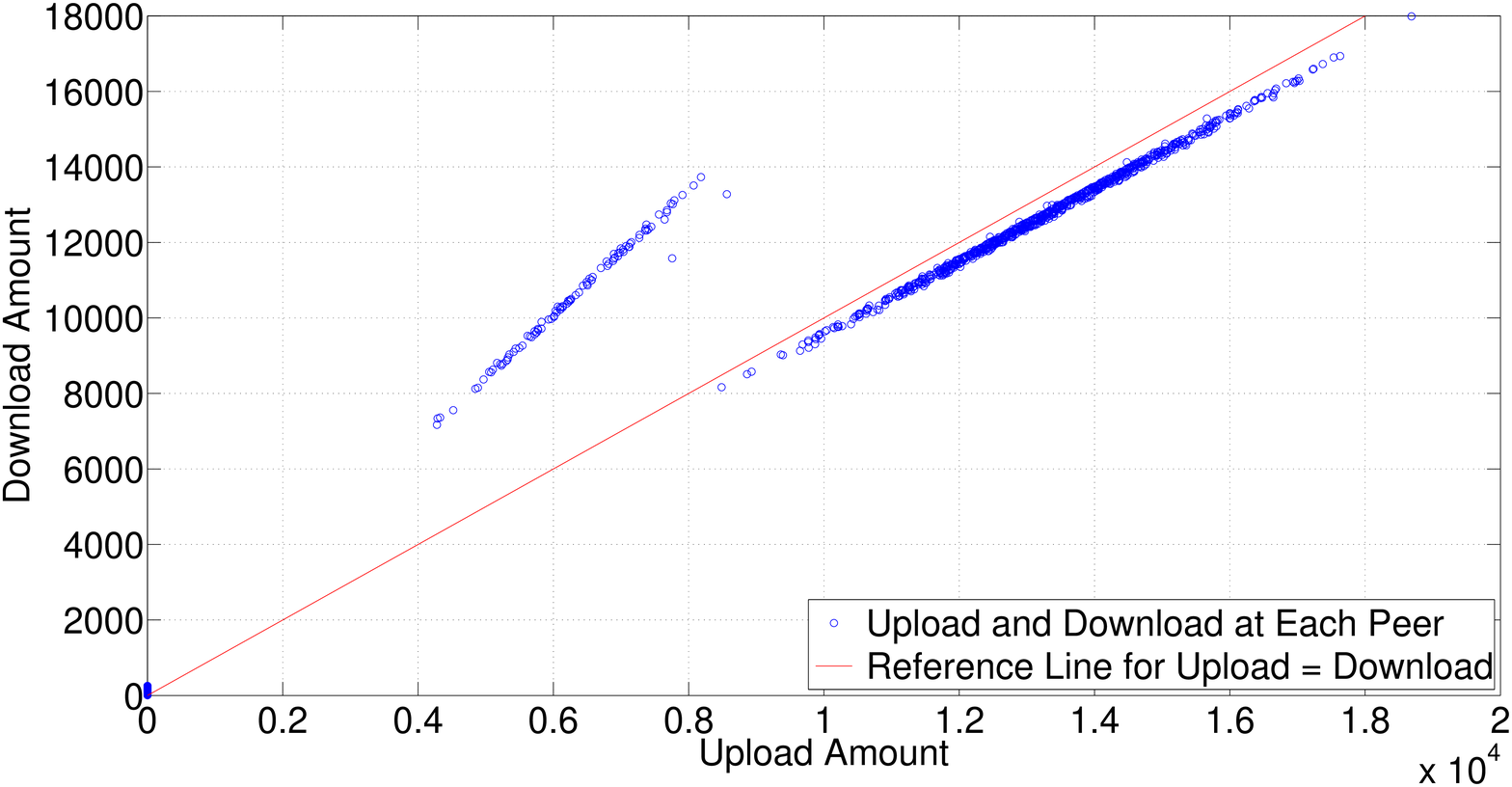}}
\subfloat[Free Riders = 20 \%, $\alpha=0.3$]{ \includegraphics[width=6cm,height=4cm,scale=.18]{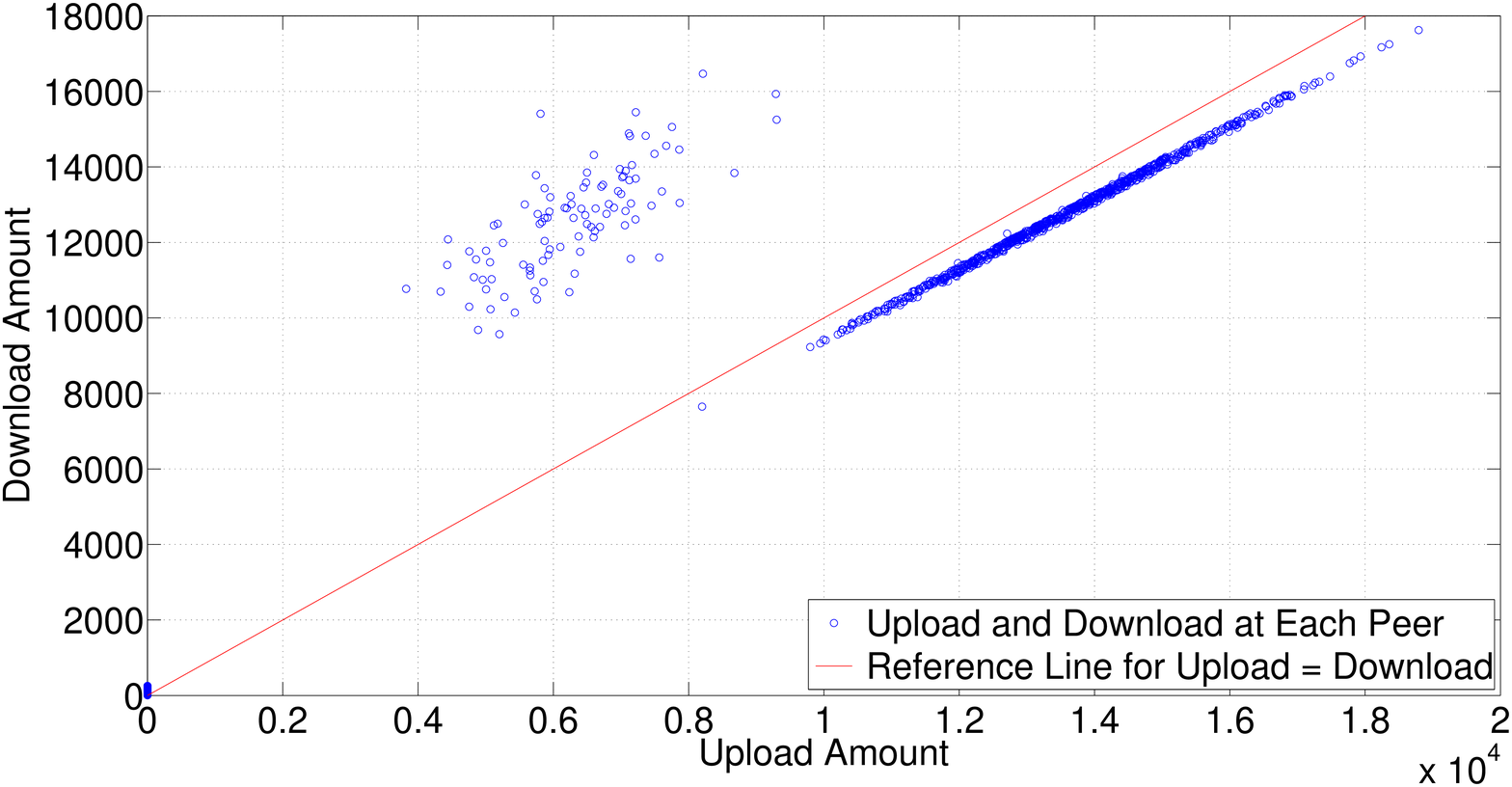}}\\
  \subfloat[Free Riders = 40 \%, $\alpha=0.9$]{ \includegraphics[width=6cm,height=4cm,scale=.18]{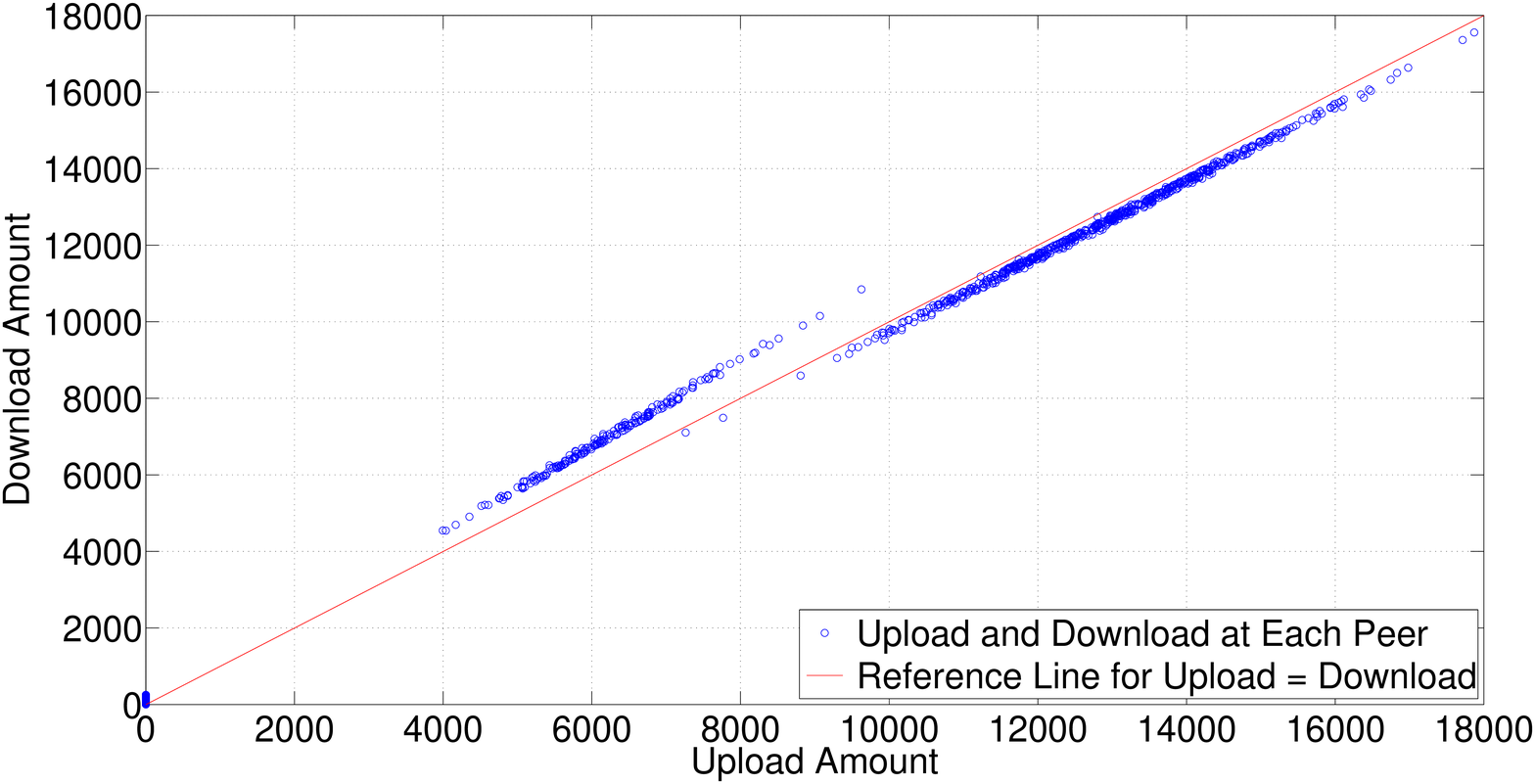}  }
\subfloat[Free Riders = 40 \%, $\alpha=0.6$]{ \includegraphics[width=6cm,height=4cm,scale=.18]{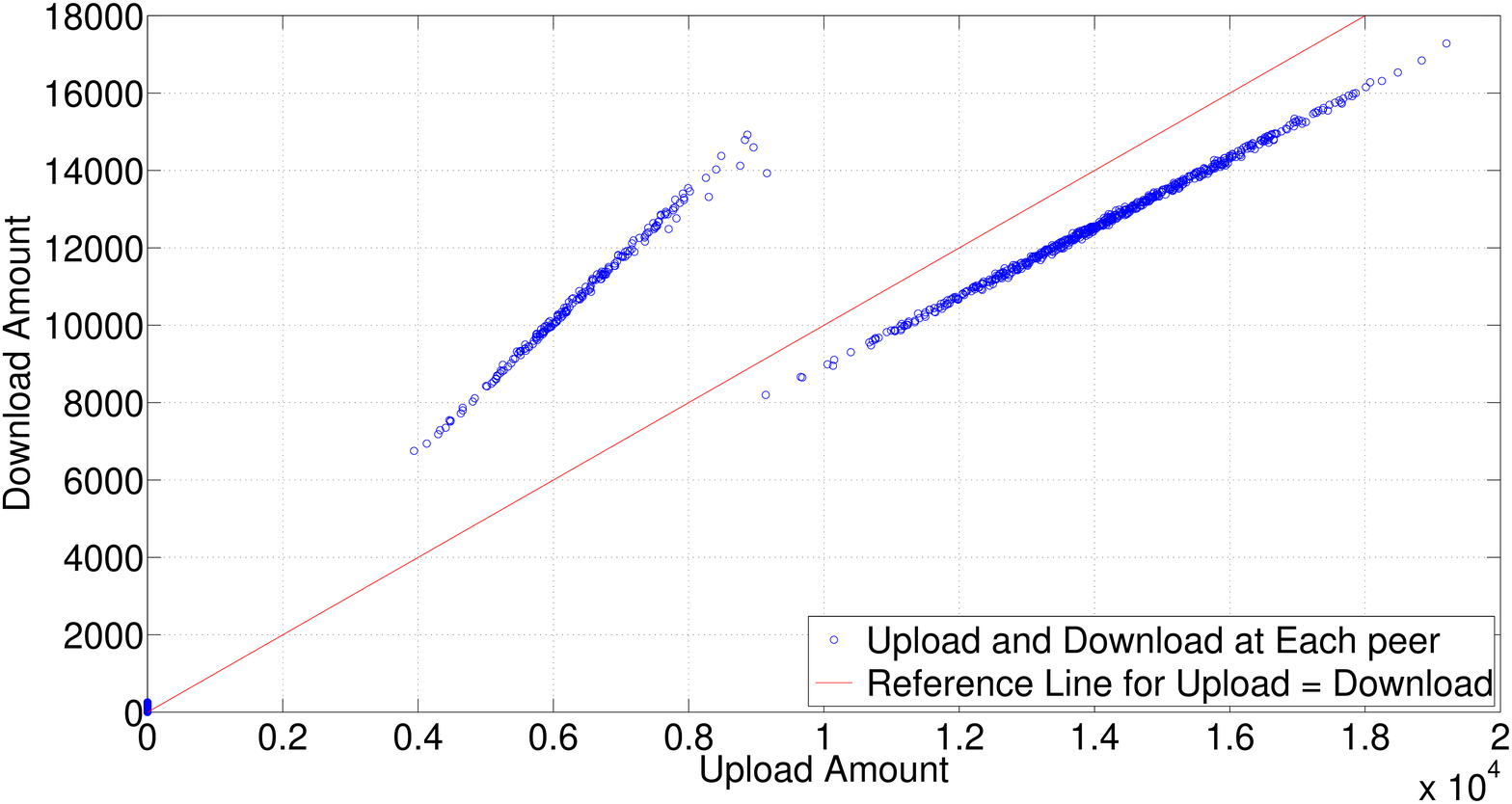}  }
\subfloat[Free Riders = 40 \%, $\alpha=0.3$]{ \includegraphics[width=6cm,height=4cm,scale=.18]{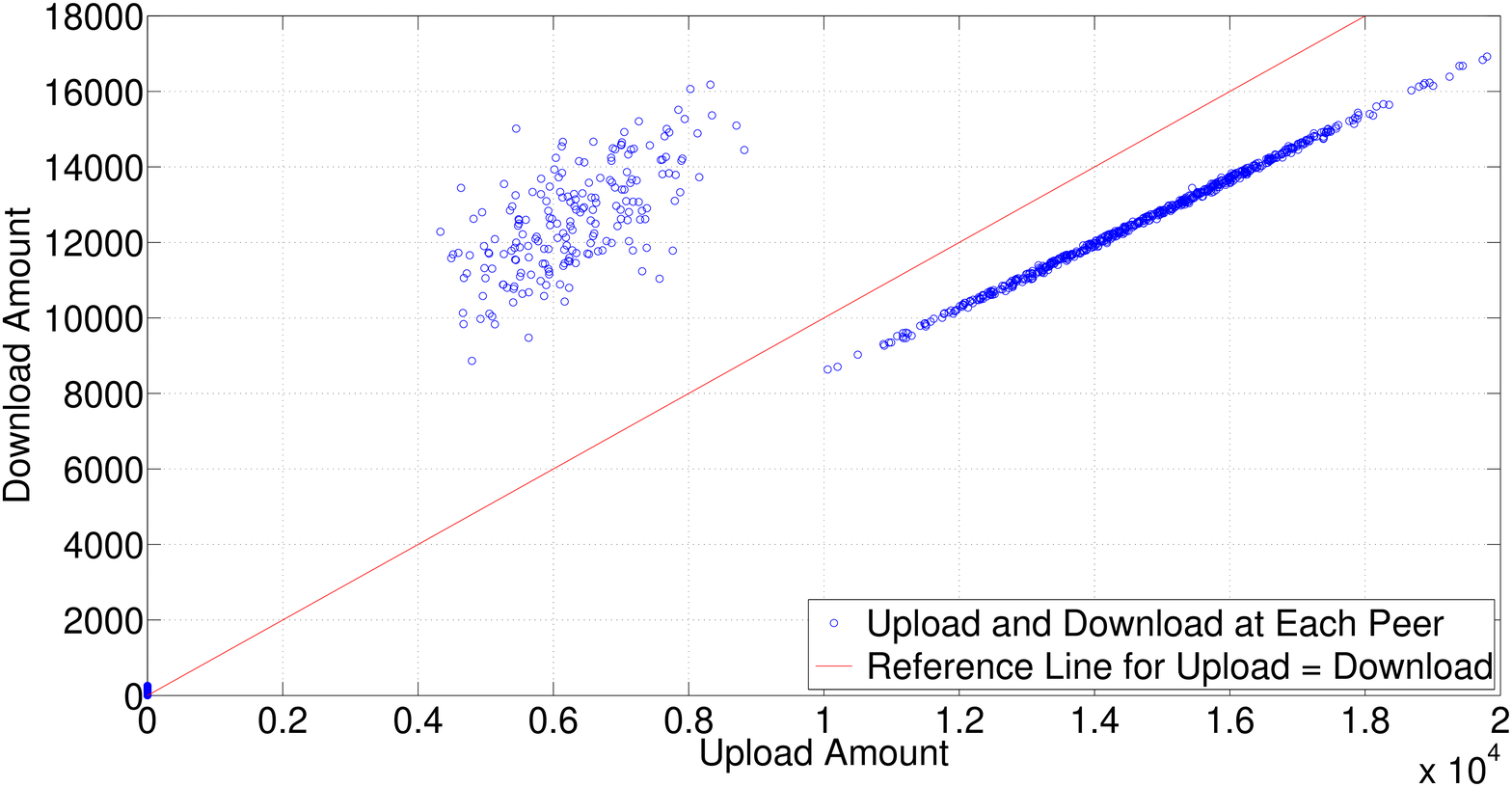}  }\\
  \subfloat[Free Riders = 60 \%, $\alpha=0.9$]{ \includegraphics[width=6cm,height=4cm,scale=.18]{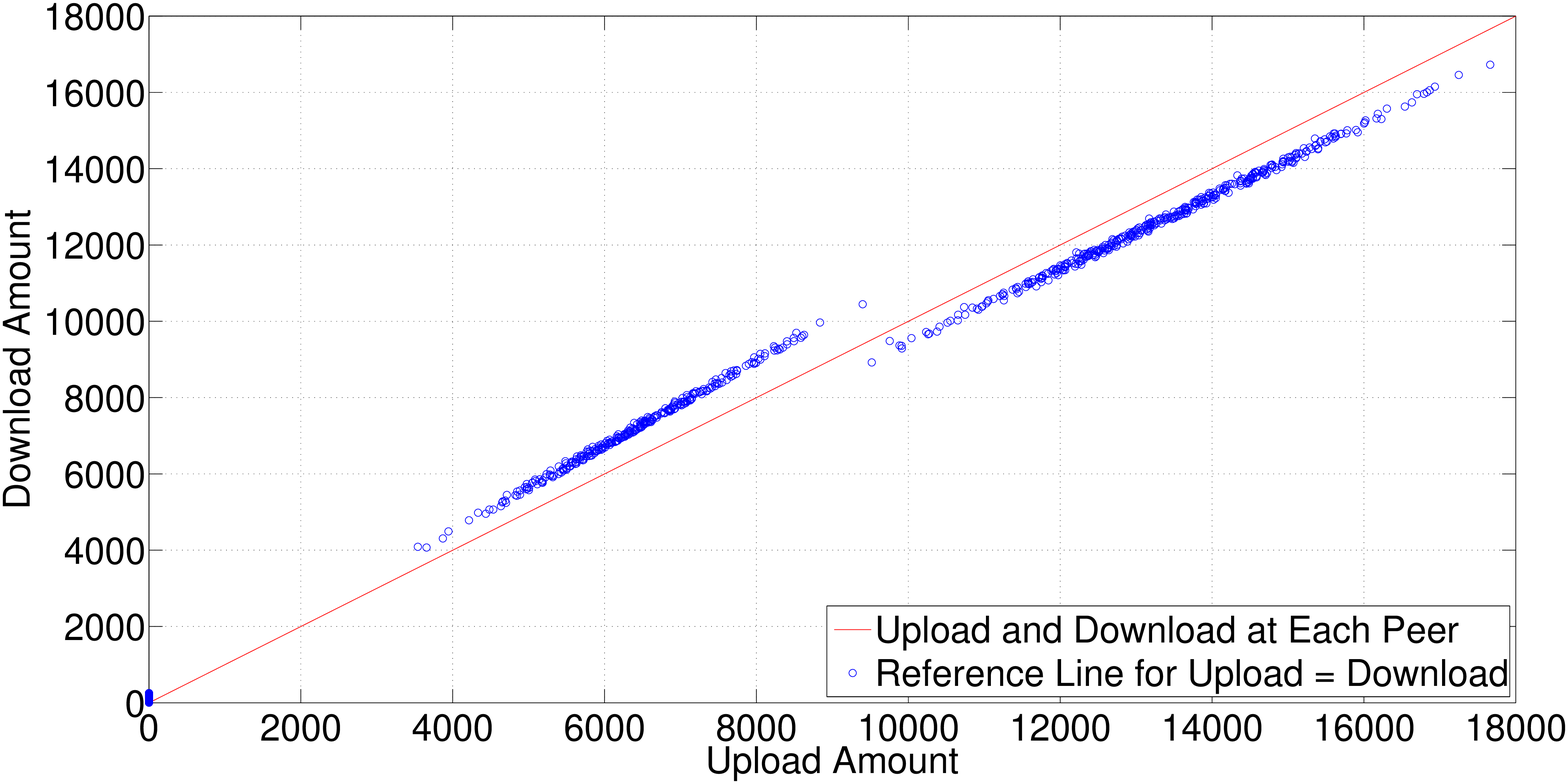}  }
\subfloat[Free Riders = 60 \%, $\alpha=0.6$]{ \includegraphics[width=6cm,height=4cm,scale=.18]{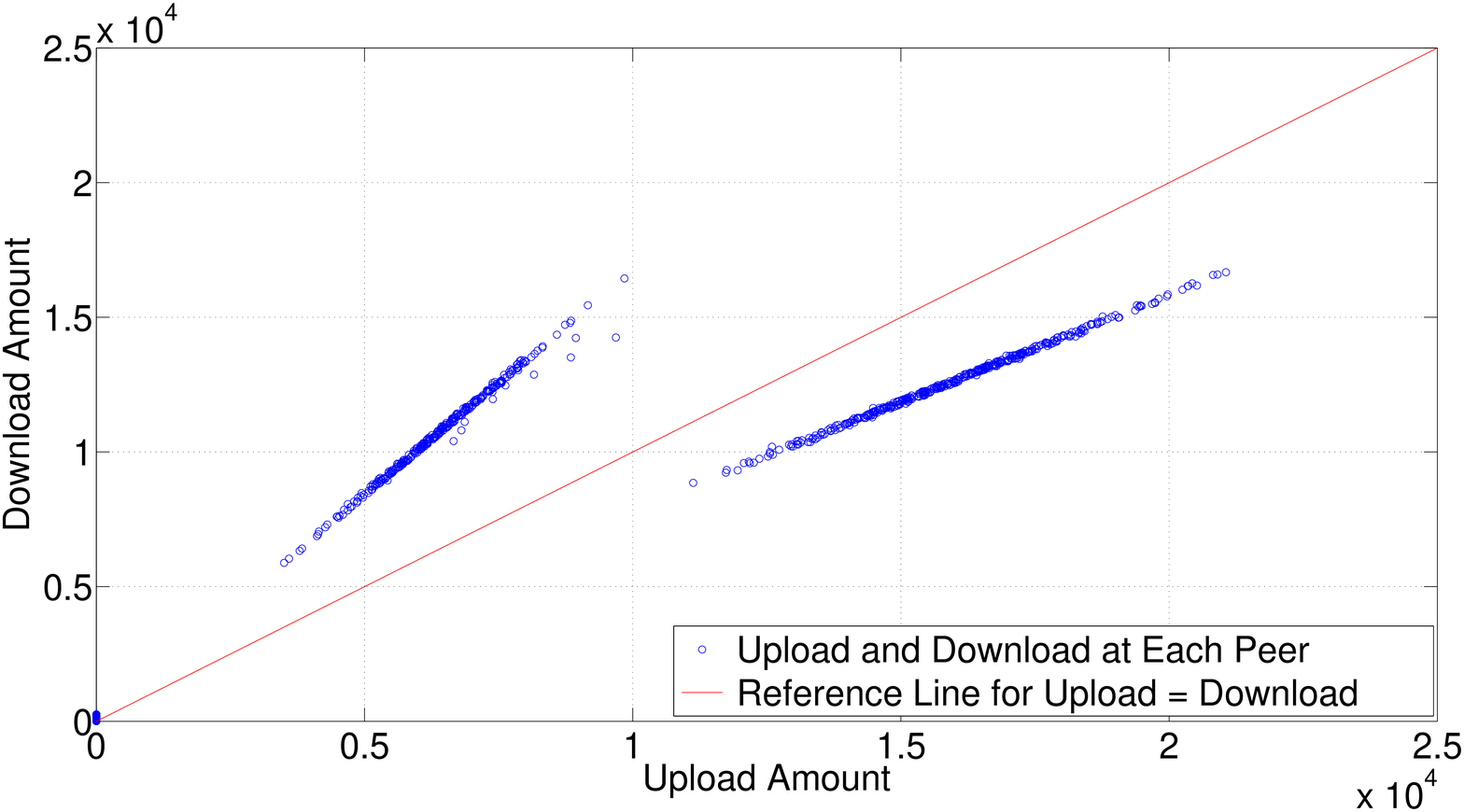}  }
\subfloat[Free Riders = 60 \%, $\alpha=0.3$]{ \includegraphics[width=6cm,height=4cm,scale=.18]{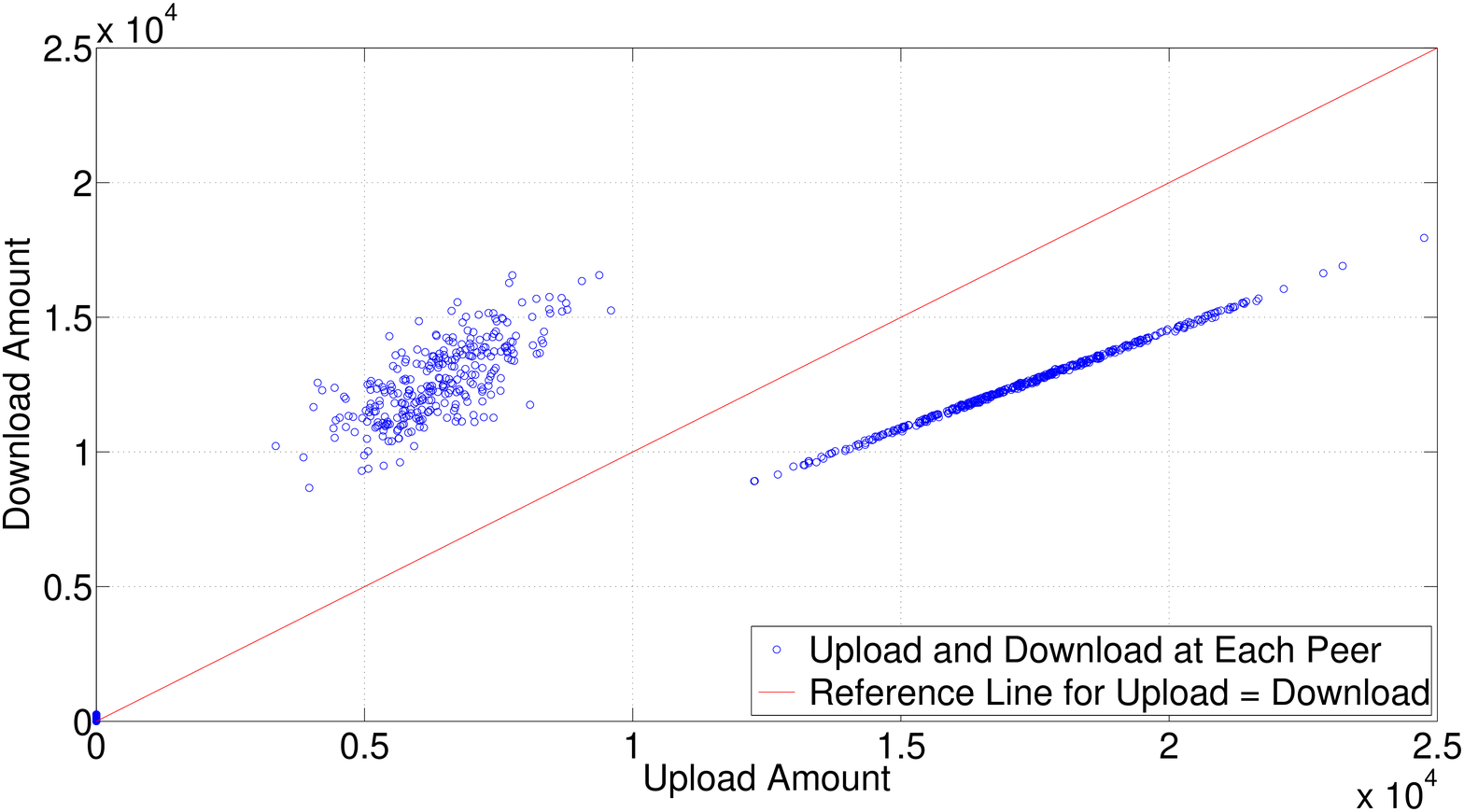}  }
\caption{Upload and Download Amount at Each Peer for SBCI in Adaptive  Model for  Simple Procedure of Peer Selection. Free-riders are varied from $20\% - 60\%$ and the value of $\alpha$ is taken as $0.9, 0.6$ and 0.3}\label{sixty}
\end{figure*}
\begin{table}
\begin{center}
\caption{ AAD and $\%$ of Rejections for SBCI in Adaptive Model for Simple Procedure of Peer Selection}\label{table4.2}
\begin{tabular}{ | m{0.5cm} | m{1em}| m{5em}|m{5em}| m{5.5em}|} 
 \hline
  S.N.& $\alpha$ &Free-riders & AAD &$\%$ of Rejections \\[1ex] 
 \hline
 $1$& 0.9  & 20\% &  0.118469&1.87\\
    \hline
 $2$& 0.9  & 40\% &  0.236766&0.606\\
    \hline
 $3$& 0.9  & 60\% &  0.354868&0.198\\
    \hline
 $4$& 0.6  & 20\%&  0.175055&0.054\\
    \hline
 $5$& 0.6  & 40\% & 0.351099&0.028\\
    \hline
 $6$& 0.6  & 60\% &  0.527092&0.011\\
    \hline
 $7$& 0.3  & 20\% &  0.200618&0.020\\
    \hline
 $8$& 0.3  & 40\% &  0.401786&0.011\\
    \hline
 $9$& 0.3  & 60\% &  0.600657&0.008\\
  \hline
\end{tabular}
\end{center}
\end{table}
In Adaptive Model, free-riders earn the SBCI and thereafter use this SBCI to download maximum resources from the network.  Simulation results for this model are shown in Fig. \ref{sixty}. Corresponding $AAD$ and percentage of rejections among cooperative peers are shown in Table \ref{table4.2}.  We can observe from this figure that  for higher $\alpha$, algorithm performs better. For $\alpha=0.9$, even in the presence of a large number of free-riders, the algorithm is able to balance the upload and download amount in the network. We can also observe from Table \ref{table4.2} that for higher $\alpha$ the percentage of rejection among cooperative peers is higher but corresponding $AAD$ is very less. Thus, impact of $\alpha$ is clearly evident.
\par And finally, we conducted the simulation for SBCI in Extreme Model. Results for upload and download at each peer are  shown in Fig. \ref{eighty}. Corresponding $AAD$ and percentage of rejections among cooperative peers are shown in Table \ref{table4.3}. We can observe from the figure that for $\alpha=0.9$ the algorithm is able to balance the upload and download amount in the network. For $\alpha=0.9$, at the cost of less than 2\% of rejections among the  cooperative peers, algorithm is able to maintain $AAD$ as 0.211228.

\begin{figure*}
\centering 
  \subfloat[Free Riders = 80\%, $\alpha=0.9$]{ \includegraphics[width=6cm,height=4cm,scale=.18]{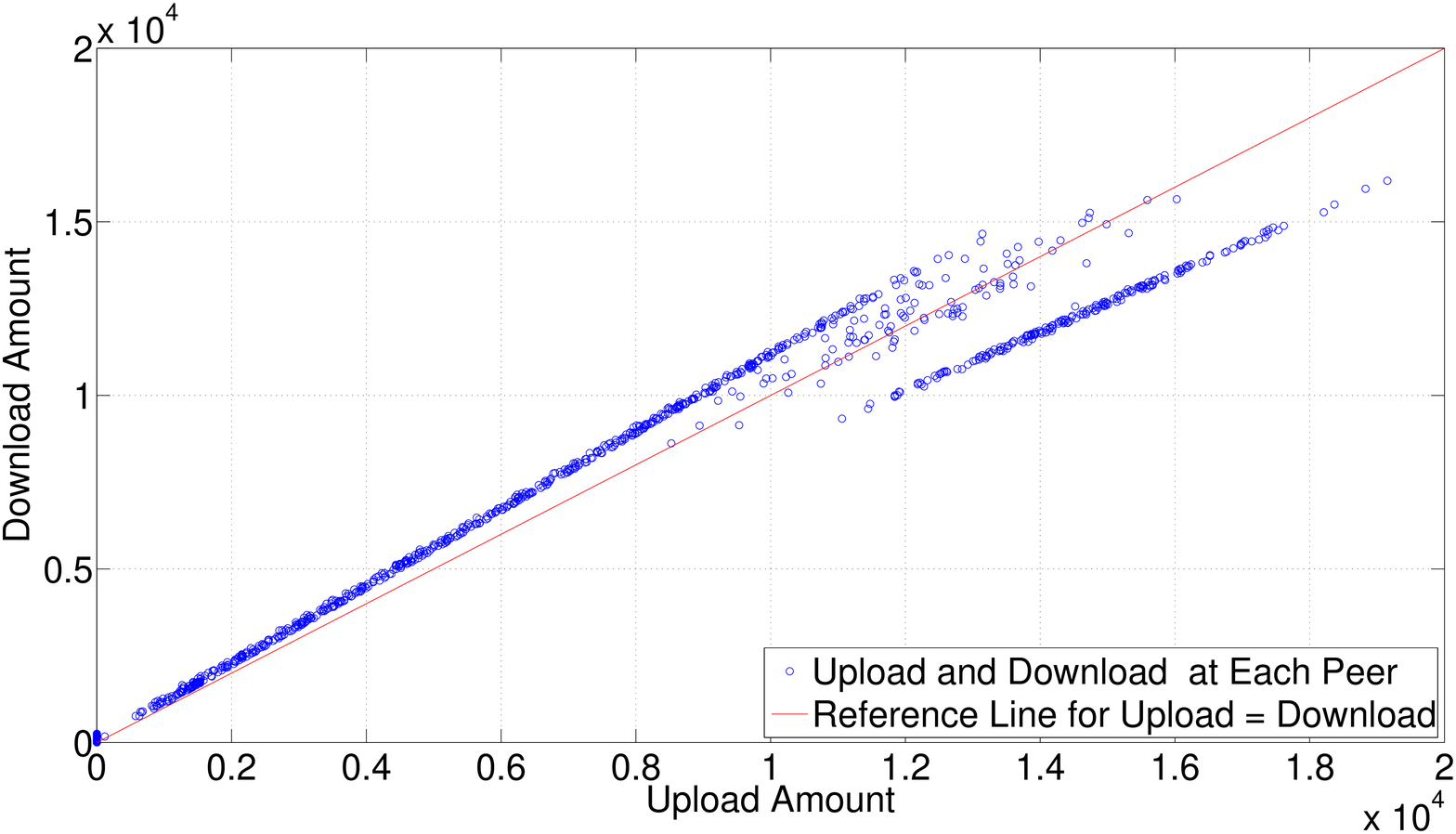}  }
\subfloat[Free Riders = 80\%, $\alpha=0.6$]{ \includegraphics[width=6cm,height=4cm,scale=.18]{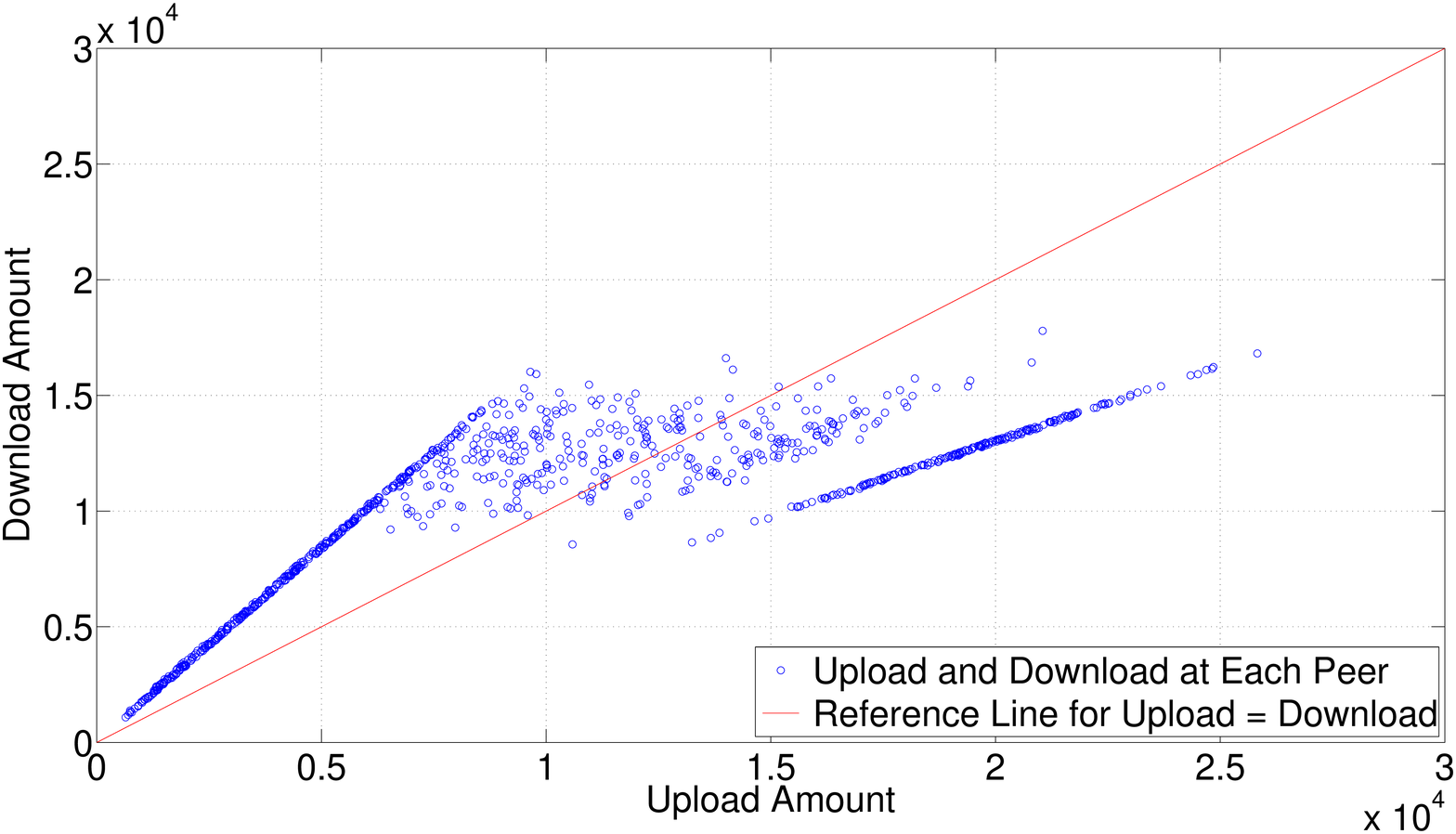}  }
\subfloat[Free Riders = 80\%, $\alpha=0.3$]{ \includegraphics[width=6cm,height=4cm,scale=.18]{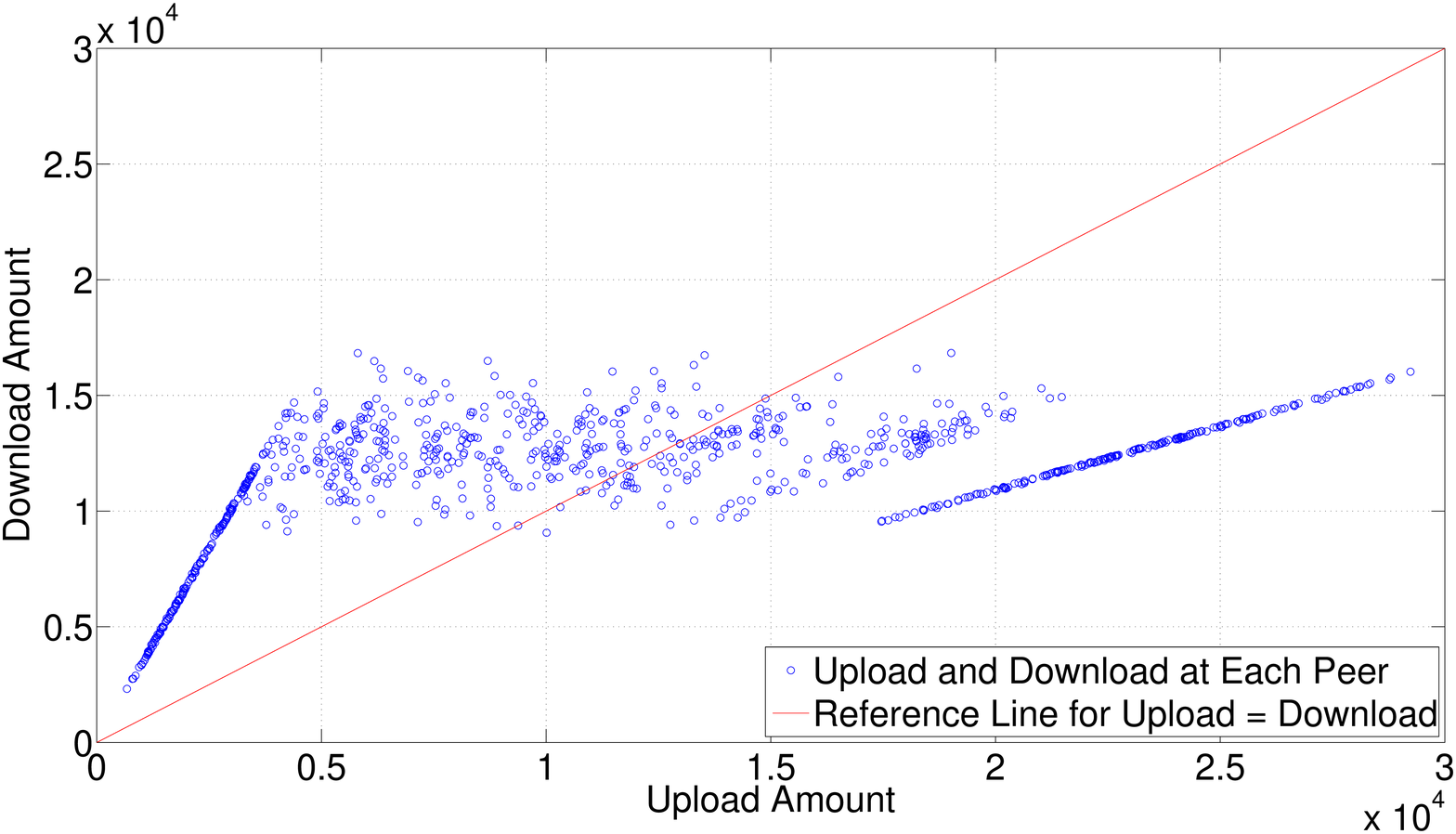}  }
\caption{Upload and Download Amount at Each Peer for SBCI in Extream Model for Simple Procedure of Peer Selection. At the beginning of simulation, 10\% peers are free-riders. After completion of every 12.5\% of total transactions, 10\% more peers convert themselves to free-riders. The value of $\alpha$ is taken as $0.9, 0.6$ and 0.3}\label{eighty}
\end{figure*}

\begin{table}
\begin{center}
\caption{ AAD and $\%$ of Rejections for SBCI in Simple Model for Simple Procedure of Peer Selection }\label{table4.3}
\begin{tabular}{ | m{0.5cm} | m{1em}| m{5em}|m{5em}| m{5.5em}|} 
 \hline
  S.N.& $\alpha$ &Free-riders & AAD &$\%$ of Rejections \\[1ex] 
 \hline
 $1$& 0.9  & 80\% &  0.211228&1.794\\
    \hline
 $2$& 0.6  & 80\% &  0.423116&0.068\\
    \hline
 $3$& 0.3  & 80\% &  0.567303&0.010\\
    \hline
\end{tabular}
\end{center}
\end{table}
We also reported the simulation results of GC for all peer distribution models in Fig. \ref{GC}. Corresponding $AAD$ and percentage of rejections among cooperative peers are reported in Table \ref{table4.6}. We can see from the figure that GC can also balance the upload and download amounts in each peer. In Adaptive Model and in Extreme Model GC can maintain better fairness compared to SBCI but the percentage of rejections among cooperative peers are higher in GC for all the models. Thus, it is less efficient compared to SBCI.   \par

\begin{figure*}
\centering 
  \subfloat[Simple Model,  Free Riders = 30\%]{ \includegraphics[width=6cm,height=4cm,scale=.18]{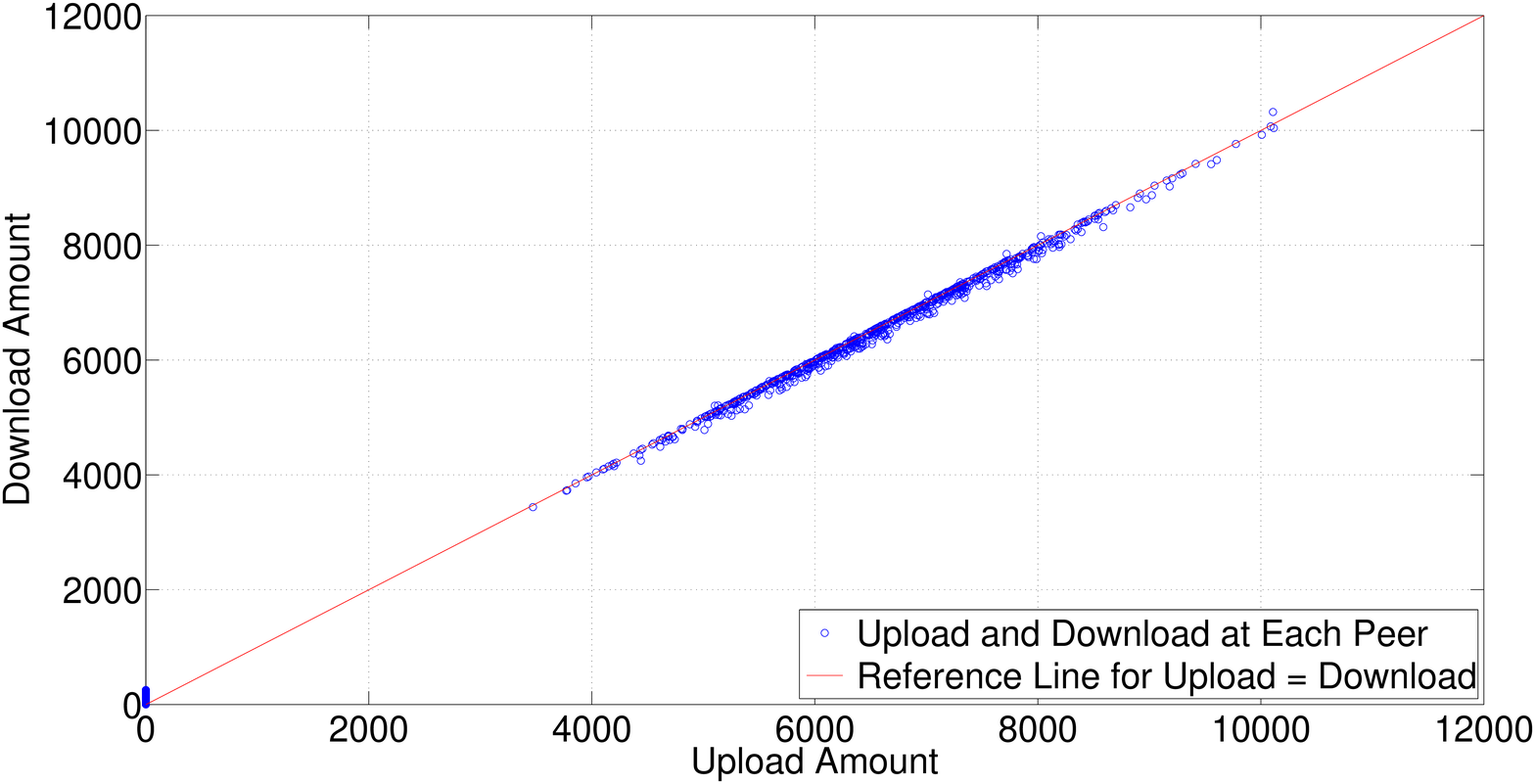}  }
\subfloat[Adaptive Model, Free Riders = 60\%]{ \includegraphics[width=6cm,height=4cm,scale=.18]{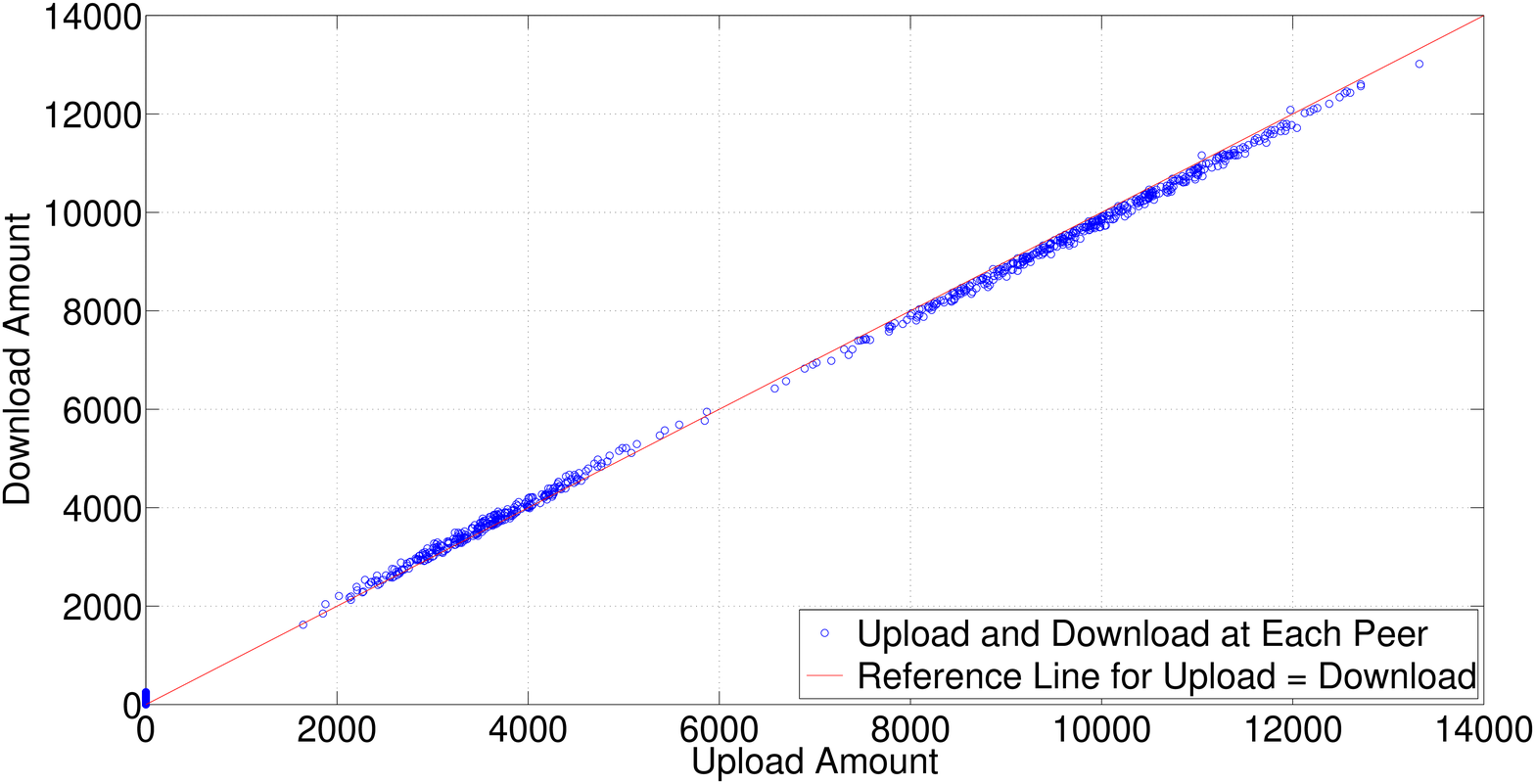}  }
\subfloat[Extreme Model, Free Riders = 80\%]{ \includegraphics[width=6cm,height=4cm,scale=.18]{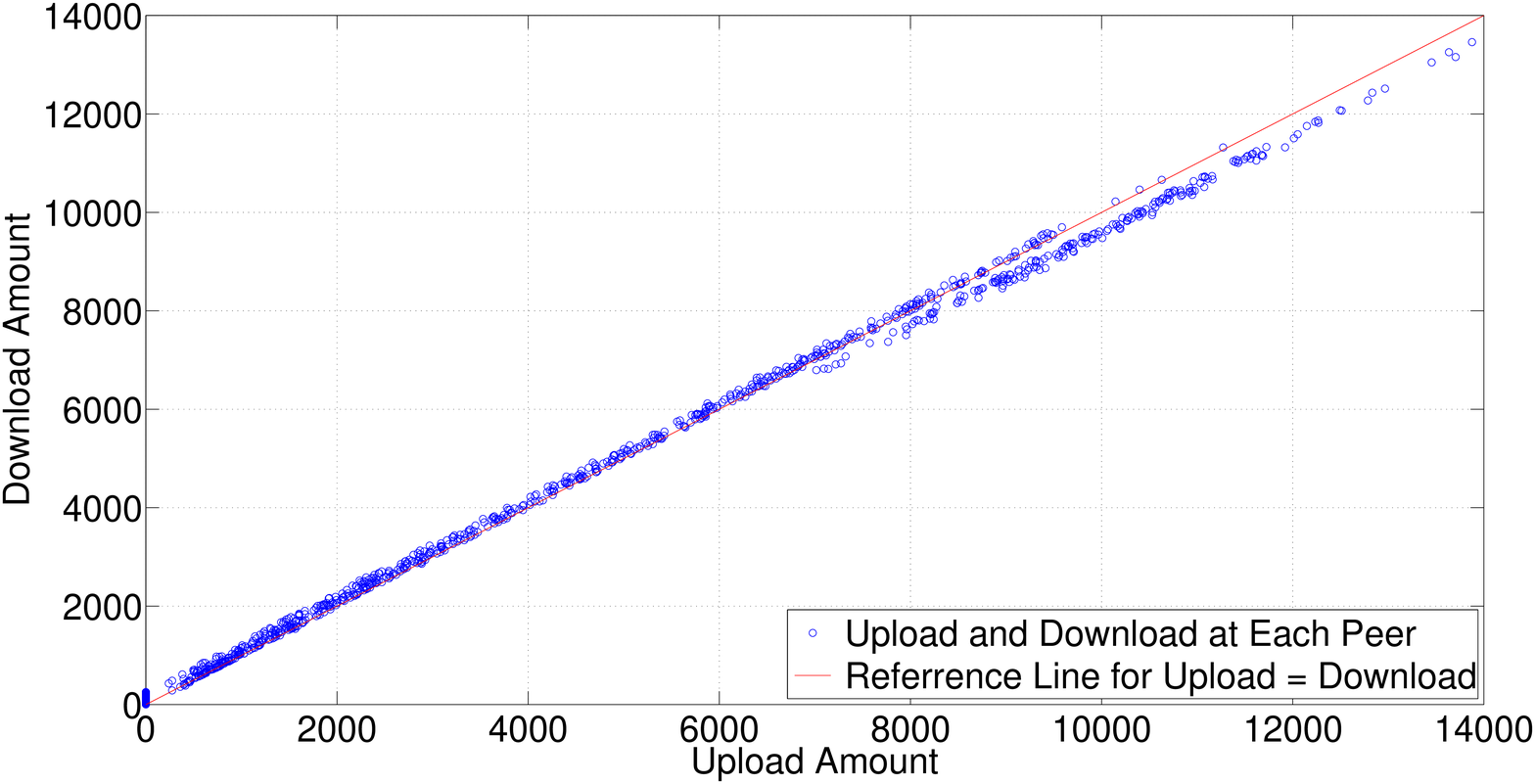}  }
\caption{Upload and Download Amount at Each Peer for best case of GC, i.e., $\alpha=0.8, \beta=0.2$ in different distribution models for Simple Procedure of Peer Selection.}\label{GC}
\end{figure*}

\begin{table}
\begin{center}
\caption{ AAD and $\%$ of Rejections for Different Distribution Model in Best case of GC for Simple Procedure of Peer Selection}\label{table4.6}
\begin{tabular}{ | m{0.5cm} | m{4em}| m{4.5em}|m{2.5em}| m{5.5em}|} 
 \hline
  S.N.& Model &Free-riders & AAD &$\%$ of Rejections \\[1ex] 
 \hline
 $1$& Simple  & 30\% &  0.3059& 33.657 \\
    \hline
 $2$& Adaptive  & 60\% &  0.3146& 16.27 \\
    \hline
 $3$& Extreme  & 80\% &  0.1370& 19.269 \\
    \hline
\end{tabular}
\end{center}
\end{table}

\begin{figure*}
\centering 
  \subfloat[Free Riders = 10\%, $\alpha=0.9$]{ \includegraphics[width=6cm,height=4cm,scale=.18]{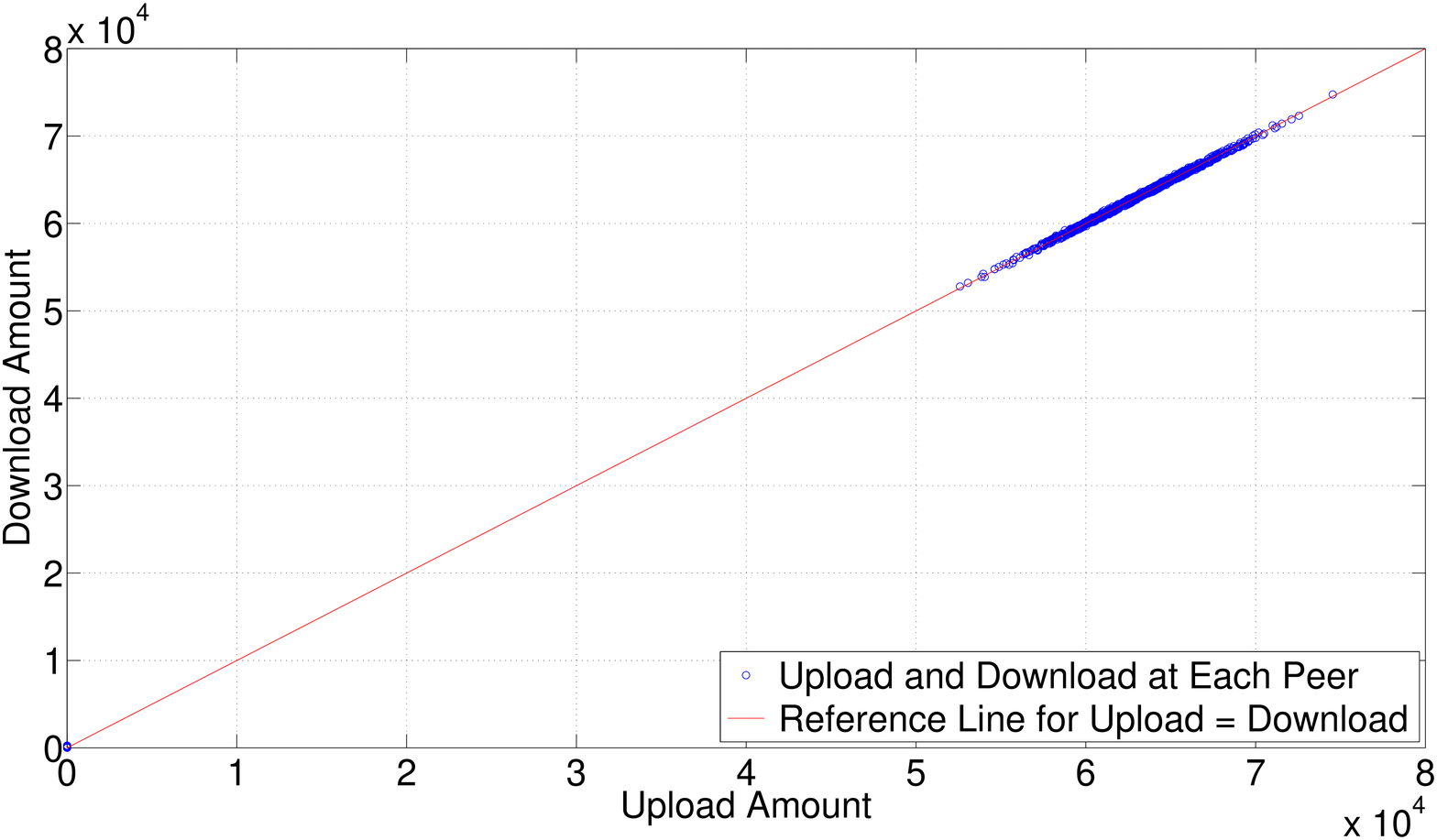}  }
\subfloat[Free Riders = 10\%, $\alpha=0.6$]{
\includegraphics[width=6cm,height=4cm,scale=.18]
{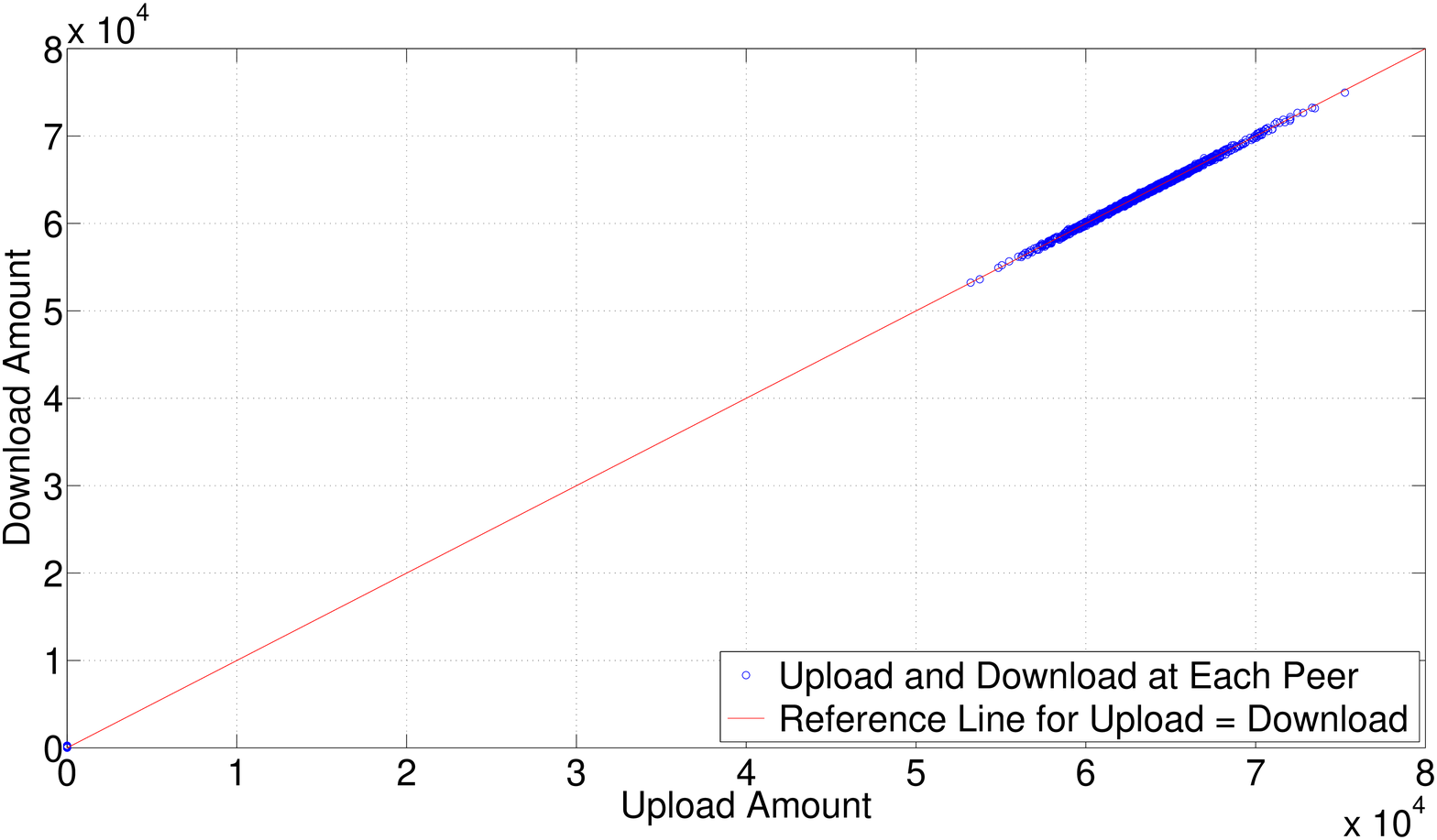} }
\subfloat[Free Riders = 10\%, $\alpha=0.3$]{
\includegraphics[width=6cm,height=4cm,scale=.18]
{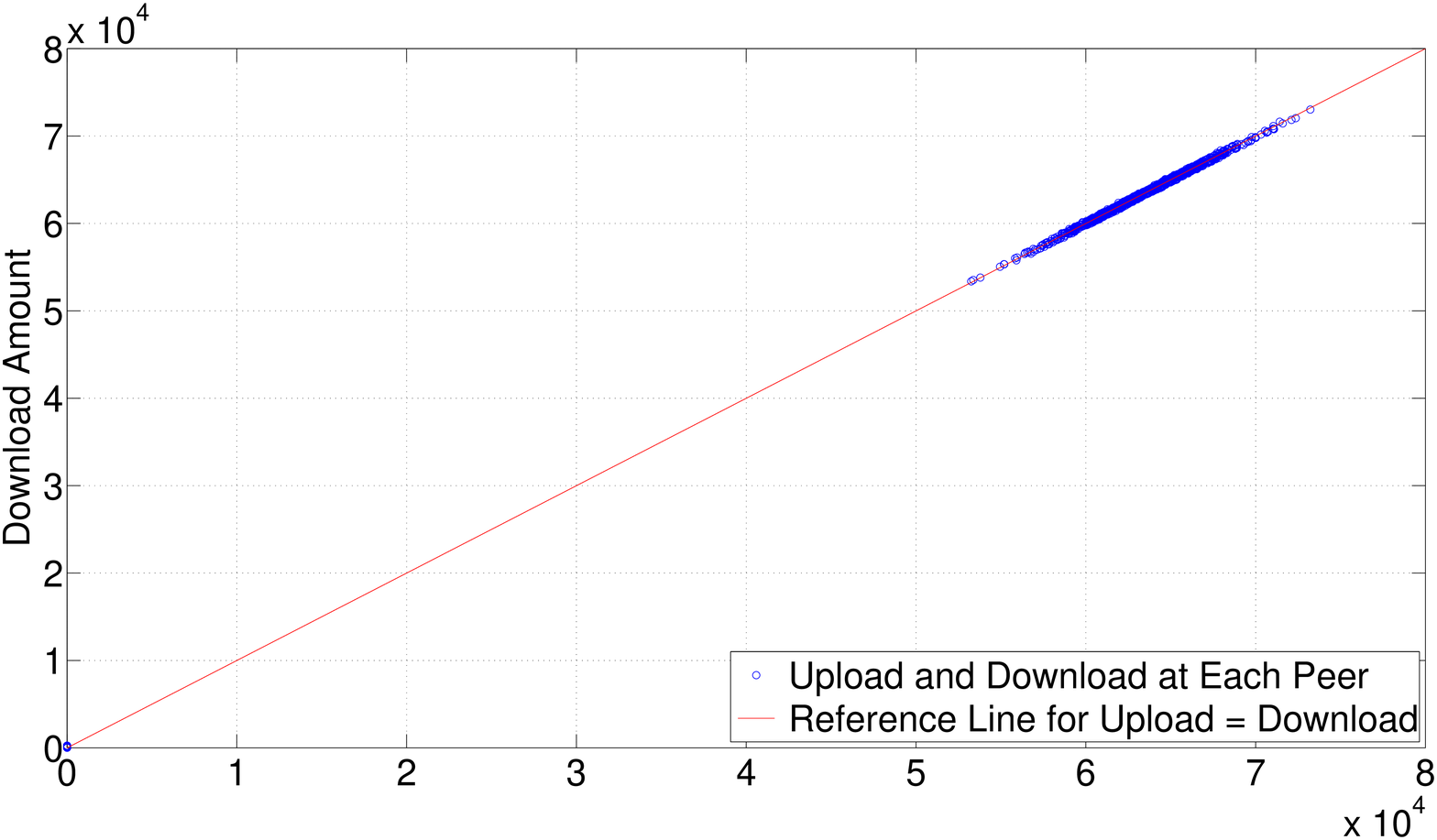}}\\
\subfloat[Free Riders = 30\%, $\alpha=0.9$]{
\includegraphics[width=6cm,height=4cm,scale=.18]
{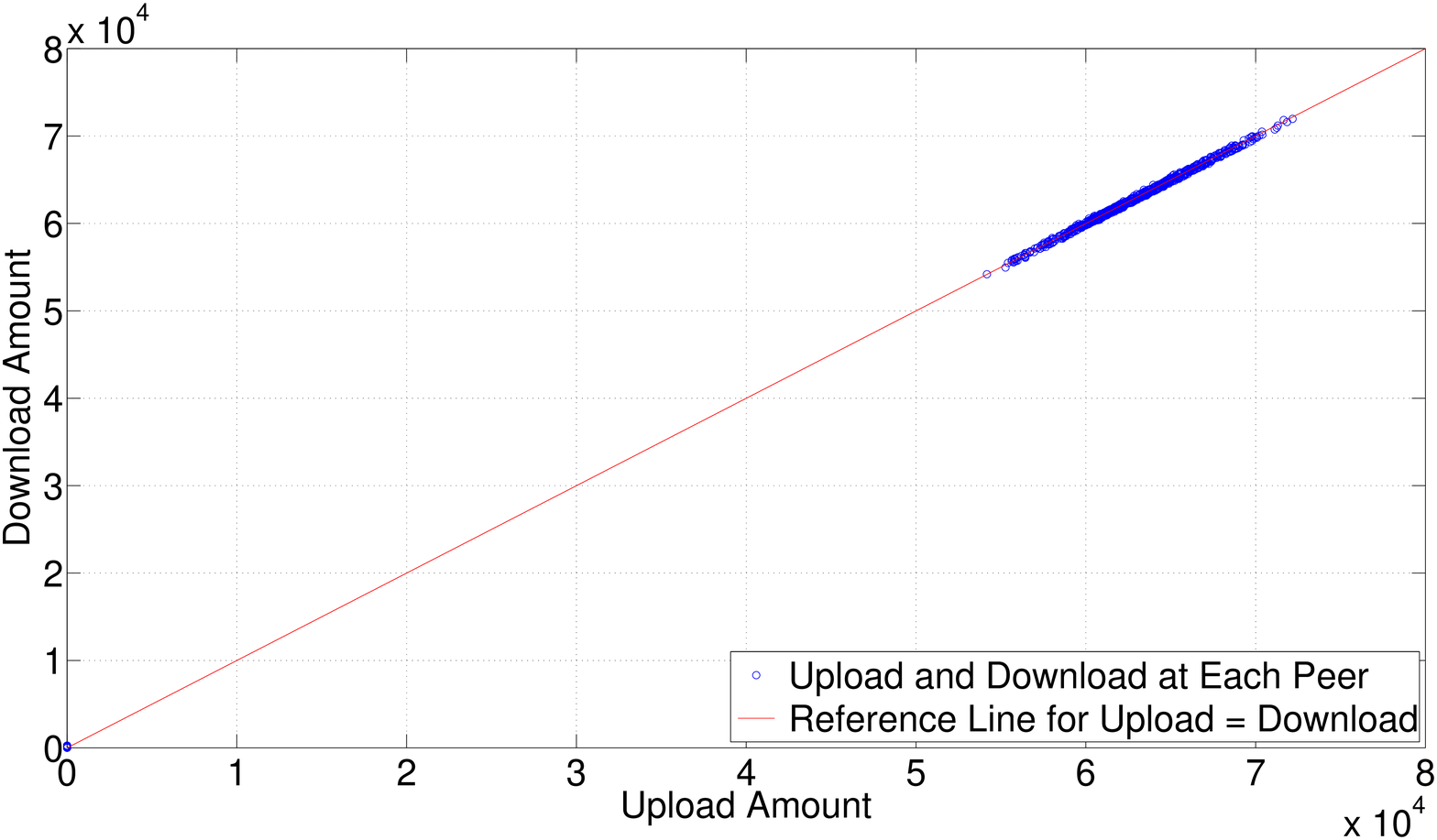} }
\subfloat[Free Riders = 30\%, $\alpha=0.6$]{
\includegraphics[width=6cm,height=4cm,scale=.18]
{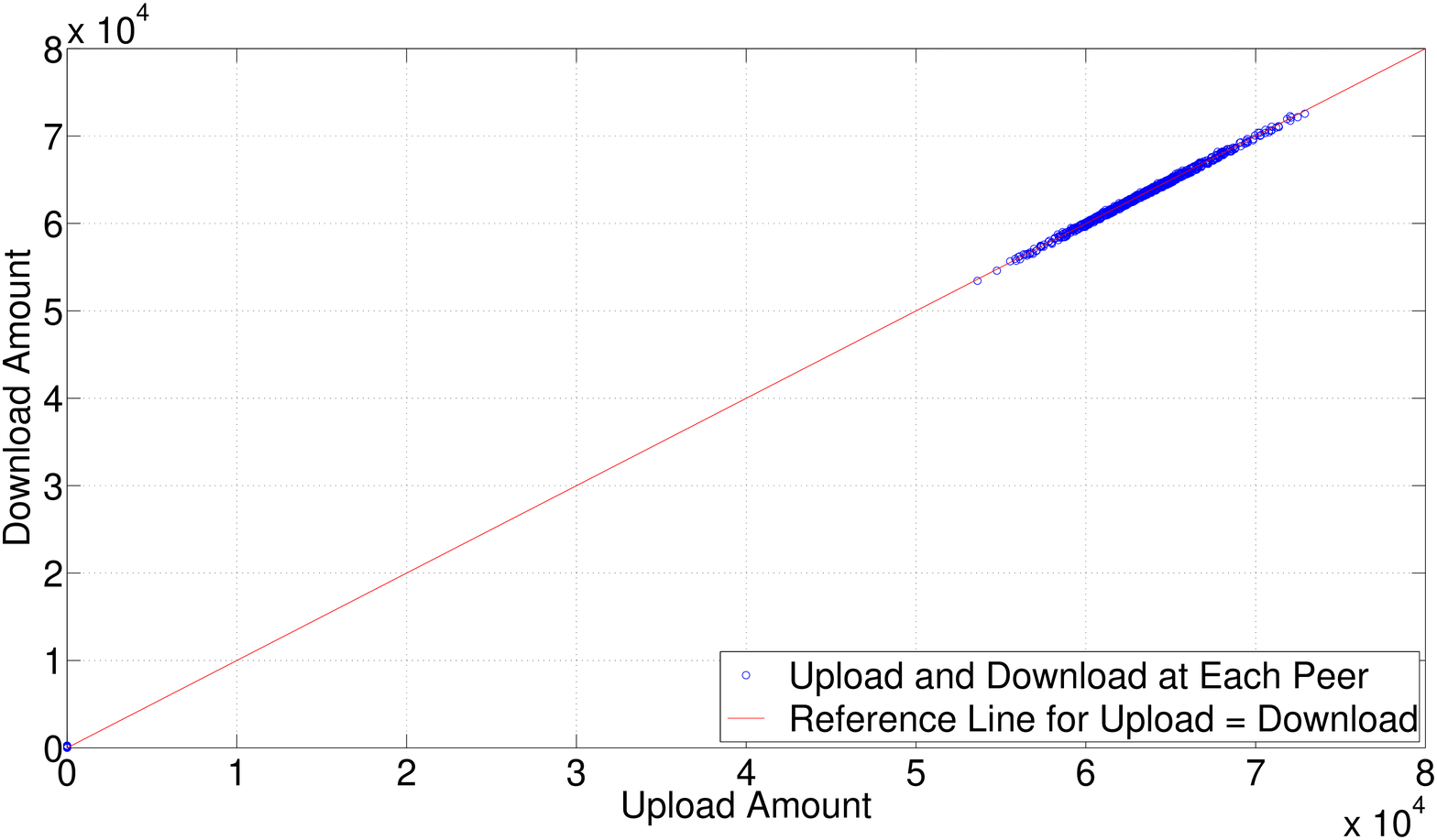}}
\subfloat[Free Riders = 30\%, $\alpha=0.3$]{
\includegraphics[width=6cm,height=4cm,scale=.18]
{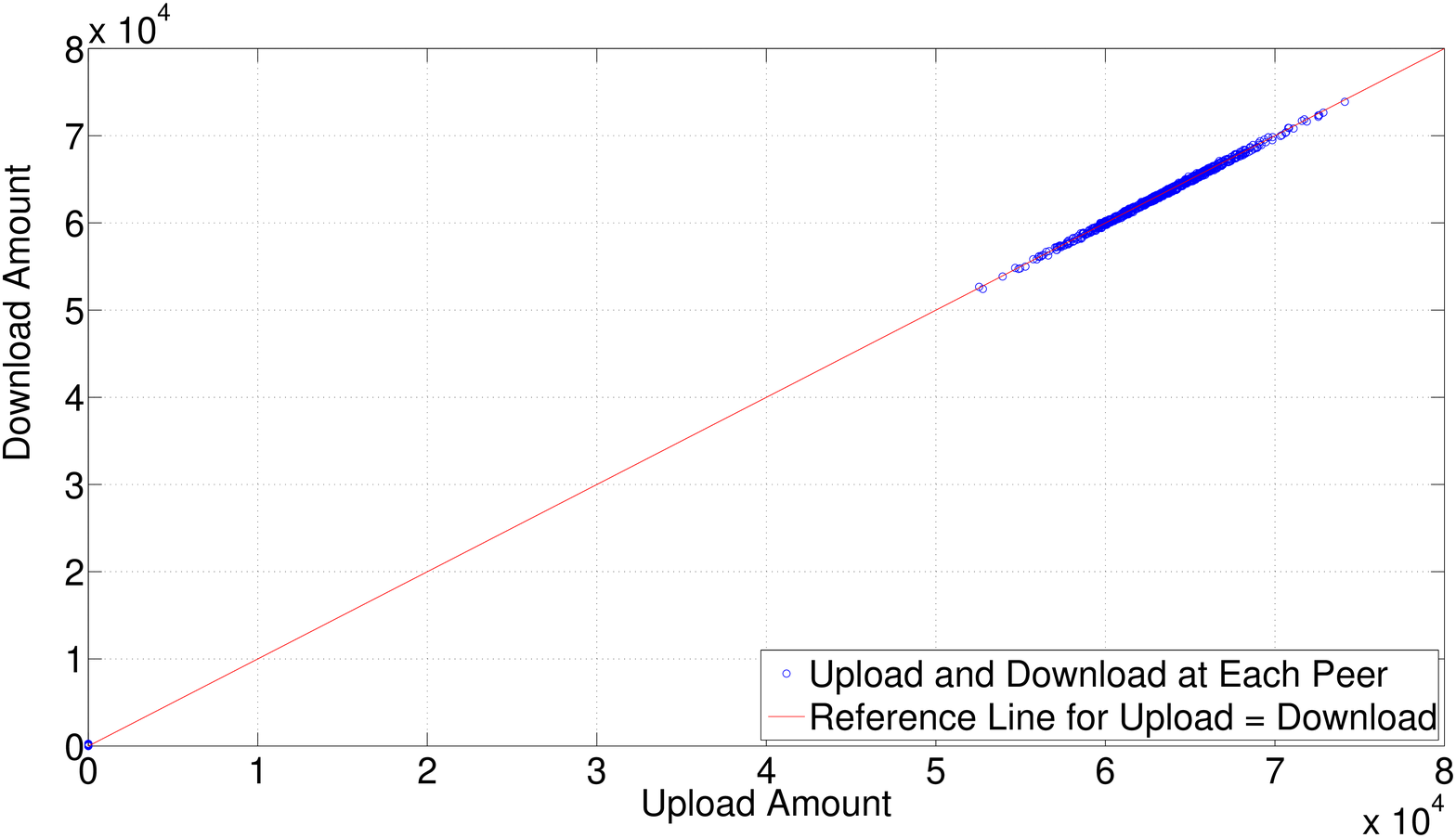} }\\
\subfloat[Free Riders = 50\%, $\alpha=0.9$]{
\includegraphics[width=6cm,height=4cm,scale=.18]
{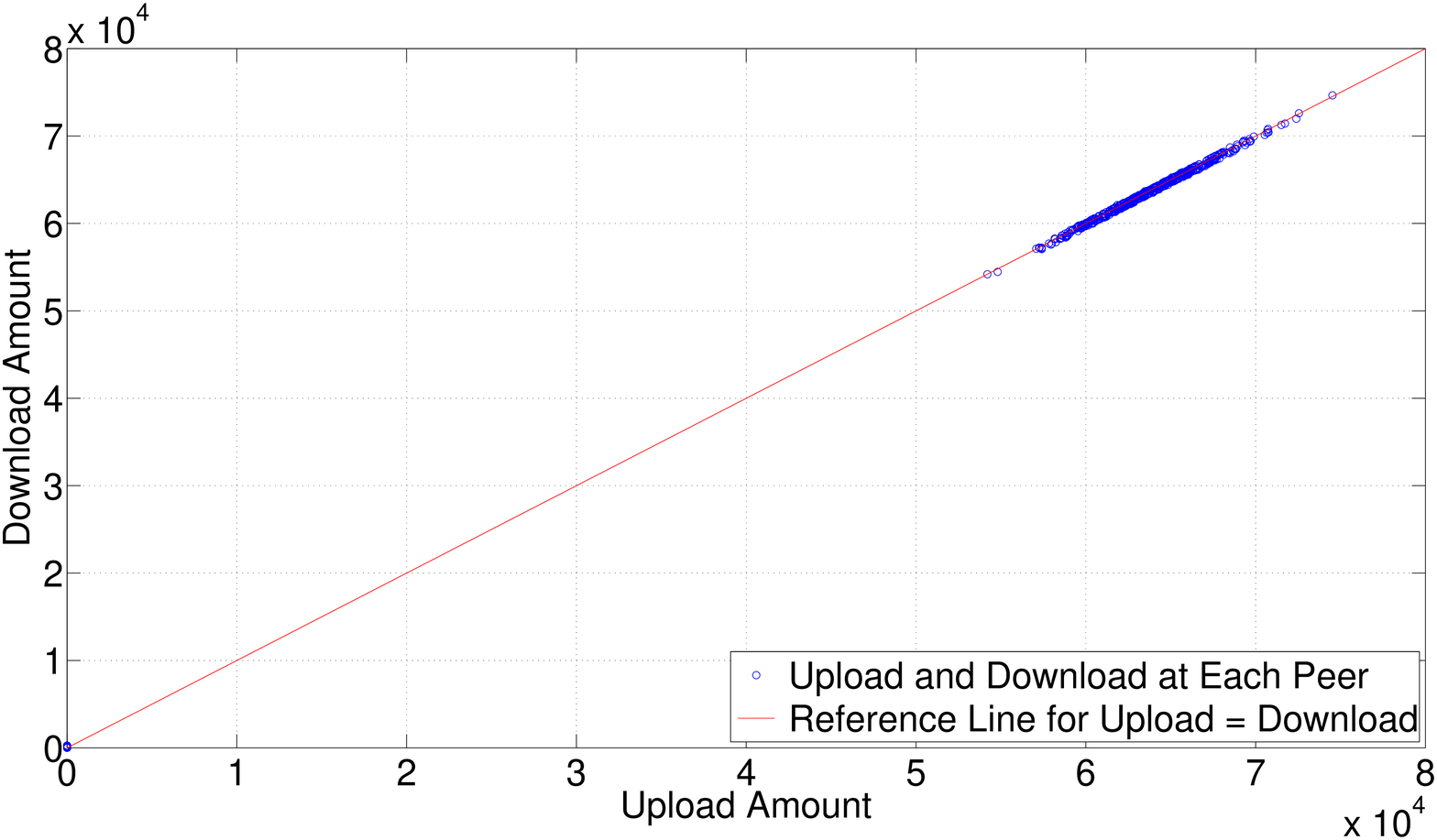}}
\subfloat[Free Riders = 50\%, $\alpha=0.6$]{
\includegraphics[width=6cm,height=4cm,scale=.18]
{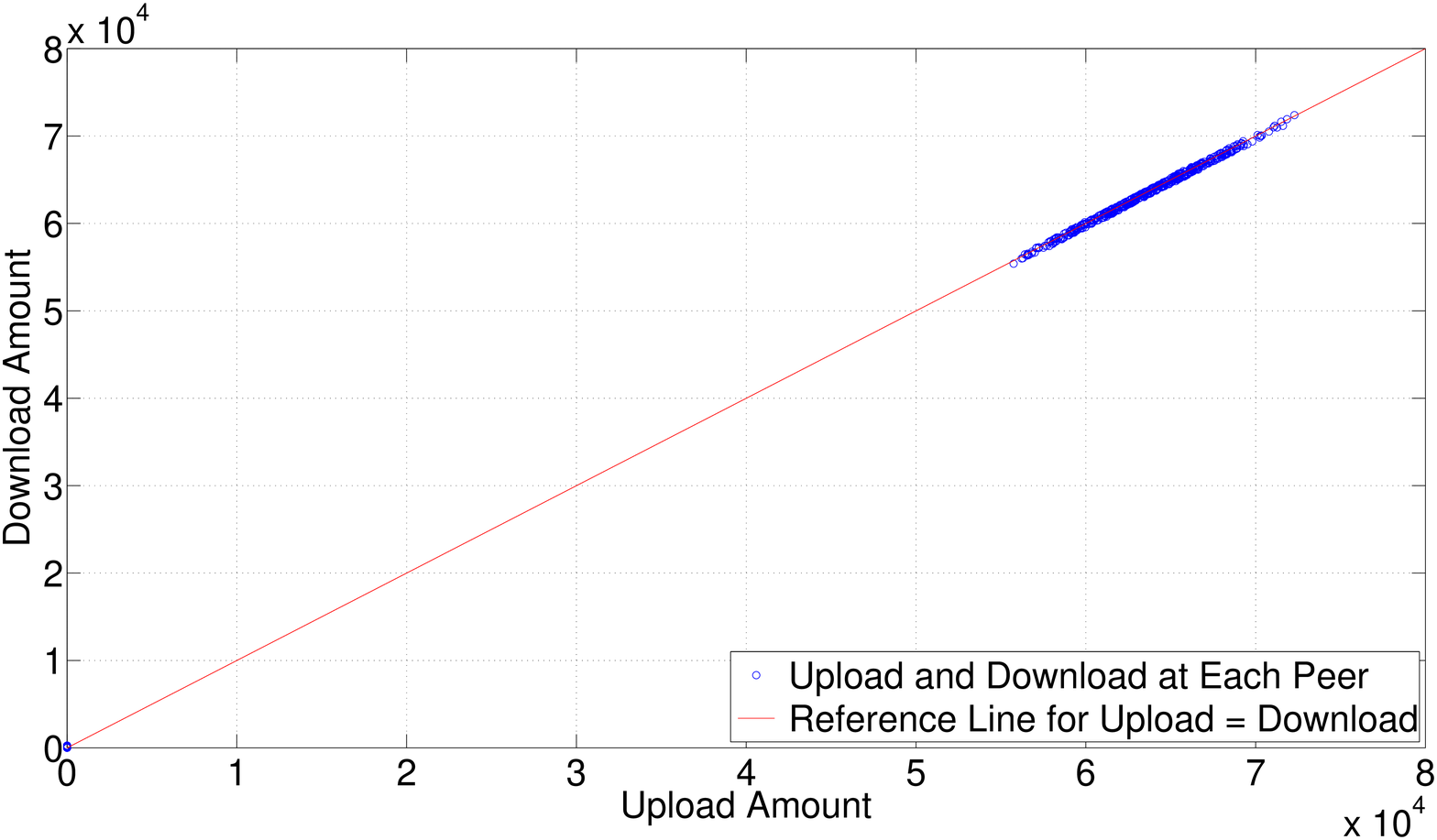}  }
\subfloat[Free Riders = 50\%, $\alpha=0.3$]{
\includegraphics[width=6cm,height=4cm,scale=.18]
{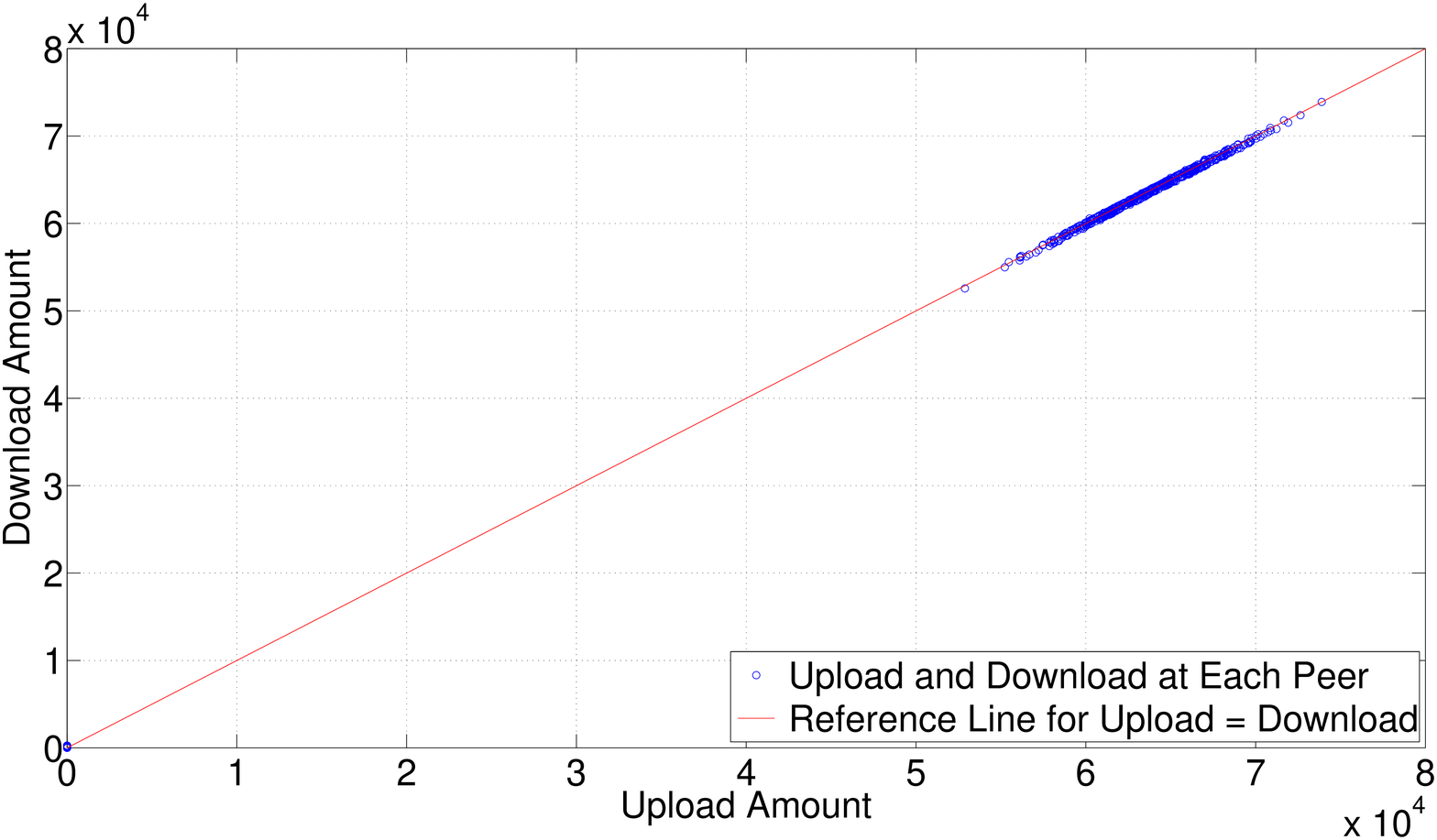}}\\
\subfloat[Free Riders = 70\%, $\alpha=0.9$]{
\includegraphics[width=6cm,height=4cm,scale=.18]
{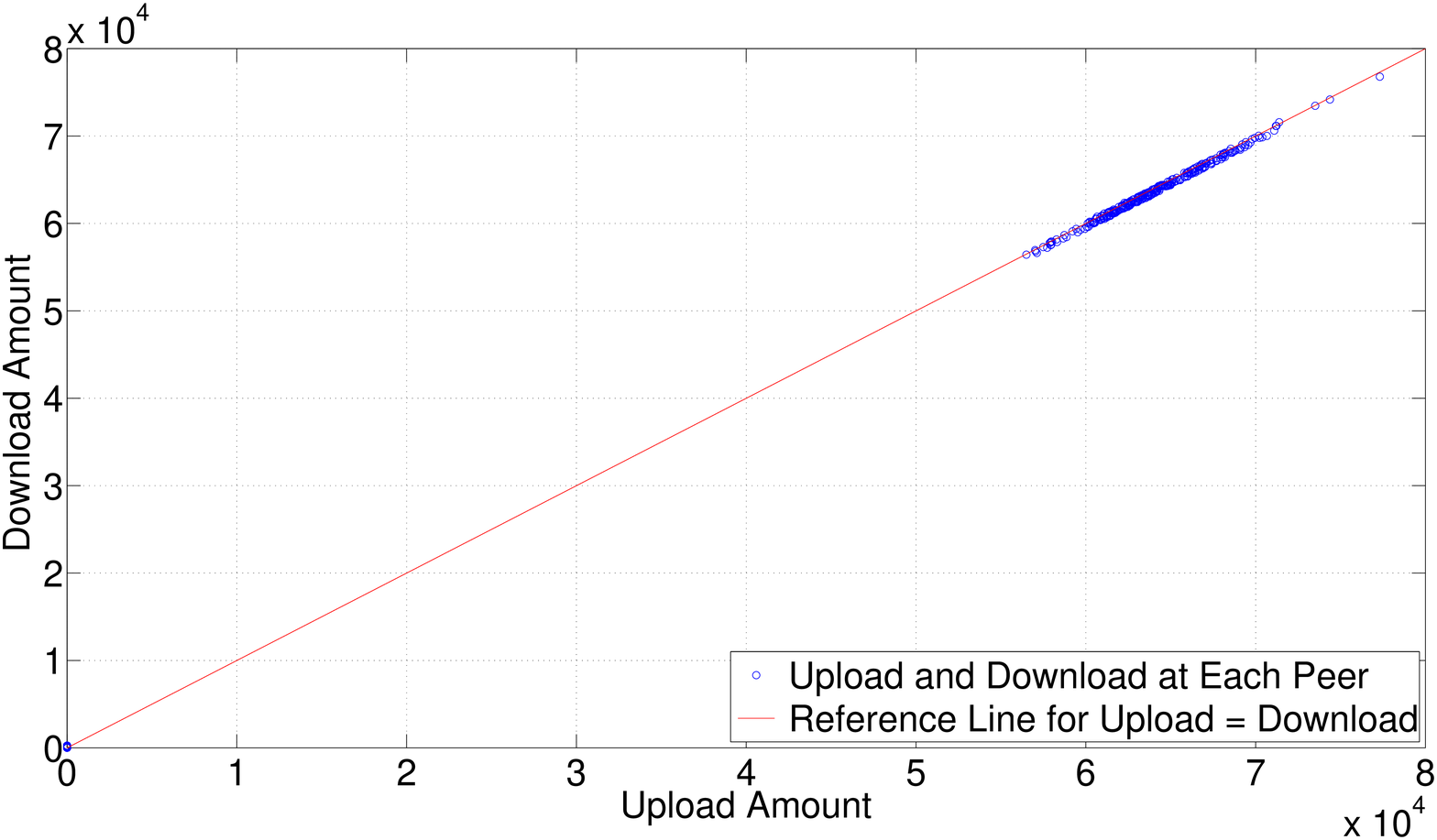} }
\subfloat[Free Riders = 70\%, $\alpha=0.6$]{
\includegraphics[width=6cm,height=4cm,scale=.18]
{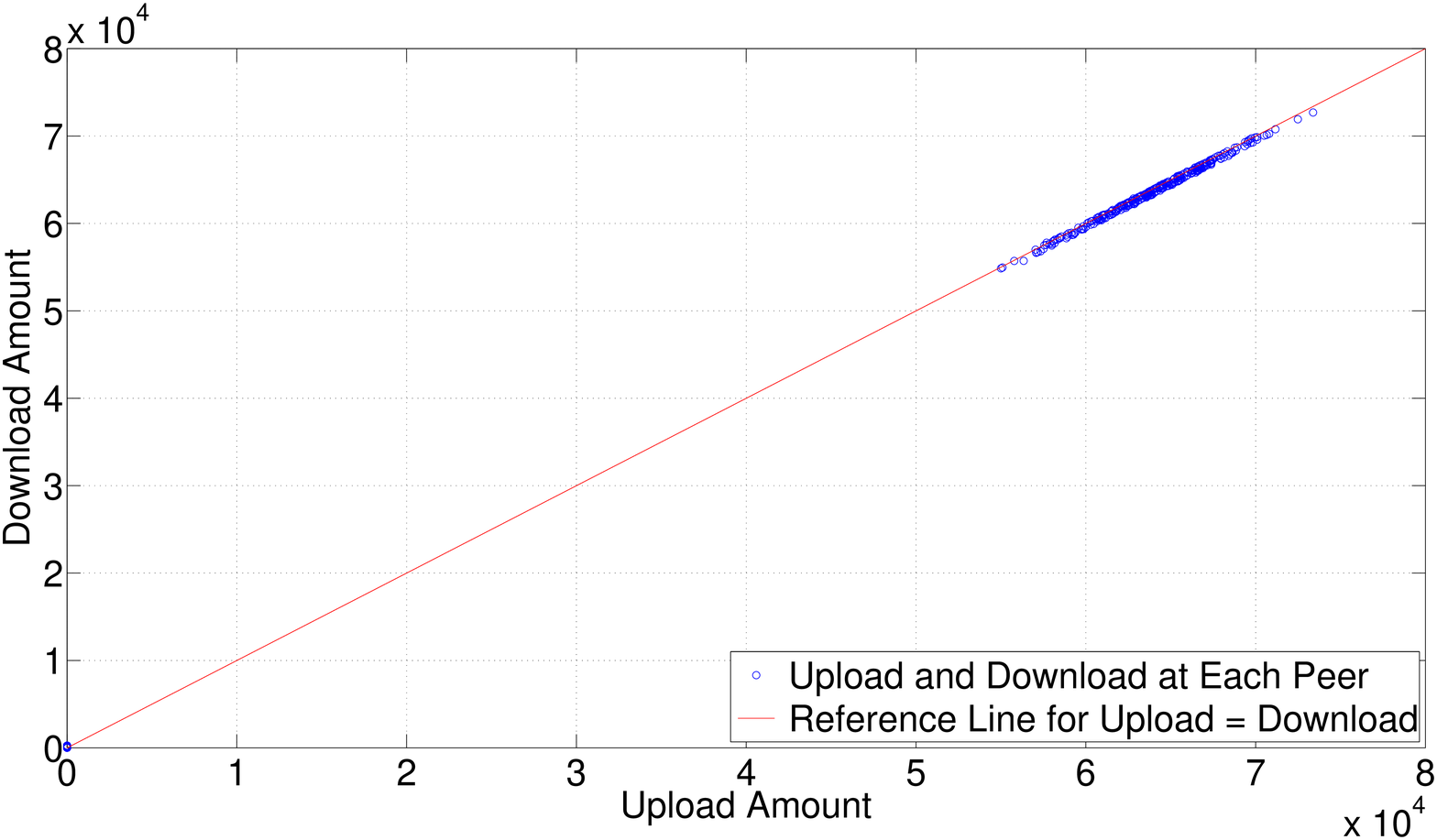}}
\subfloat[Free Riders = 70\%, $\alpha=0.3$]{
\includegraphics[width=6cm,height=4cm,scale=.18]
{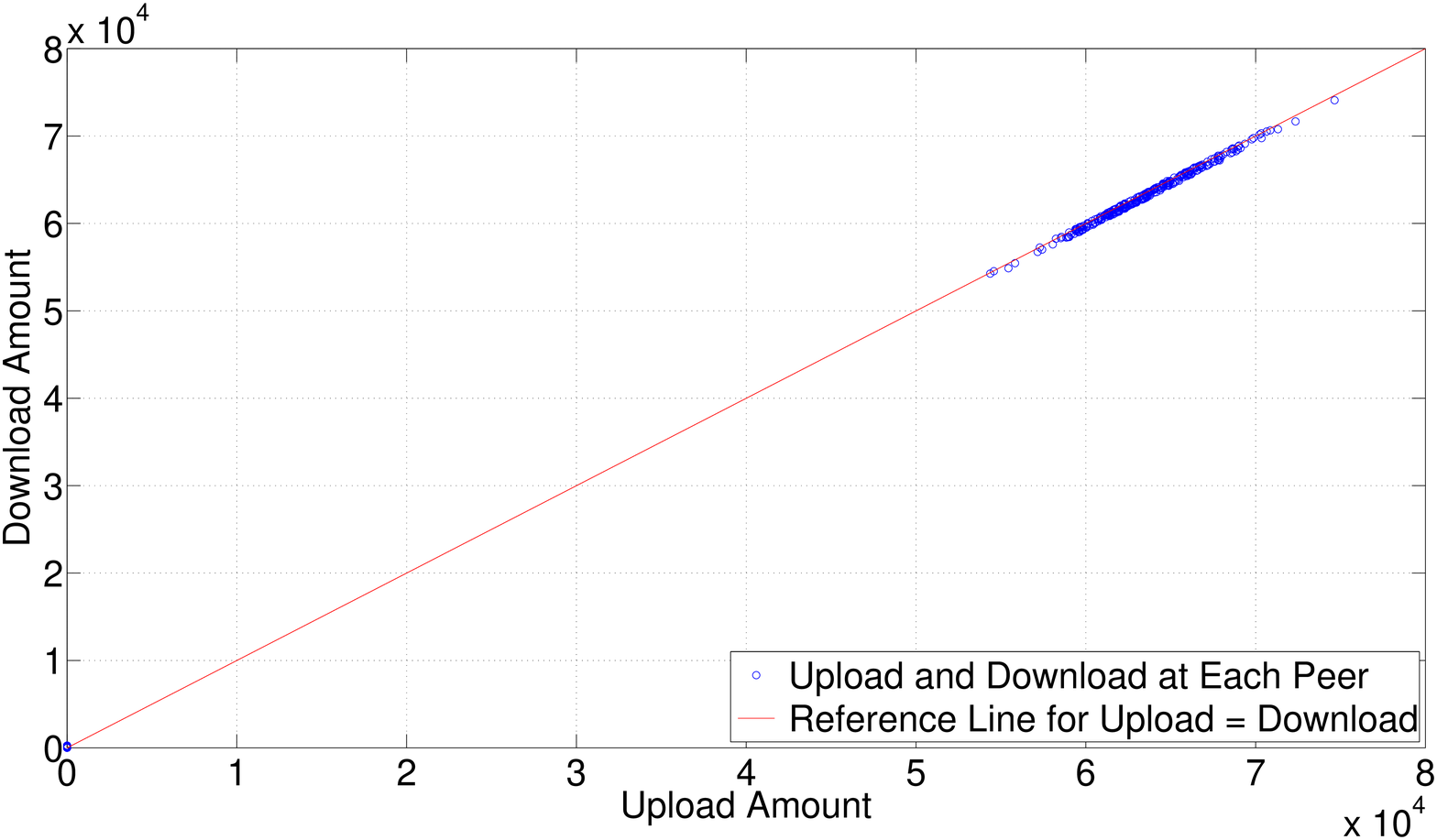} }
\caption{Upload and Download Amount at Each Peer in the Presence of $10\%-70\%$ Free-riders for Different values of $\alpha$. Peer selection approach is based on the problem of "College Admission and The Stability of Marriage". Bandwidth distribution is as type 1.}\label{sseventy}
\end{figure*}

\begin{table}
\begin{center}
\caption{ AAD and $\%$ of Rejections for SBCI in Simple Model for Stable merriage approach with two different bandwith peers }\label{table4.4}
\begin{tabular}{ | m{0.5cm} | m{1em}| m{5em}|m{5em}| m{5.5em}|} 
 \hline
  S.N.& $\alpha$ &Free-riders & AAD &$\%$ of Rejections \\[1ex] 
 \hline
 $1$& 0.9  & 10\% &  0.102265&1.0114\\
    \hline
 $2$& 0.9  & 30\% &  0.301927&0.287\\
    \hline
 $3$& 0.9  & 50\% &  0.501482&0.0244\\
    \hline
 $4$& 0.9  & 70\%&  0.701471&0.0026\\
    \hline
 $5$& 0.6  & 10\% & 0.102312&0.046\\
    \hline
 $6$& 0.6  & 30\% &  0.30196&0.0134\\
    \hline
 $7$& 0.6  & 50\% &  0.50149&0.0034\\
    \hline
 $8$& 0.6  & 70\% &  0.701432&0.0008\\
    \hline
 $9$& 0.3  & 10\% &  0.102367&0.0114\\
  \hline
  $10$& 0.3  & 30\% &0.301967&0.0036\\
    \hline
 $11$& 0.3  & 50\% &  0.501484&0.0016\\
    \hline
 $12$& 0.3  & 70\% &  0.701441&0\\
    \hline
\end{tabular}
\end{center}
\end{table}


\begin{figure*}
\centering 
  \subfloat[Free Riders = 10\%, $\alpha=0.9$]{ \includegraphics[width=6cm,height=4cm,scale=.18]{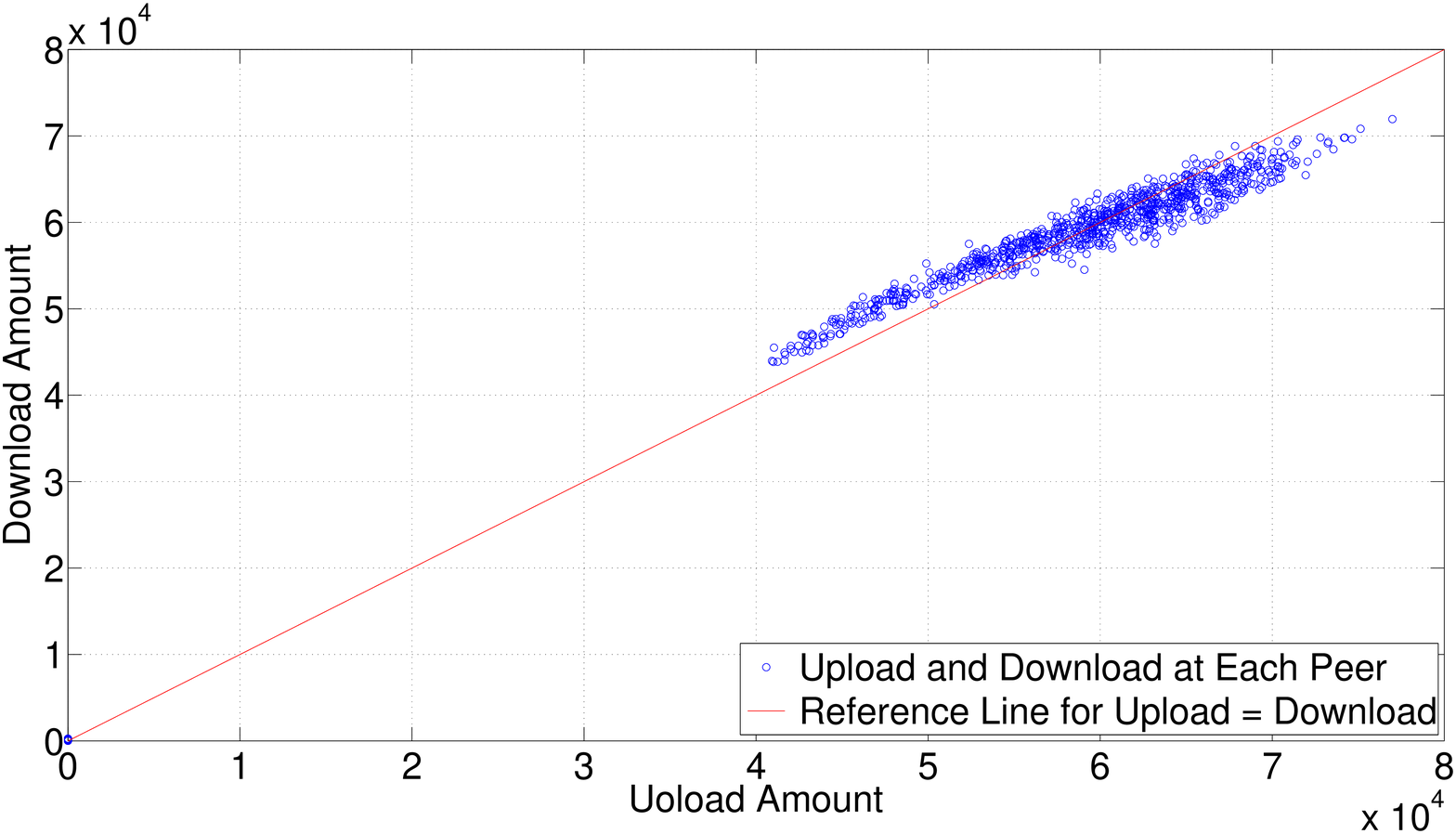}  }
\subfloat[Free Riders = 10\%, $\alpha=0.6$]{
\includegraphics[width=6cm,height=4cm,scale=.18]
{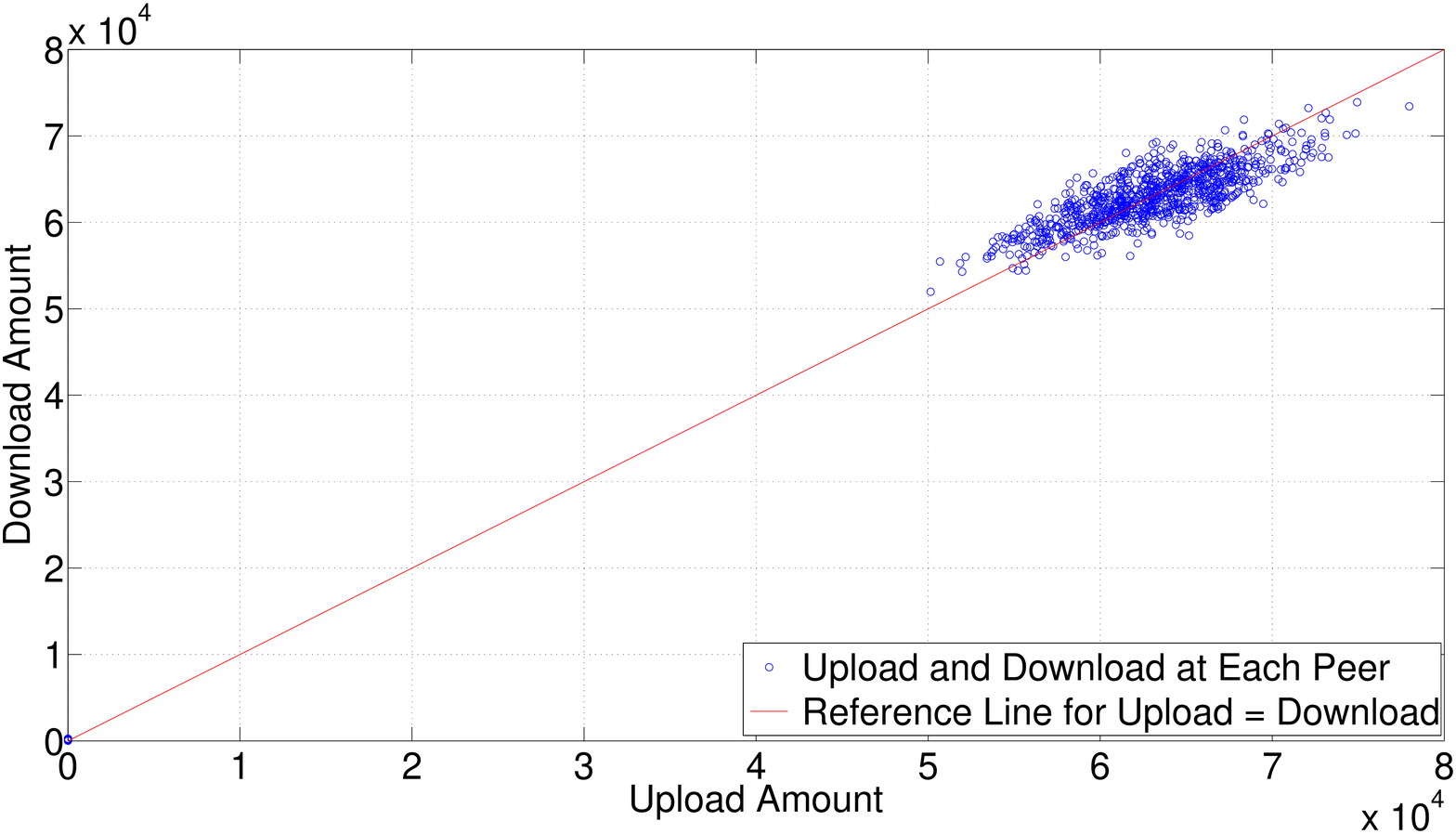} }
\subfloat[Free Riders = 10\%, $\alpha=0.3$]{
\includegraphics[width=6cm,height=4cm,scale=.18]
{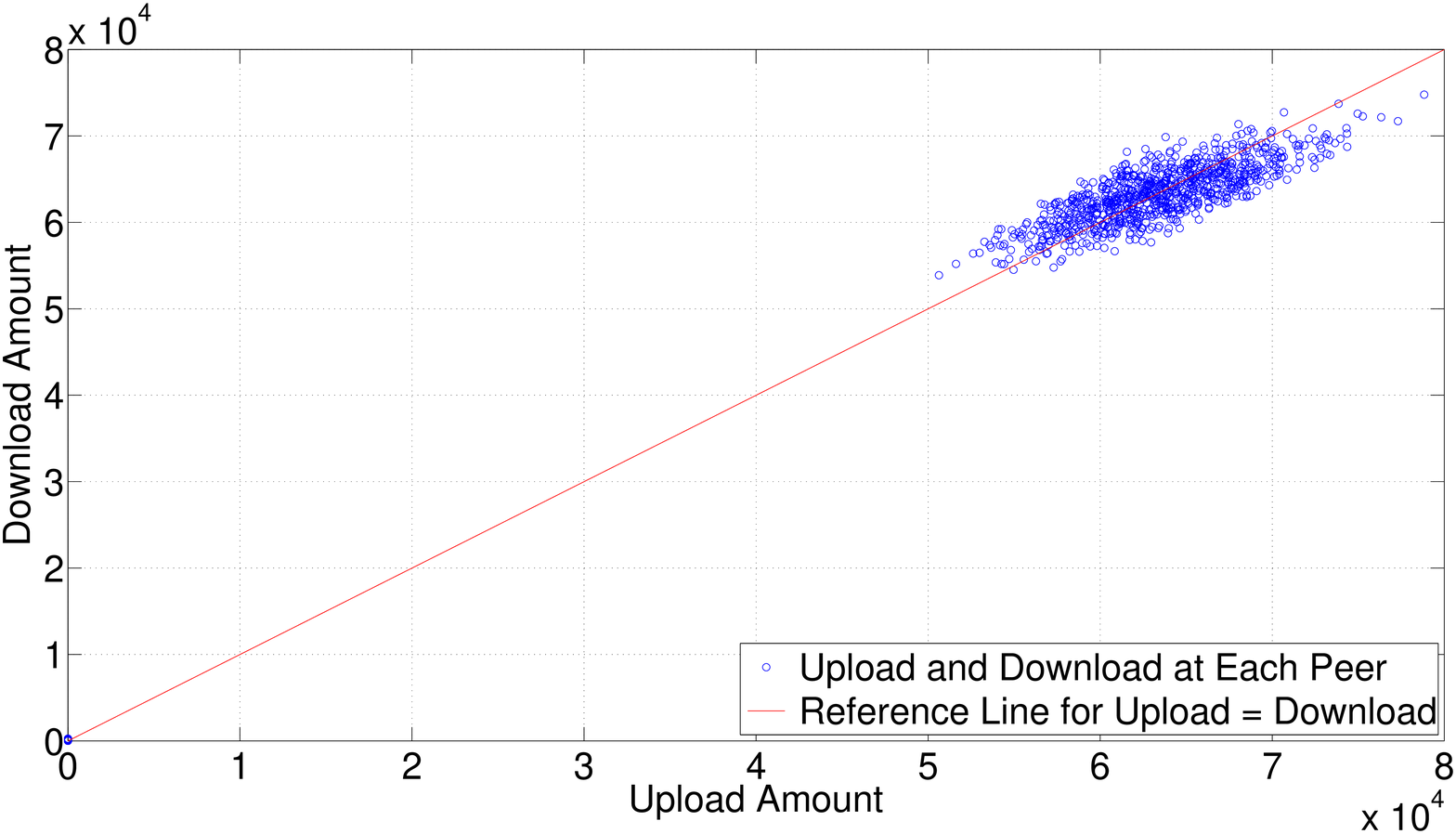}}\\
\subfloat[Free Riders = 30\%, $\alpha=0.9$]{
\includegraphics[width=6cm,height=4cm,scale=.18]
{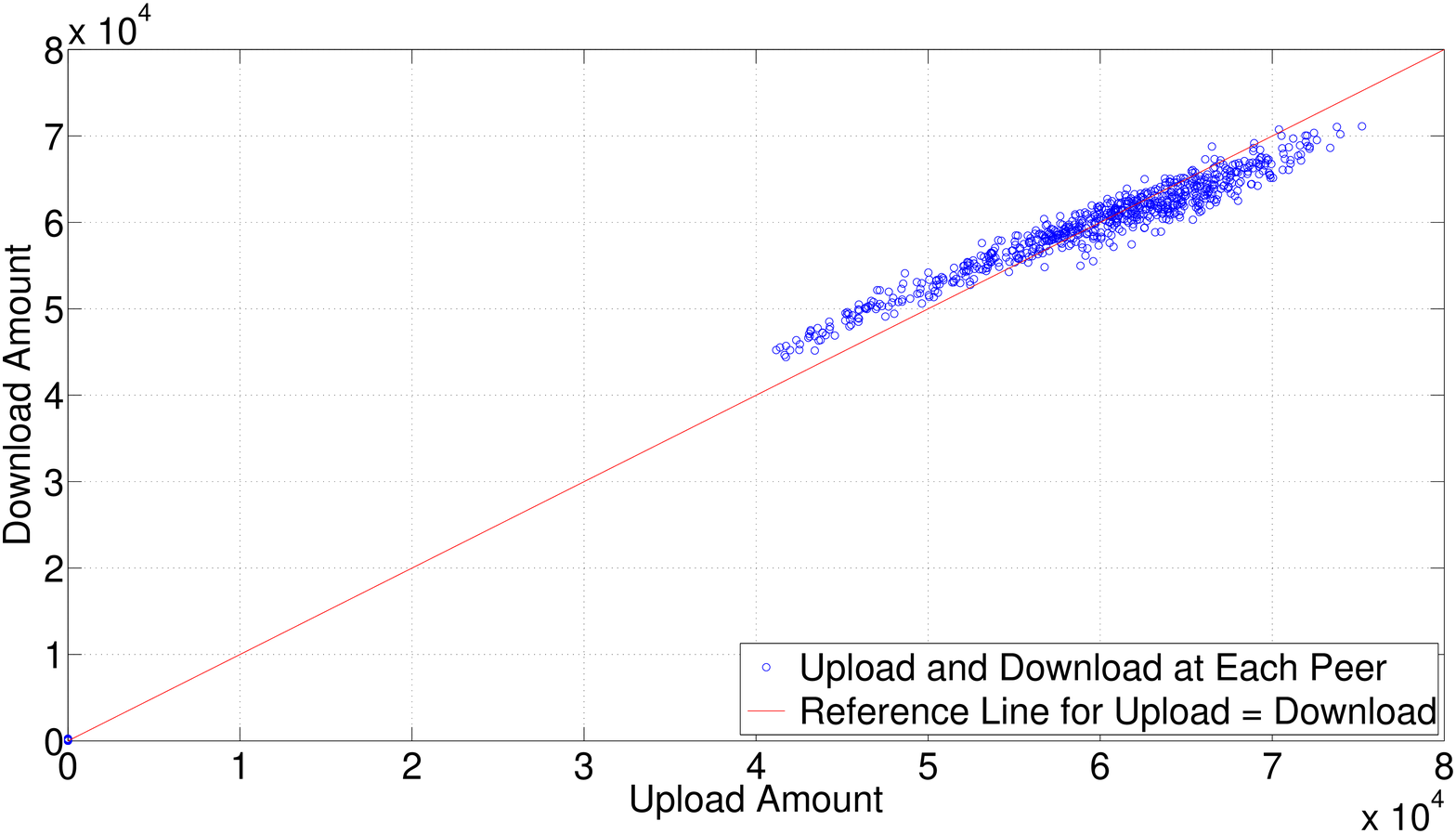} }
\subfloat[Free Riders = 30\%, $\alpha=0.6$]{
\includegraphics[width=6cm,height=4cm,scale=.18]
{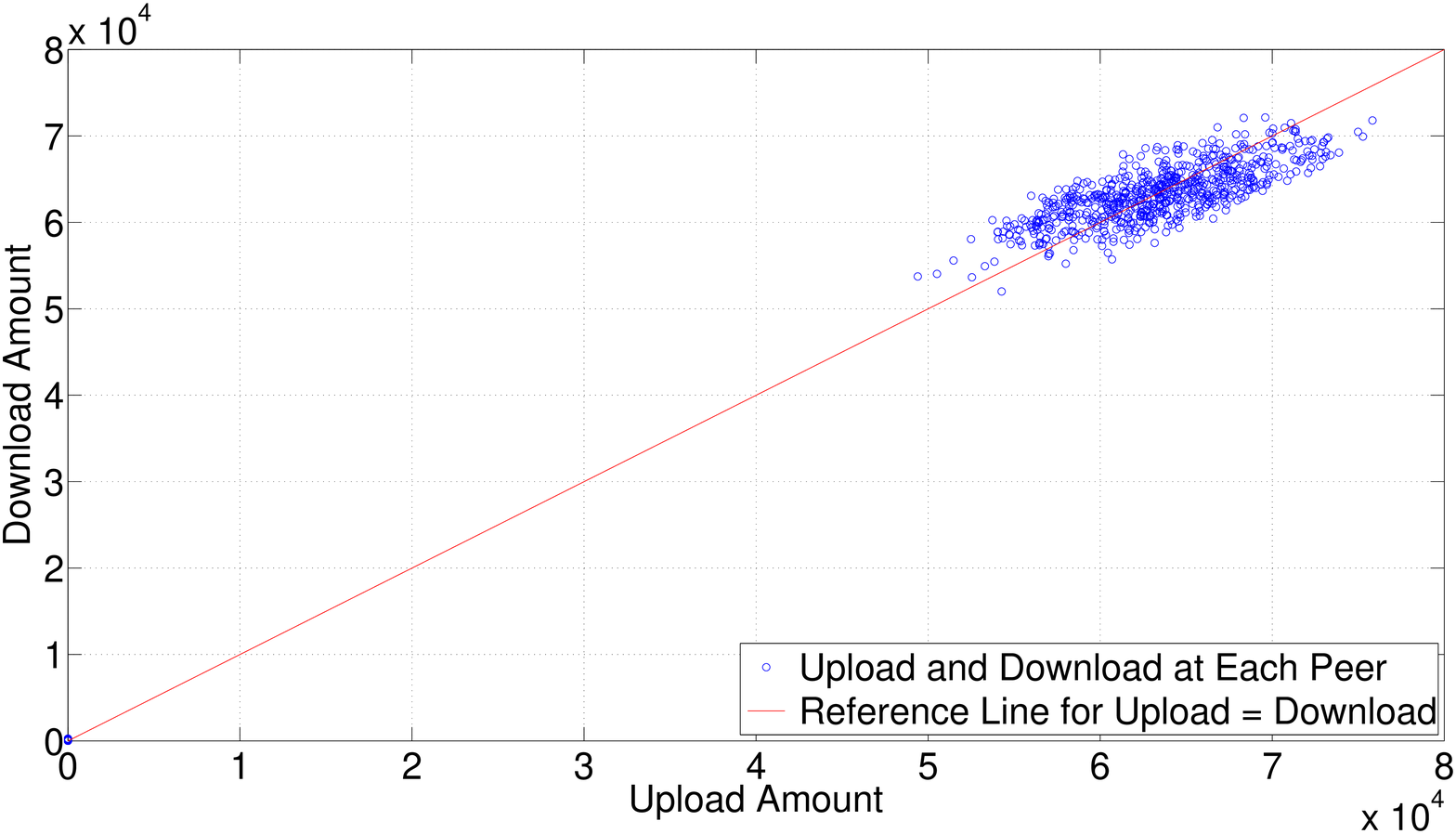}}
\subfloat[Free Riders = 30\%, $\alpha=0.3$]{
\includegraphics[width=6cm,height=4cm,scale=.18]
{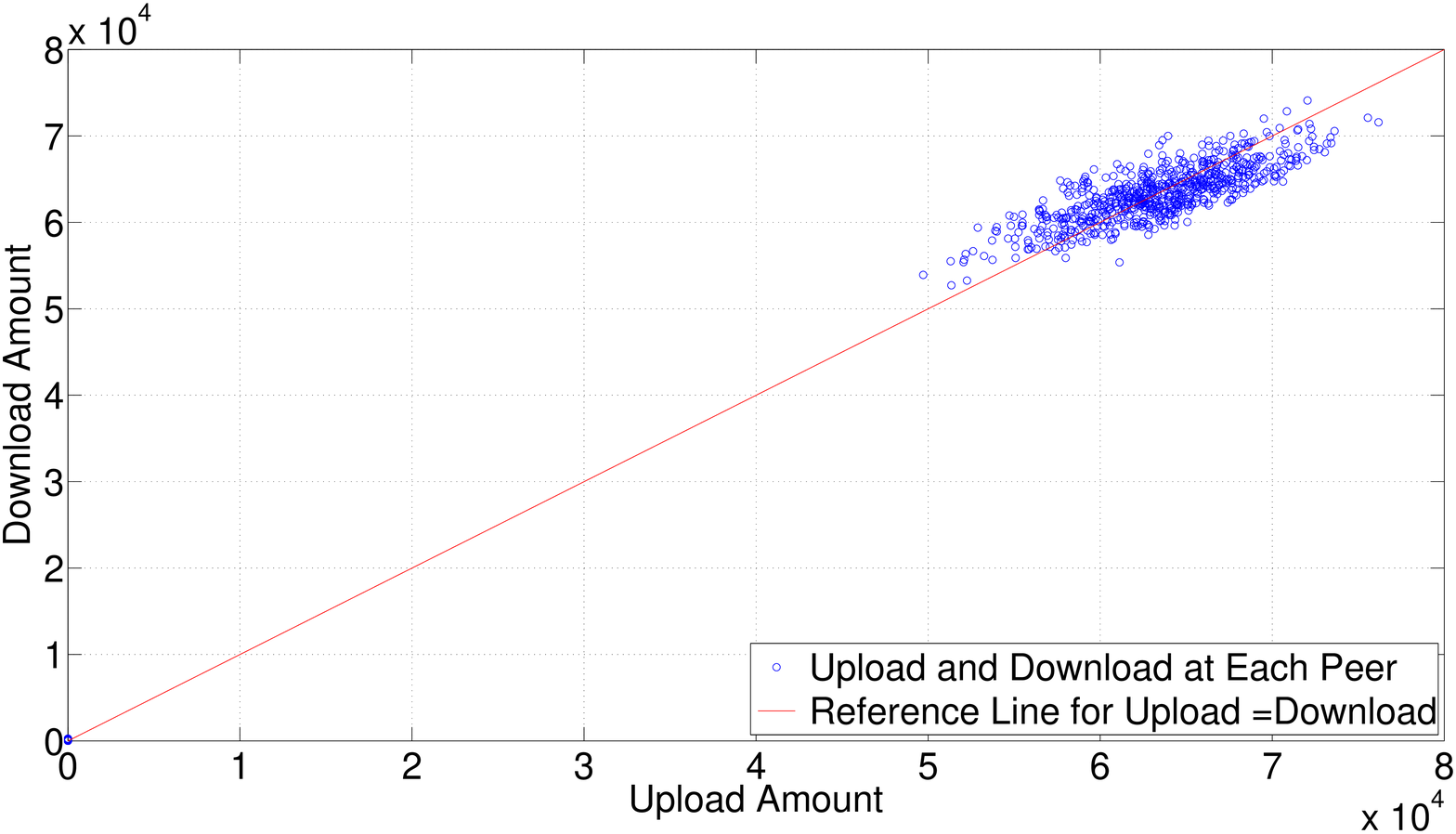} }\\
\subfloat[Free Riders = 50\%, $\alpha=0.9$]{
\includegraphics[width=6cm,height=4cm,scale=.18]
{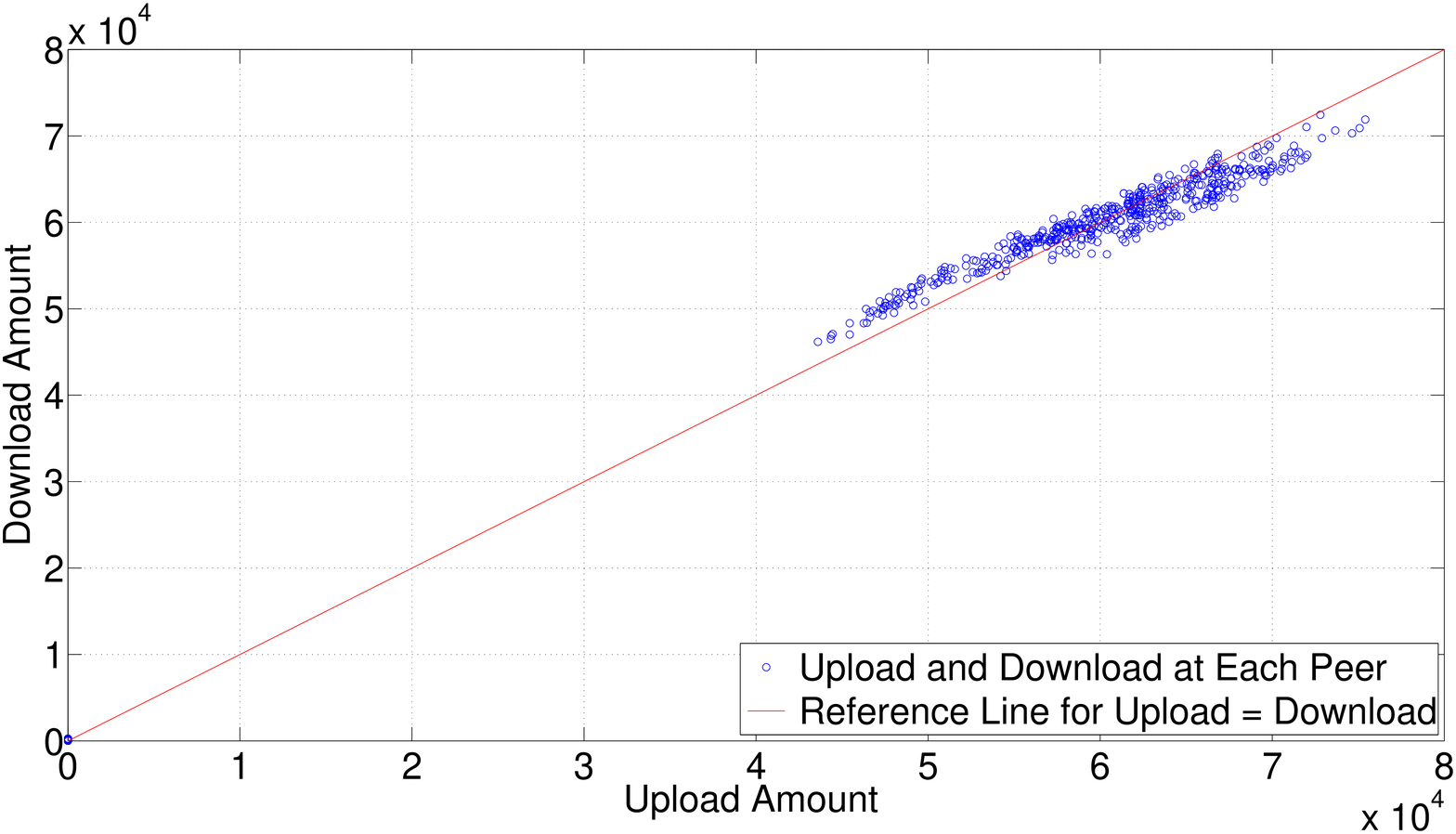}}
\subfloat[Free Riders = 50\%, $\alpha=0.6$]{
\includegraphics[width=6cm,height=4cm,scale=.18]
{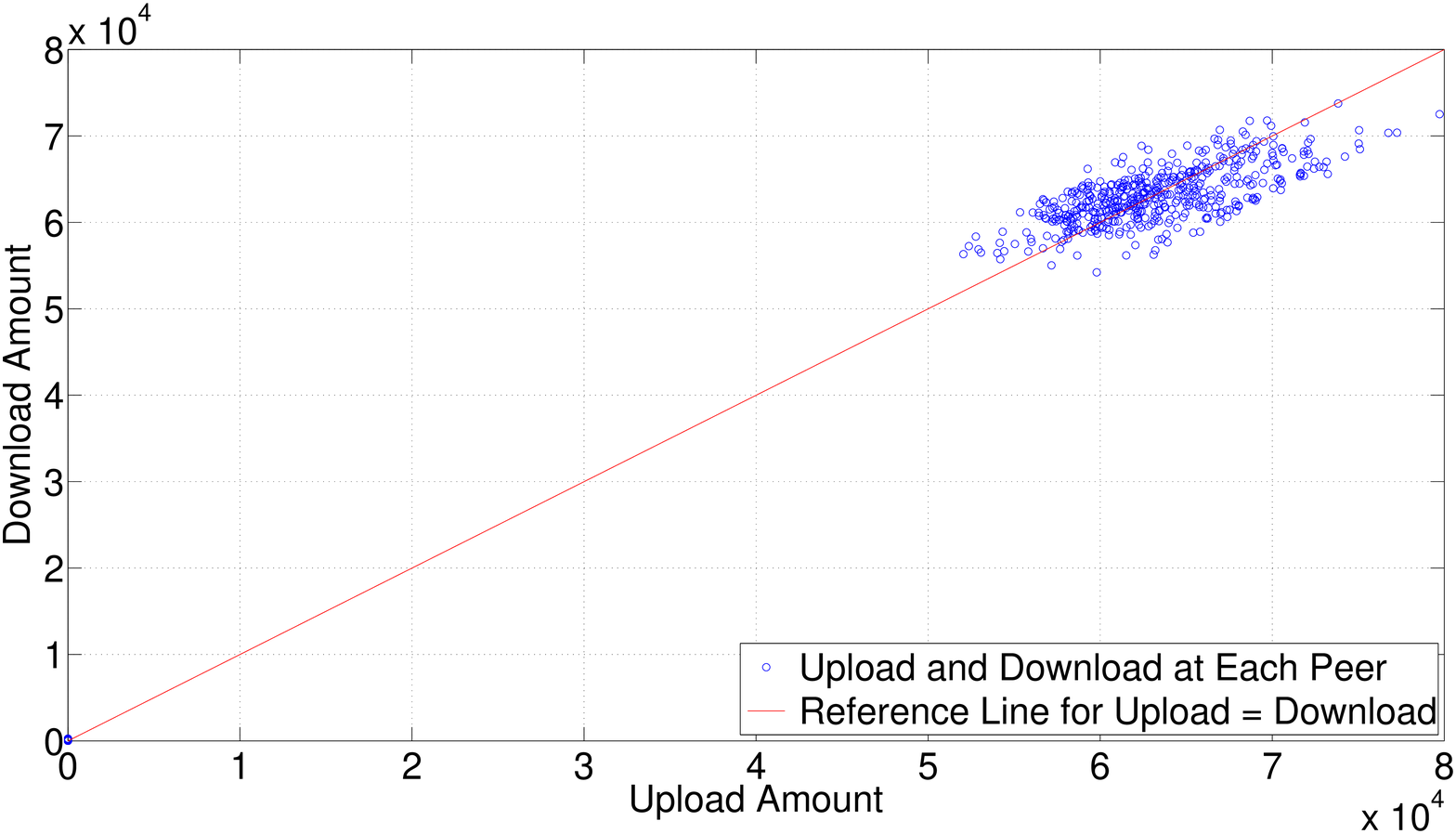}  }
\subfloat[Free Riders = 50\%, $\alpha=0.3$]{
\includegraphics[width=6cm,height=4cm,scale=.18]
{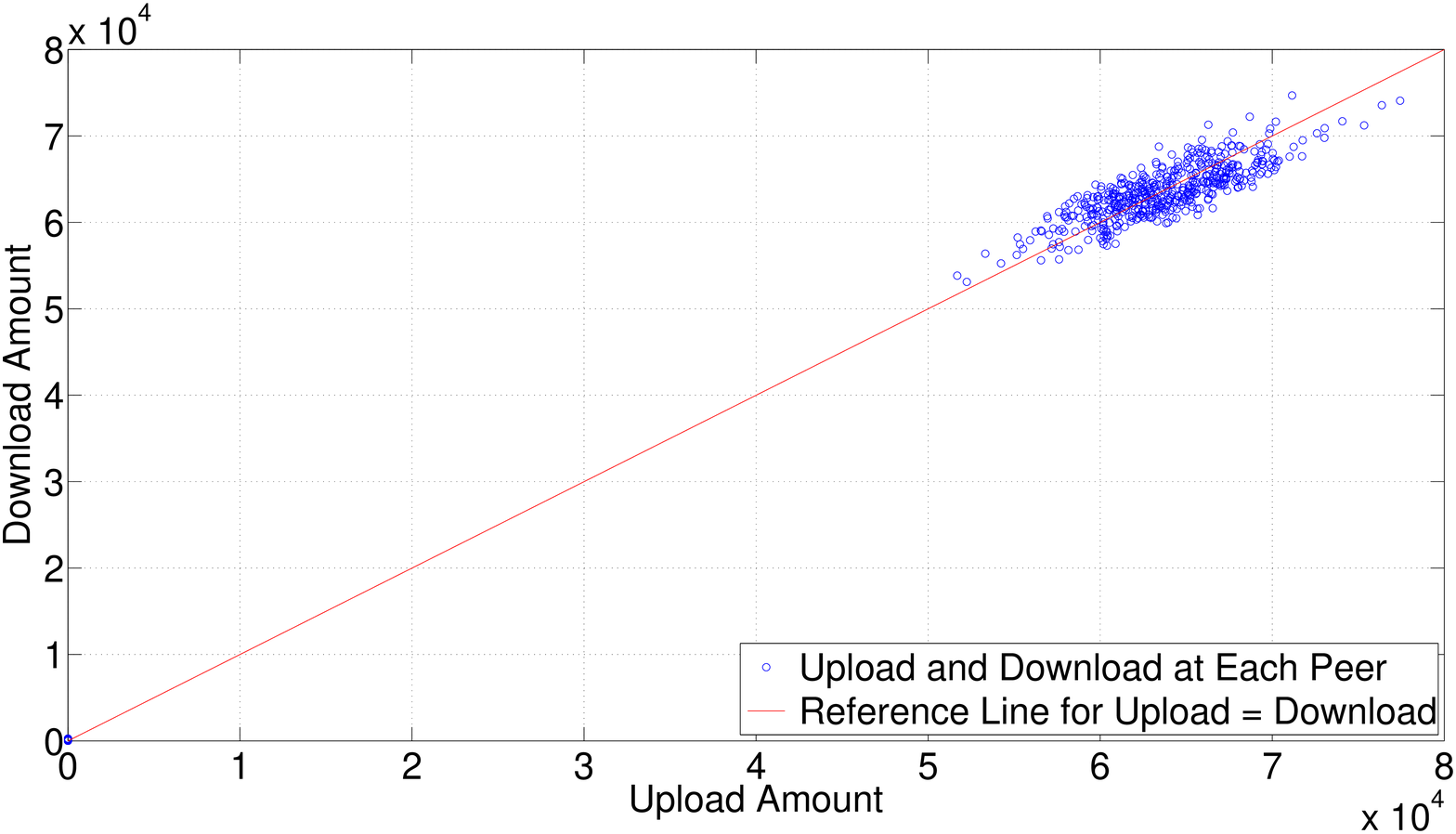}}\\
\subfloat[Free Riders = 70\%, $\alpha=0.9$]{
\includegraphics[width=6cm,height=4cm,scale=.18]
{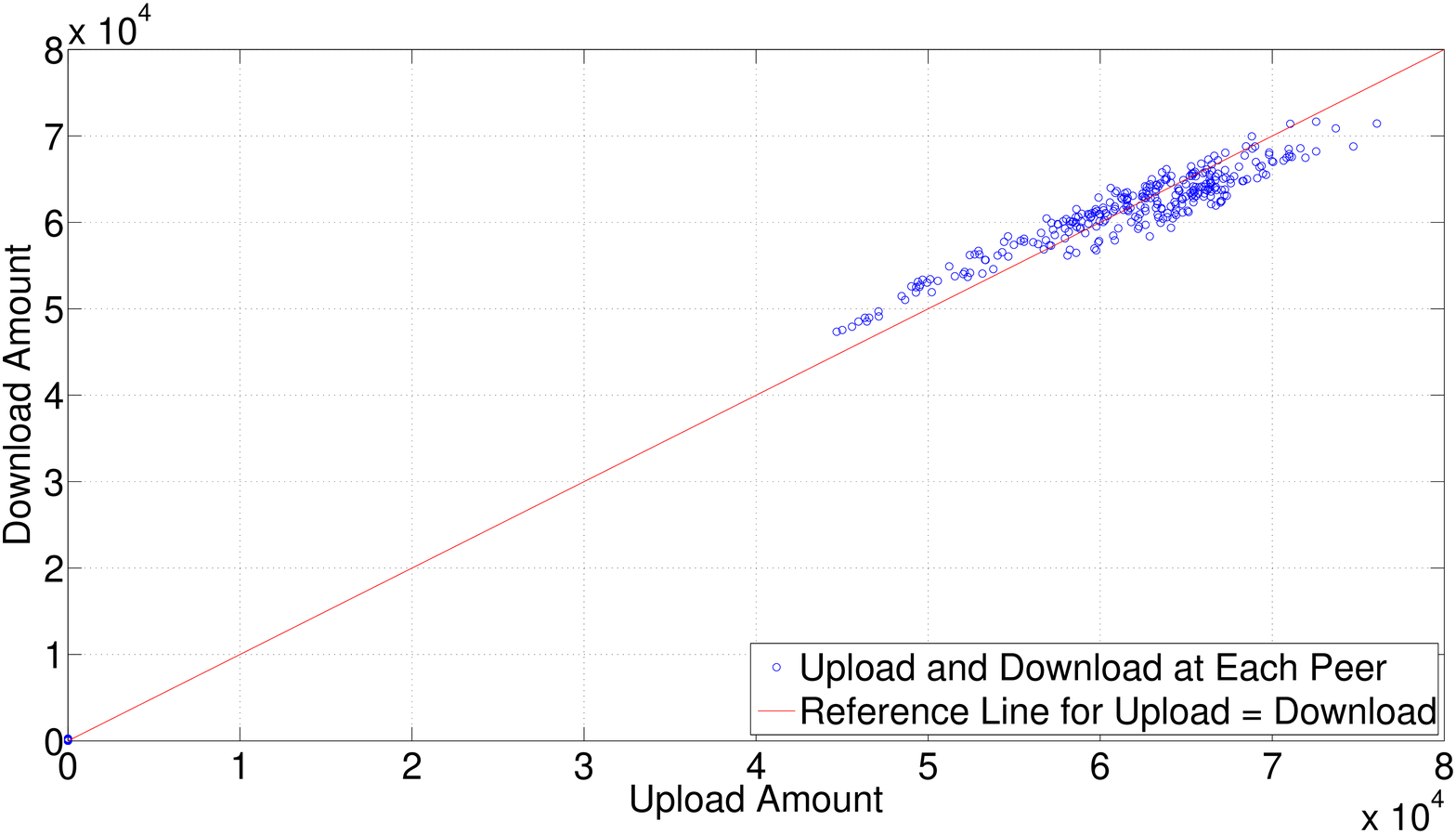} }
\subfloat[Free Riders = 70\%, $\alpha=0.6$]{
\includegraphics[width=6cm,height=4cm,scale=.18]
{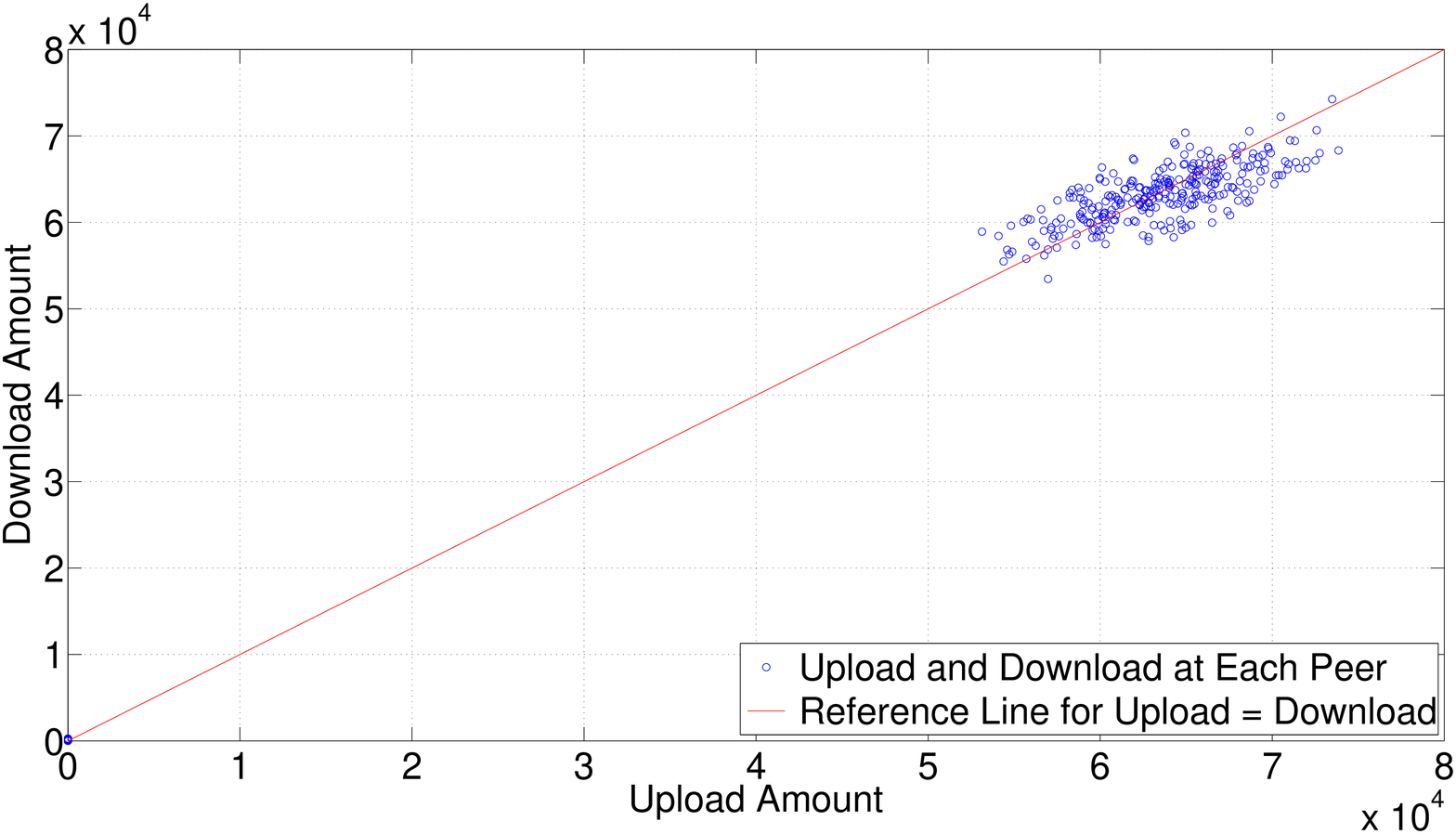}}
\subfloat[Free Riders = 70\%, $\alpha=0.3$]{
\includegraphics[width=6cm,height=4cm,scale=.18]
{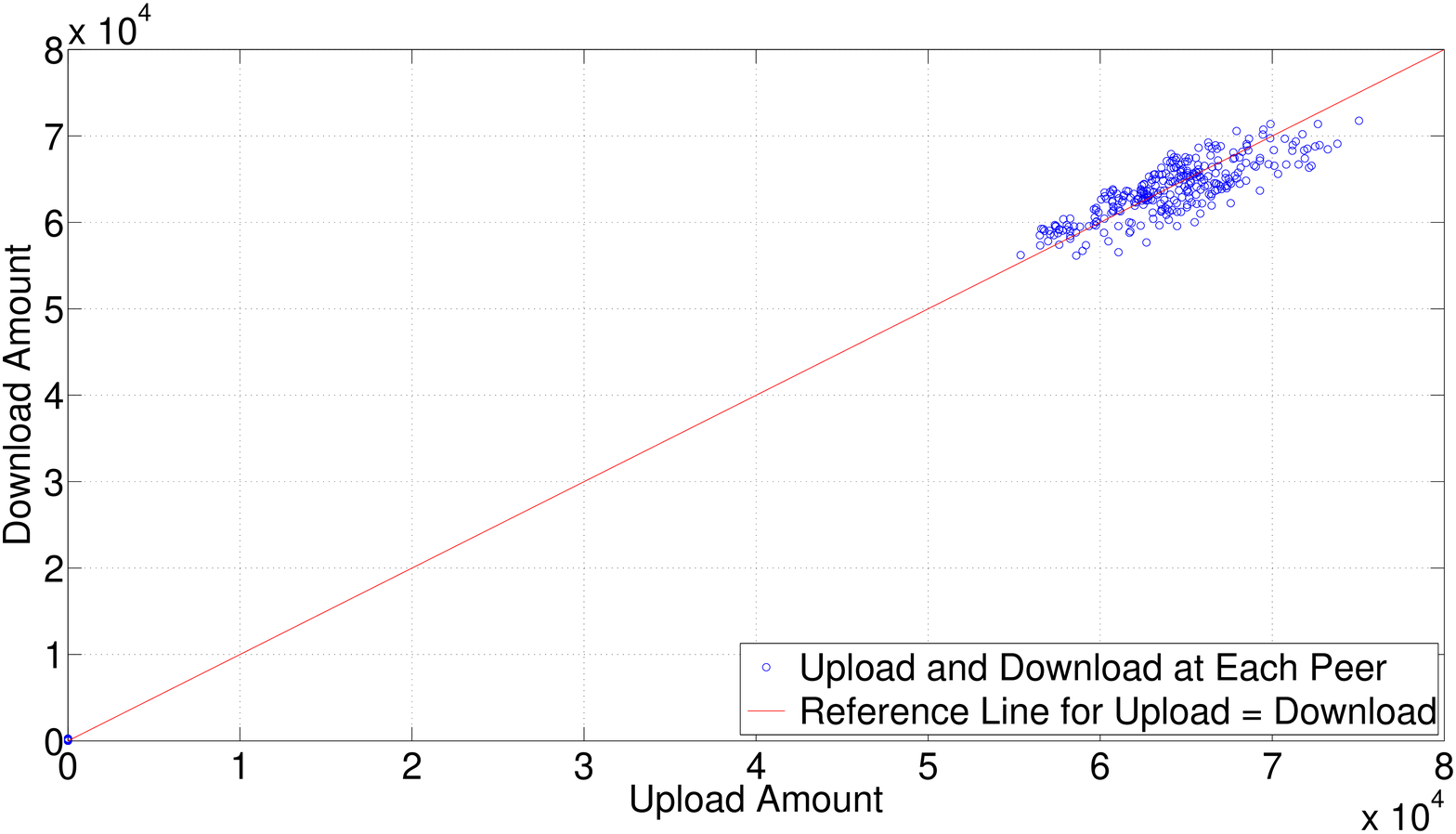} }
\caption{Upload and Download Amount at Each Peer in the Presence of $10\%-70\%$ Free-riders for Different values of $\alpha$. Peer selection approach is based on the problem of "College Admission and The Stability of Marriage". Bandwidth distribution is as type 2.}\label{soseventy}
\end{figure*}

\begin{table}
\begin{center}
\caption{ AAD and $\%$ of Rejections for SBCI in Simple Model for Stable merriage approach with ten different bandwith peers }\label{table4.5}
\begin{tabular}{ | m{0.5cm} | m{1em}| m{5em}|m{5em}| m{5.5em}|} 
 \hline
  S.N.& $\alpha$ &Free-riders & AAD &$\%$ of Rejections \\[1ex] 
 \hline
 $1$& 0.9  & 10\% &  0.131826&6.1534\\
    \hline
 $2$& 0.9  & 30\% &  0.322511&4.2034\\
    \hline
 $3$& 0.9  & 50\% &  0.515015&2.3352\\
    \hline
 $4$& 0.9  & 70\%&  0.709604&0.9604\\
    \hline
 $5$& 0.6  & 10\% & 0.128717&0.4344\\
    \hline
 $6$& 0.6  & 30\% &  0.327149&0.2746\\
    \hline
 $7$& 0.6  & 50\% &  0.520278&0.1396\\
    \hline
 $8$& 0.6  & 70\% &  0.710663&0.0348\\
    \hline
 $9$& 0.3  & 10\% &  0.130575&0.1214\\
  \hline
  $10$& 0.3  & 30\% &  0.323463&0.0652\\
    \hline
 $11$& 0.3  & 50\% &  0.514418&0.024\\
    \hline
 $12$& 0.3  & 70\% &  0.70845&0.0054\\
    \hline
\end{tabular}
\end{center}
\end{table}

\subsection{Simulation Results of College Admission and The Stability of Marriage Based Approach For the Peer Selection}
We also conducted the experiment for college admission and the stability of marriage based approach for the peer selection. For simplicity, we considered only Simple model. Bandwidth of peers is assumed to be different, so that they can also include the bandwidth, as a criteria for peer selection. Selection of peer for downloading and uploading is done according to Algorithm \ref{algo4.3}. The stable match for uploader and downloader is made downloader optimal. To observe the impact of heterogeneity, we simulated the Simple Model for two different types of bandwidth distributions, i.e., type 1 and type 2.  \par 
In type 1, half of the peers have bandwidth 10 units and the rest have 20 units. Simulation results for this type are shown in Fig. \ref{sseventy}. Corresponding $AAD$ and percentage of rejections among cooperative peers are shown in Table \ref{table4.4}. We can see from the figure that upload and download amount  increases in each peer compared to simple procedure. Because each peer, who request for resources, is getting some option for downloading. Uploads and downloads in each peer are close to the reference line and corresponding $AAD$ are lesser compared to simple procedure.  Thus, the algorithm is able to balance the upload and download amount in each peer.\par 
In type 2, 10\% of the peers have bandwidth 10 units, next 10\% of the peers have bandwidth 20 units, next 10\% of peers have bandwidth 30 units and so on. In this way, last 10\% of peers will have bandwidth 100 units. Simulation results for this type are shown in Fig. \ref{soseventy}. Corresponding $AAD$ and percentage of rejections among cooperative peers are shown in Table \ref{table4.5}. We can observe from this figure that upload and download amounts for most of the peers are far from reference line and corresponding $AAD$ are also higher. Thus, the impact of heterogeneity is clearly evident. It also supports the argument that if we will select the source peer according to bandwidth rather than SBCI, we will loose the fairness in the network. 

\section{Conclusion}\label{conclusion}
In this work, we proposed a new algorithm to make the P2P network fair and efficient. The algorithm  ranks the peers based on their simplified biased contribution index (SBCI) which can vary from 0 to 1. Estimation of SBCI is based on two factors, the resources contributed by the peer and the SBCI of peer with whom it is transacting.  We propose the design rules to make the network fair and efficient. With the help of mathematical justification, we have shown that our algorithm  can fulfill all the design objectives and is able to maintain the fairness in the network. This algorithm can be implemented in the truly distributed fashion. Since,  no iterative calculation is needed, it can be implemented with lesser message overhead and storage capacity. \par We proposed two different peer selection approaches, namely simple procedure and college admission and the stability of marriage based approach. Simulation results show that the algorithm is able to suppress the free-riders in highly free-riding environment. The algorithm is also able to suppress the dynamic free-riders, i.e., those who change their behavior dynamically.\par In future, we would like to implement this mechanism in unstructured P2P network.


%



\ifCLASSOPTIONcompsoc


\ifCLASSOPTIONcaptionsoff
  \newpage
\fi



%

\begin{IEEEbiography}
He was born in Uttarkashi, India. He is currently pursuing Ph.D in the Department of Electrical Engineering at  IIT, Kanpur. His research interests include Peer-to-Peer Networks, Wireless Sensor Networks, Complex Networks, Social Networks, Solution of non-linear equations, Application of Linear Algebra and Game theory in Networks.
\end{IEEEbiography}
\begin{IEEEbiography}
He was born in Delhi, India. He was awarded Ph.D for his work on optical amplifier placement problem in all-optical broadcast networks in 1997 by IIT Delhi. In July 1997, he joined EE Department, IIT Kanpur. He was given AICTE young teacher award in 2003. Currently, he is working as professor. He is fellow of IETE, senior member of IEEE and ICEIT, and member ISOC. He has interests in telecommunications' networks specially optical networks, switching systems, mobile communications, distributed software system design. He has supervised 10 Ph.D and more than 125 M.Tech theses so far. He has filed three patents for switch architectures, and have published many journal and conference research publications. He has also written lecture notes on Digital Switching which are distributed as open access content through content repository of IIT Kanpur. He has also been involved in opensource software development. He has started Brihaspati (brihaspati.sourceforge.net) initiative, an opesource learning management system, BrihaspatiSync – a live lecture delivery system over Internet, BGAS – general accounting systems for academic  institutes.
\end{IEEEbiography}

%








\end{document}